\documentclass[a4paper]{amsart} 
\usepackage[utf8]{inputenc}
\usepackage[english]{babel}
\usepackage[T1]{fontenc}

\usepackage{latexsym,amsmath,amsfonts,amssymb,amsthm,amstext,textcomp,mathabx,relsize}
\usepackage{mathrsfs,dutchcal,bbm}
\usepackage{enumerate}
\usepackage{hyperref}

\usepackage[capitalize]{cleveref}

\makeatletter
\def\@seccntformat#1{%
	\protect\textup{\protect\@secnumfont
		\ifnum\pdfstrcmp{subsection}{#1}=0 \bfseries\fi
		\ifnum\pdfstrcmp{subsubsection}{#1}=0 \itshape\fi
		\csname the#1\endcsname
		\protect\@secnumpunct
	}%
}
\renewcommand{\@upn}{}
\DeclareRobustCommand{\crefnosort}[1]{%
	\begingroup\@cref@sortfalse\cref{#1}\endgroup
}
\makeatother

\newcommand{\EE}{\mathbb{E}}
\newcommand{\CC}{\mathbb{C}}

\newcommand{\NN}{\mathbb{N}}
\newcommand{\PP}{\mathbb{P}}

\newcommand{\RR}{\mathbb{R}}

\newcommand{\ii}{\mathrm{i}}
\newcommand{\eul}{\mathrm{e}}

\newcommand{\dist}{\mathrm{dist}}

\newcommand{\loc}{\mathrm{loc}}
\newcommand{\reg}{\mathrm{reg}}

\newcommand{\Id}{\mathrm{d}}

\newcommand{\restr}{\mathord{\upharpoonright}}

\newcommand{\LO}{\mathscr{B}}
\newcommand{\dom}{\mathcal{D}}
\newcommand{\fdom}{\mathcal{Q}}
\newcommand{\HP}{\mathfrak{k}}
\newcommand{\Fock}{\mathcal{F}}
\newcommand{\NO}{N}
\newcommand{\ad}{a^\dagger}
\newcommand{\expv}[1]{\epsilon(#1)}
\newcommand{\Geb}{\Lambda}
\newcommand{\id}{\mathbbm{1}}
\newcommand{\UV}{\kappa}
\newcommand{\bbeta}{\gamma}
\newcommand{\ball}[1]{\mathcal{B}_{#1}}
\newcommand{\ap}{{\dot a}}
\newcommand{\ve}{\varepsilon}
\newcommand{\vp}{\varphi}

\newcommand{\vr}{\varrho}
\newcommand{\vt}{\vartheta}
\newcommand{\vs}{\varsigma}

\newcommand{\wt}[1]{\widetilde{#1}}
\newcommand{\ol}[1]{\overline{#1}}

\newcommand{\mr}[1]{\mathring{#1}} 
\newcommand{\fr}[1]{\mathfrak{#1}}
\newcommand{\mc}[1]{\mathcal{#1}}
\newcommand{\scr}[1]{\mathscr{#1}}
\renewcommand{\le}{\leqslant}
\renewcommand{\ge}{\geqslant}
\theoremstyle{plain}
\newtheorem{thm}{Theorem}[section]

\newtheorem{lem}[thm]{Lemma}
\newtheorem{cor}[thm]{Corollary}
\newtheorem{prop}[thm]{Proposition}
\theoremstyle{definition}
\newtheorem{defn}[thm]{Definition}

\theoremstyle{remark}
\newtheorem{example}[thm]{Example} 
\newtheorem{rem}[thm]{Remark}
\numberwithin{equation}{section}
\crefname{equation}{}{}
\Crefname{equation}{}{}
\crefname{enumi}{}{}
\Crefname{enumi}{}{}
\crefname{lem}{Lemma}{Lemmas}
\Crefname{lem}{Lemma}{Lemmas}
\crefname{thm}{Theorem}{Theorems}
\Crefname{thm}{Theorem}{Theorems}
\crefname{prop}{Proposition}{Propositions}
\Crefname{prop}{Proposition}{Propositions}
\crefname{cor}{Corollary}{Corollaries}
\Crefname{cor}{Corollary}{Corollaries}
\crefname{defn}{Definition}{Definitions}
\Crefname{defn}{Definition}{Definitions}

\title[Feynman--Kac formulas for multi-polaron semigroups]{Feynman--Kac formulas for semigroups generated by
multi-polaron Hamiltonians in magnetic fields and on general domains}
\author{Benjamin Hinrichs}
\address{Benjamin Hinrichs, Universit\"at Paderborn, Institut f\"ur Mathematik, Institut f\"ur Photonische Quantensysteme, Warburger Str. 100, 33098 Paderborn, Germany}
\email{benjamin.hinrichs@math.upb.de}

\author{Oliver Matte}
\address{Oliver Matte, Aalborg Universitet, Institut for Matematiske Fag, Skjernvej 4a, 9220 Aalborg, Denmark}
\email{oliver@math.aau.dk}
\usepackage{xcolor}\usepackage{cancel}

\begin{document}

\begin{abstract} 
	\noindent 
	We prove Feynman--Kac formulas for the semigroups generated by selfadjoint operators in a class containing
	Fr\"ohlich Hamiltonians known from solid state physics. The latter model multi-polarons, i.e., 
	a fixed number of quantum mechanical
	electrons moving in a polarizable crystal and interacting with the quantized phonon field generated by the crystal's
	vibrational modes. Both the electrons and phonons can be confined to suitable open subsets of Euclidean space.
	We also include possibly very singular magnetic vector potentials and electrostatic potentials. 
	Our Feynman--Kac formulas comprise Fock space operator-valued multiplicative functionals 
	and can be applied to every vector in the underlying Hilbert space.
	In comparison to the renormalized Nelson model, for which analogous Feynman--Kac formulas are known, 
	the analysis of the creation and annihilation terms in the multiplicative functionals requires novel ideas to overcome difficulties
	caused by the phonon dispersion relation being constant.
	Getting these terms under control and generalizing other construction steps so as to cover confined systems
	are the main achievements of this article.
\end{abstract}

\maketitle

\section{Introduction and main results}

\subsection{General introduction}

When electrons move in a crystal lattice comprised of oppositely charged ions they create lattice distortions (phonons)
in their neighbourhoods, which  back-react on the electrons via the polarization they carry. 
This results in each electron being accompanied by a cloud of phonons lowering its mobility. 
Such a composite object is called a polaron;
when several electrons are considered we speak of multi-polarons.
In \cite{Frohlich.1954}, H.~Fr\"ohlich introduced a Hamiltonian governing the dynamics of multi-polarons.
In his model the electrons are treated as non-relativistic quantum mechanical particles without spin degrees of freedom
whereas the phonons, which can be created and annihilated along the time evolution, are described by a 
non-relativistic bosonic quantum field.

Starting with the seminal work of Feynman \cite{Feynman.1955}, one main technique in the investigation of polaron models 
has been functional integration, both in theoretical physics and mathematics.
Shortly, in \cref{ssec:prevwork}, we shall give numerous references to mathematical papers exploiting various
Feynman--Kac formulas for vacuum expectation values of members of the semigroup generated by Fr\"ohlich's Hamiltonian.

Building on recent mathematical studies of Feynman--Kac formulas in non- and semi-relativistic quantum field theory \cite{GueneysuMatteMoller.2017,MatteMoller.2018,Matte.2021,HinrichsMatte.2022}, we devote this article to the derivation of Feynman--Kac formulas in Fr\"ohlich's multi-polaron model
for semigroup members applied to arbitrary vectors in the underlying Hilbert space.
Since electrons interact via repulsive Coulomb potentials and 
polarons exposed to external electric and magnetic fields are often treated -- see \cite{AnapolitanosGriesemer.2014,Ghanta.2021,GriesemerWellig.2013,Loewen.1988} for mathematical results
on polarons in magnetic fields -- we shall in fact work under almost optimal conditions on the electrostatic
potential and optimal conditions on the magnetic vector potential still permitting to define the Hamiltonian 
via semibounded quadratic forms. In some articles, the electrons are confined to open regions of Euclidean space 
\cite{AnapolitanosLandon.2013,FrankLiebSeiringerThomas.2011}, for technical reasons at least,
and sometimes both the electrons and the phonons are confined \cite{FrankSeiringer.2021,BrooksMitrouskas.2023}. 
Therefore, we shall work under general hypotheses on the electron-phonon interaction covering 
the latter two situations as well as the original Fr\"{o}hlich model.

Together with the inequalities established in this article, our Feynman--Kac formulas can form the basis for further studies
of the semigroup and ground state eigenvectors (if any) in polaron models in analogy to the theory of
magnetic Schr\"{o}dinger semigroups \cite{BroderixHundertmarkLeschke.2000,Simon.1982} 
and its extensions to the related Pauli--Fierz model of non-relativistic quantum electrodynamics \cite{Matte.2016}
and Nelson's model for nucleon-meson interactions \cite{MatteMoller.2018,HiroshimaMatte.2019}.

\subsection{Brief description of the main result}

\noindent
The Hamiltonian studied in this article and denoted $H(v)$ acts in the Hilbert space $L^2(\Geb,\Fock)$
where $\Geb\subset\RR^d$ is open and non-empty, $d\ge2$ and $\Fock$ is the bosonic Fock space modeled over
the separable Hilbert space $\HP=L^2(\mc{K},\fr{K},\mu)$ for one phonon. The operator $H(v)$ is a selfadjoint realization 
via quadratic forms of the heuristic expression
\begin{align}\label{Hheur}
	\frac{1}{2}(-\ii\nabla_x-A(x))^2+V(x)+N+\int_{\mc{K}}(v(x,k)\ad(k)+\ol{v}(x,k)a(k))\Id\mu(k).
\end{align}
Here the nabla-operator acts on the position variables $x\in\Geb$ of the electron(s). 
The vector potential $A:\Geb\to\RR^d$ is merely assumed to be locally square-integrable. 
The electrostatic potential $V:\Geb\to\RR$
has a locally integrable positive part and its negative part has an extension to $\RR^d$
belonging to the $d$-dimensional Kato class.
Further, $N$ is the phonon number operator and in \eqref{Hheur} we use, for presentational purposes,
physics notation for the pointwise creation and annihilation operators $\ad(k)$ and $a(k)$, respectively, for each $k\in\mc{K}$. 
Finally, in applications to polarons,
 $v(\cdot,k):\Geb\to\CC$ is a proper or generalized eigenfunction of the Dirichlet Laplacian on $\Geb$ for every $k\in\mc{K}$;
 when multiple polarons are treated, it is a suitable combination of possibly generalized eigenfunctions.
 Then the measure space $(\mc{K},\fr{K},\mu)$ is given in terms of some spectral decomposition of the appropriate
 Dirichlet Laplacian.
Canonical mathematical interpretations of all contributions to \eqref{Hheur} and the Hamiltonian $H(v)$
itself will be introduced carefully in \cref{sec:H(v)}.

Under a natural assumption on the probability of Brownian motion moving large distances inside $\Geb$, precisely stated in \cref{tailbound} and for example satisfied for any convex and open $\Geb$, our main result \cref{thm:mainFK}
is a Feynman--Kac formula for the semigroup generated by $H(v)$ of the form
\begin{align}\label{FKintro}
	(\eul^{-tH(v)}\Psi)(x)&=\EE\big[\chi_{\{t<\tau_{\Geb}(x)\}}
	\eul^{-\ol{S}_t(x)}W_{t}(x)^*\Psi(b^x_t)\big],\quad \text{a.e. $x\in\Geb$},
\end{align}
for all $\Psi\in L^2(\Geb,\Fock)$ and $t\ge0$. 
Here $b=(b_t)_{t\ge0}$ is a $d$-dimensional Brownian motion, $b_t^x\coloneq x+b_t$, and
\begin{align}\label{defexitGeb}
\tau_{\Geb}(x)&\coloneq\inf\{t\ge0|\,b_t^x\in\Geb^c\}
\end{align}
is the first exit time of $b^x$ from $\Geb$. Further,
$\ol{S}_t(x)$ contains a path integral of $V$ along $b^x$ and a suitably generalized Stratonovich integral
of $A(b^x)$ with respect to $b$. Finally, the Fock space operator-valued random variable $W_{t}(x)$ is
explicitly given in terms of a generalization $u_{t}(x)$ of Feynman's complex action  \cite{Feynman.1955}
and two stochastic processes $(U^{\pm}_{t}(x))_{t\ge0}$ attaining
values in the one-phonon Hilbert space $\HP$.
More precisely, it is given by the expression
\begin{align*}
	W_{t}(x)
	&=\eul^{u_{t}(x)}\bigg(\sum_{n=0}^\infty
	\frac{(-1)^n}{n!}\ad(U^{+}_{t}(x))^n\eul^{-tN/2}\bigg)\bigg(\sum_{n=0}^\infty
	\frac{(-1)^n}{n!}\ad(U^{-}_{t}(x))^n\eul^{-tN/2}\bigg)^*.
\end{align*}
Here $\ad(f)$ with $f\in\HP$ is a ``smeared'' creation operator and the two series converge in Fock space operator norm.
Notice that $W_{t}(x)$ is formally normal ordered. In particular, 
$\langle\expv{0}|W_{t}(x)\expv{0}\rangle_{\Fock}=\eul^{u_t(x)}$
with $\expv{0}$ denoting the vacuum vector in $\Fock$, so that \cref{FKintro} implies
\begin{align}\label{FKvacuumRR}
\langle f_1\expv{0}|\eul^{-tH(v)}f_2\expv{0}\rangle&=
\int_{\Geb}\EE\big[\chi_{\{t<\tau_{\Geb}(x)\}}\eul^{u_{t}(x)-\ol{S}_t(x)}\ol{f_1}(x)f_2(b_t^x)\big]\Id x,
\end{align}
for all $f_1,f_2\in L^2(\Geb)$ and $t\ge0$.

\subsection{Remarks on closely related previous work}\label{ssec:prevwork}

The idea to write $W_t(x)$ in the above form stems from \cite{GueneysuMatteMoller.2017}. 
The Feynman--Kac formulas derived in \cite{GueneysuMatteMoller.2017} for $\Geb=\RR^d$ and $A=0$
apply to a class of models containing the Pauli--Fierz model, Nelson's model and the polaron model
provided that ultraviolet regularizations are introduced in the particle-field interaction terms in all these models.
In fact, spin degrees of freedom are allowed for in \cite{GueneysuMatteMoller.2017} as well, which lead to more complicated
expressions for $W_t(x)$.
The proper Nelson model, where the artificial regularizations can be removed by an energy renormalization \cite{Nelson.1964},
has been covered subsequently in \cite{MatteMoller.2018} for $\Geb=\RR^d$ and $A=0$; 
a relativistic version of Nelson's model in two spatial dimensions is treated in \cite{HinrichsMatte.2022,HinrichsMatte.2023}. 
An overview over other types of Feynman--Kac formulas for semigroups in ultraviolet regular 
quantum field theoretic models and over their applications can be found in the textbook \cite{HiroshimaLorinczi.2020}; 
see also \cite{BetzSpohn.2005} for the ultraviolet regularized polaron model.

The mathematical analysis of the interaction term involving $v$ in \eqref{Hheur} requires
some care as well, since $v(x,\cdot)$ is not square-integrable over $\mc{K}$ in physically relevant
applications. For instance, the interaction can directly be introduced
as an infinitesimal form perturbation \cite{LiebYamazaki.1958}; see also Theorem~\ref{thm:LY} below
which covers general open subsets $\Geb\subset\RR^d$ and $A\in L_{\loc}^2(\Geb,\RR^{d})$.
Furthermore, Nelson's operator theoretic renormalization procedure \cite{Nelson.1964},
where a sequence of ultraviolet cutoffs going to infinity is considered,
can be adapted to construct polaron Hamiltonians; the articles \cite{GriesemerWuensch.2016} and 
\cite{FrankSeiringer.2021} elaborate on this approach in the
case $A=0$ for $\Geb=\RR^d$ and certain bounded $\Geb\subset\RR^d$, respectively. 
Finally, the more recently developed method of interior boundary 
conditions applies to the polaron model \cite{LampartSchmidt.2019,Posilicano.2020} and yields formulas
for the domain of $H(v)$ and its action on it, at least when 
$\Geb=\RR^d$, $A=0$ and $V$ is slightly more regular.

The Feynman--Kac formula \cref{FKvacuumRR} for matrix elements of the semigroup with respect
to vectors of the form $f_i\expv{0}$ is actually well-known for the Fr\"{o}hlich
multi-polaron Hamiltonian with phonons living on the whole $\mc K = \RR^3$, in the case $A=0$ at least.
In fact, according to known results, both sides of \cref{FKvacuumRR} can be approximated 
by their ultraviolet regularized analogues in this situation, whence
it suffices to have Feynman--Kac formulas for the semigroups of
polaron Hamiltonians with ultraviolet cutoffs. In \cref{exuFeyn} we recall
Feynman's famous expression for $u_t(x)$ in the multi-polaron model on $\RR^3$ \cite{Feynman.1955}
and how it can be obtained as a limit of ultraviolet regularized complex actions.
Suitable bounds on the exponential moments $\EE[\eul^{pu_t(x)}]$, $p>0$, of Feynman's complex action needed to
establish \cref{FKvacuumRR} follow from \cite{DonskerVaradhan.1983,BleyThomas.2015,Bley.2016}.
The same reasoning applies to fiber Hamiltonians in the translation invariant case, i.e., when $\Geb=\mc K=\RR^3$,
$A=0$, $V=0$, and corresponding analogues of \cref{FKvacuumRR} are
well-known as well.
In fact, formulas of type \cref{FKvacuumRR} and their relatives for fiber Hamiltonians have been exploited in numerous
mathematical works on the polaron model addressing properties of minimal energies, the mass shell,
the renormalized mass and related polaron path measures \cite{AnapolitanosLandon.2013,BazaesMukherjeeSellkeVaradhan.2023,BetzPolzer.2022,BetzPolzer.2023,Bley.2016,BleyThomas.2015,DonskerVaradhan.1983,DybalskiSpohn.2020,FrankLiebSeiringerThomas.2011,MukherjeeVaradhan.2020,MukherjeeVaradhan.2020b,Polzer.2023,Spohn.1987}.

As a final remark we mention that the optimal condition $A\in L_\loc^2(\Geb,\RR^d)$ is known 
to be sufficient for obtaining Feynman--Kac formulas for magnetic Schr\"{o}dinger operators defined
by forms since \cite{Hundertmark.1996}.
The technical implementations adopted here are different and have been applied
to the Pauli--Fierz model in \cite{Matte.2021}. 

\subsection{Remarks on mathematical novelties}

In view of the above discussion, 
the first notable novel aspect of \eqref{FKintro} is that no ultraviolet regularization is required any longer in a
Feynman--Kac formula for the semigroup in a polaron type model that can be applied to {\em every} vector $\Psi$ 
in the Hilbert space.
Actually, at least when $\Geb=\RR^d$, formulas for $U^{\pm}_{t}(x)$ without regularizations can easily be deduced
by mimicking a procedure in \cite{MatteMoller.2018}. Onwards, a technical issue shows up, however:

In the polaron model the bosons have the constant dispersion relation $1$, which in Nelson's model is substituted
by the relativistic expression $\omega(k)=(|k|^2+m^2)^{1/2}$, $k\in\RR^3$, for some $m\ge0$.
The fact that $\omega(k)$ grows linearly in $|k|$ actually is helpful in the discussion of the analogues of
$U^{\pm}_{t}(x)$ in Nelson's model. As a consequence, the derivations of some crucial estimates on 
$U^{\pm}_{t}(x)$ in \cite{MatteMoller.2018} break down and replacements are in need for our treatment of the polaron model
(see \cref{sec:Upm}).

A second non-obvious observation made here is that the 
procedures of \cite{MatteMoller.2018} can be abstracted and pushed forward so as to cover confined bosons.
For instance, we shall obtain formulas for the complex action $u_t(x)$ similar to the ones in 
\cite{MatteMoller.2018} that are useful in our general setting to derive 
$x$-uniform exponential moment bounds on $u_t(x)$, whose right hand sides are log-linear in $t$,
and $x$-uniform convergence relations for sequences of exponentials of complex actions.

Also, in the treatment of arbitrary open regions $\Geb$, we need to make use of a large deviation type estimate for Brownian motion.
This was unnecessary in previous articles due to the choice $\Geb=\RR^d$.
Bounds similar to our assumption \cref{tailbound} were for example used in the study of Schr\"odinger operators by probabilistic methods in \cite{McGillivrayStollmannStolz.1995}.

\subsection*{Organization of the article and some notation}

The remainder of the text comprises six sections (\textsection2--\textsection7) and four appendices (A--D):
\begin{enumerate}
\item[\textsection2:] We explain all standing assumptions on $A$, $V$, $v$ and an ultraviolet regular
coupling function $\vt$ and present detailed constructions of $H(v)$ and $H(\vt)$.
\item[\textsection3:] All processes appearing in our Feynman--Kac formulas are introduced in detail and
our main theorems are stated.
\item[\textsection4:] We prove a Feynman--Kac formula for $H(\vt)$
under additional regularity assumptions on $A$ and $V$, pushing results of \cite{GueneysuMatteMoller.2017}
forward to non-zero $A$ and proper subsets $\Geb\subset\RR^d$.
\item[\textsection5:] We derive formulas for $U^{\pm}_{t}(x)$ that stay meaningful when 
ultraviolet regularizations are dropped, and use these to prove convergence relations and
$(t,x)$-uniform exponential moment bounds on $(1+1/t)\|U^{\pm}_{t}(x)\|_{\HP}^2$.
\item[\textsection6:] We prove the aforementioned results on the complex action $u_t(x)$.
\item[\textsection7:] 
We discuss the probabilistic sides of our Feynman–Kac formulas 
considered as bounded operators from $L^p(\Geb,\Fock)$ to $L^q(\Geb,\Fock)$, $1<p\le q\le\infty$.
We derive convergence theorems for these operators
and complete the proof of our Feynman--Kac formulas in a series of approximation steps.
\item[A:]
We derive a relative form bound on the electron-phonon interaction in the spirit of \cite{LiebYamazaki.1958},
allowing for non-zero $A$ and proper subsets $\Geb\subset\RR^d$.
\item[B:]
Magnetic Schr\"odinger operators depend continuously in the strong resolvent sense on the vector potential
with respect to the topology on $L_{\loc}^2(\Geb,\RR^d)$, \cite{LiskevichManavi.1997}. We generalize this result to 
polaron Hamiltonians.
\item[C:]
Differentiability properties of $\HP$-valued functions related to $v$ are discussed.
\item[D:]
For the reader's convenience we explain how Feynman's expression for the complex action in \cite{Feynman.1955}
and its direct analogues for suitable confined systems are related to our formulas for $u_t(x)$.
\end{enumerate}

Let us mention right away that $u_{t}(x)$ and $U^{\pm}_{t}(x)$ depend on an additional technical parameter
$\sigma$ in the remaining part of the text. Changing $\sigma$ will, however, alter these processes only up to indistinguishability.

For clarity we finally recall some standard notation used throughout the text:
\begin{itemize}
	\item We write $a\wedge b\coloneq\min\{a,b\}$ and $a\vee b\coloneq\max\{a,b\}$ for all $a,b\in\RR$.
	\item The characteristic function of a set $M$ is denoted by $\chi_M$.
	\item $\dom(\cdot)$ denotes domains of definition; $\fdom(\cdot)$ denotes form domains of 
	semibounded selfadjoint operators.
	\item $\LO(X)$ is the space of bounded operators on a normed vector space $X$.
	\item For any normed vector space $X$, we let $C_b(\RR^n,X)$ denote set of bounded continuous functions
	from $\RR^n$ to $X$. Likewise, $C_b^1(\RR^n,X)$ is the set of bounded, continuously differentiable functions
	from $\RR^n$ to $X$ whose derivatives are bounded as well.
\end{itemize}


\section{Standing assumptions and construction of polaron Hamiltonians}\label{sec:H(v)}

\noindent
In the following three subsections we shall, respectively, introduce the necessary elements of bosonic Fock space calculus,
explain the hypotheses on our model and discuss the Hamiltonians $H(v)$ and $H(\vt)$.

\subsection{Fock space calculus}\label{ssec:Fock}

Let us briefly introduce the relevant objects from bosonic Fock space theory and recall some of their well-known properties. 
For a textbook introduction with the same approach see \cite{Parthasarathy.1992}.

We always assume that $(\mc{K},\fr{K},\mu)$ is a $\sigma$-finite measure space with the property that
the corresponding Hilbert space
\begin{align*}
	\HP&\coloneq L^2(\mc{K},\fr{K},\mu)
\end{align*}
is separable; $\HP$ will be the state space for a single boson.
The bosonic Fock space $\Fock$ modeled over $\HP$ is then given by
\begin{align}\label{defFock}
	\Fock&\coloneq \bigoplus_{n=0}^\infty\Fock_n .
\end{align}
Here $\Fock_0\coloneq \CC$ and
$\Fock_n$ with $n\in\NN$ is the closed subspace comprised of 
all $\psi_n\in L^2(\mc{K}^{n},\otimes_{i=1}^n\fr{K},\otimes_{i=1}^n\mu)$ 
that are permutation symmetric in the sense that
\begin{align*}
	\psi_n(k_{\pi(1)},\ldots,k_{\pi(n)})&=\psi_n(k_1,\ldots,k_n),\quad(\otimes_{i=1}^n\mu)\text{-a.e.,}
\end{align*}
for every permutation $\pi$ of $\{1,\ldots,n\}$. Here $k_1,\ldots,k_n\in\mc K$.
Convenient in many computations are the exponential vectors
\begin{align*}
	\epsilon(f)&\coloneq(1,f,\ldots,(n!)^{-1/2}f^{\otimes_n},\ldots\;)\in\Fock,\quad f\in \HP,
\end{align*}
where $f^{\otimes_n}(k_1,\ldots,k_n)\coloneq \prod_{j=1}^nf(k_j)$. 
The map $\HP\ni f\mapsto\expv{f}\in\Fock$ is analytic and the set of all exponential vectors is total in $\Fock$.

Next, we introduce the most basic Fock space operators employed in this article:
The creation and annihilation operators corresponding to $h\in\HP$ are, respectively, given by
\begin{align}\label{def:ada}
	\ad(h)\expv{f}&\coloneq\epsilon'(f)h,\quad a(h)\expv{f}\coloneq \langle h|f\rangle_{\HP}\expv{f},\quad f\in\HP,
\end{align}
plus linear and closed extension. We know that
\begin{align}\label{eq:adavp}
	\ad(h)^*&=a(h),\quad\vp(h)=\vp(h)^*,\quad\text{where $\vp(h)\coloneq(\ad(h)+a(h))^{**}$.}
\end{align}
The operator $\vp(h)$ is called the field operator corresponding to $h\in\HP$.

The number operator on $\Fock$ is given by 
\begin{align*}
	\NO\phi\coloneq (n\phi_n)_{n=0}^\infty,
\end{align*} 
for all Fock space vectors $\phi=(\phi_n)_{n=0}^\infty$ such that $N\phi$ again belongs to $\Fock$.
Its action on an exponential vector reads
\begin{align}\label{Nexpv}
	\NO\expv{f}&=\ad(f)\expv{f}=\epsilon'(f)f,\quad f\in\HP.
\end{align}
The form domain of $\NO$ is contained in $\dom(a(f))$, $\dom(\ad(f))$ and $\dom(\vp(f))$
for all $f\in\HP$ and, for $\phi\in\fdom(N)$, we have the relative bounds
\begin{align}\label{rbNa}
\|a(f)\phi\|_{\Fock}&\le \|f\|_{\HP}\|\NO^{1/2}\phi\|_{\Fock},
\\\label{rbNvp}
	\|\vp(f)\phi\|_{\Fock}&\le2^{1/2} \|f\|_{\HP}\|(\NO+1)^{1/2}\phi\|_{\Fock},
	\quad|\langle\phi|\vp(f)\phi\rangle_{\Fock}|\le2\|f\|_{\HP}\|\NO^{1/2}\phi\|_{\Fock}\|\phi\|_{\Fock}.
\end{align}
Finally, we shall sometimes deal with the pointwise annihilation operator. 
Slightly deviating from the notation used in the introduction, we use the symbol $\ap$ here, 
to distinguish it from the smeared annihilation operator defined above. 
Its most convenient mathematical interpretation for us is to consider it as an operator on $L^2(\Geb,\fdom(\NO))$, where
\begin{align}\label{defGebd}
d\in\NN,\;d\ge2\quad\text{and $\Geb\subset\RR^d$ is open and non-empty,}
\end{align}
and $\fdom(\NO)$ is considered as a Hilbert space equipped with the form norm associated with $N$.
We can then define the pointwise annilation operator as the unique bounded linear map
\begin{align}\label{eq:pointwise0}
	\ap:L^2(\Geb,\fdom(\NO))&\longrightarrow L^2(\mu;L^2(\Geb,\Fock))\coloneq\int_{\mc{K}}^\oplus L^2(\Geb,\Fock)\Id\mu
\end{align}
such that
\begin{align}\label{eq:pointwise}
	\ap g\expv{f}=fg\expv{f},\quad f\in\HP,\,g\in L^2(\Geb).
\end{align}
Henceforth, representatives of $\ap\Psi\in L^2(\mu;L^2(\Geb,\Fock))$ are denoted by $\ap(\cdot)\Psi$. 
Then, for any map $x\mapsto f_x$ in $\mc{L}^\infty(\Geb,\HP)$ and all $\Psi\in L^2(\Geb,\fdom(\NO))$ and $\Phi\in L^2(\Geb,\Fock)$, 
\begin{align}\label{eq:pwsmeared}
	\int_{\Geb}\langle \Phi(x)| a(f_x)\Psi(x)\rangle_{\Fock} \Id x
	=\int_{\mc{K}}\int_{\Geb} \langle f_x(k)\Phi(x) | \ap(k)\Psi(x)\rangle_{\Fock}\Id x\,\Id\mu(k).
\end{align}

\subsection{Standing hypotheses}\label{ssec:hypvAV}

Heading towards a mathematical definition of the Hamiltonians studied in this article,
we use this \lcnamecref{ssec:hypvAV} to explain our standing assumptions on the coupling functions 
$\vt$ and $v$ determining the electron-phonon interaction
as well as on the electrostatic potential~$V$ and magnetic vector potential~$A$. 

\subsubsection{Assumptions on coupling functions}\label{sssec:coupling}

Occasionally, in technical proof steps for instance, we shall consider an ultraviolet regular coupling function
\begin{align}\label{def:vt}
\vt\in \mc{L}^\infty(\RR^d,\HP),
\end{align}
and we shall typically write $\vt_x$ for $\vt(x)$. Then the field operators $\vp(\vt_x)$
are well-defined and their domain contains $\fdom(N)$ for all $x\in\RR^d$. This will in general not be the case for the coupling function
$v$ covering the physically most relevant cases:

Recalling \cref{defGebd},
we always assume that $v:\Geb\times\mc{K}\to\CC$ and $\lambda:\mc{K}\to[0,\infty)$ 
are measurable functions having the following properties:
\begin{enumerate}
\item[(a)] For every $k\in\mc{K}$, the function $v(\cdot,k)\in C^\infty(\Geb)$ is bounded with bounded first order partial derivatives and
\begin{align*}
-\frac{1}{2}\Delta_{x}v(x,k)&=\lambda(k)v(x,k),\quad x\in\Geb.
\end{align*}
\item[(b)] $L_1(v)<\infty$ and $\lim_{E\to\infty}L_E(v)=0$ with $L_E(v)\ge0$ given by
\begin{align*}
L_E(v)^2&\coloneq\sup_{x\in\Geb}\int_{\mc{K}}
\frac{E|v(x,k)|^2+|\nabla_xv(x,k)|^2/2}{(E+\lambda(k))^2}\Id\mu(k),
\quad E\in[1,\infty).
\end{align*}
\end{enumerate}
Again we shall typically write $v_x$ for the function $v(x,\cdot)$.

\begin{example}\label{exFpolaron}
It is elementary to verify the above hypotheses (a) and (b) in
the Fr\"{o}hlich model for $\nu\in\NN$ polarons, where $\Geb=\RR^{3\nu}$, $\HP=L^2(\RR^3)$,
\begin{align*}
\lambda(k)=\frac{1}{2}|k|^2\quad\text{and}\quad
v(x,k)=g\sum_{j=1}^\nu\frac{\eul^{-\ii k\cdot x_j}}{(2\pi)^{3/2}}\cdot\frac{2^{1/2}}{|k|},\;\text{if $k\not=0$,}
\end{align*}
for all $k\in\RR^3$ and $x=(x_1,\ldots,x_\nu)\in\RR^{3\nu}$ and some coupling constant $g\in\RR\setminus\{0\}$. 
\end{example}

\begin{example}\label{exFS}
Our assumptions cover the confined polaron model treated in \cite{FrankSeiringer.2021,BrooksMitrouskas.2023}. More generally, let
$m\in\NN$, $m\ge2$, and 
$\mc{G}\subset\RR^m$ be bounded, open and connected with a $C^{1,\delta}$-boundary for some $\delta\in(0,1)$.
(For instance, $\partial\mc{G}$ could be a $C^2$-hypersurface.) Let $\Delta_{\mc{G}}$ denote the Dirichlet Laplacian
on $\mc{G}$ and let $0<\lambda(1)<\lambda(2)\le\lambda(3)\le \ldots\:$ be the eigenvalues of $-(1/2)\Delta_{\mc{G}}$, 
counting multiplicities. Further, let $\{\phi_n:\,n\in\NN\}$ be an orthonormal basis of $L^2(\mc{G})$
of eigenfunctions such that $-\Delta_{\mc{G}}\phi_n=2\lambda(n)\phi_n$ for all $n\in\NN$.
By elliptic regularity, $\phi_n\in C^\infty(\mc{G})$ for every $n\in\NN$, and thanks to, e.g., 
\cite[Appendix~C]{FrankSeiringer.2021} we know that all these eigenfunctions are bounded with bounded
partial derivatives of first order. Now
let $\theta:[0,\infty)\to\RR$ be a strictly positive non-increasing function such that
$\sup_{t\ge0}t^{\ve-1+m/2}\theta(t)<\infty$ for some $\ve>0$. By virtue of \cite[Equation~(C.10)]{FrankSeiringer.2021} we further know that
\begin{align*}
\sup_{y\in\mc{G}}\sum_{n=1}^\infty\theta(\lambda(n))\frac{E|\phi_n(y)|^2+|\nabla\phi_n(y)|^2/2}{(E+\lambda(n))^2}
&\le c_{\mc{G}}\int_{\RR^m}\frac{\theta(|k|^2/2)}{E+|k|^2/2}\Id k,
\end{align*}
where the integral on the right hand side is finite for every $E\ge1$.
Thus, choosing $d=m\nu$ and $\Geb=\mc{G}^\nu$ for some $\nu\in\NN$, as well as
$\HP= L^2(\NN,\fr{P}(\NN),\zeta)$ with $\zeta$ denoting the counting measure
on the power set $\fr{P}(\NN)$ of $\NN$, and finally
\begin{align*}
v(x,n)=g\sum_{j=1}^\nu\theta(\lambda(n))^{1/2}\phi_n(x_j),\quad x=(x_1,\ldots,x_\nu)\in\mc{G}^\nu,\, n\in\NN,
\end{align*}
for some $g\in\RR\setminus\{0\}$,
we see that the hypotheses (a) and (b) are fulfilled in the present example.
The confined multi-polaron model for $\nu$ polarons (as appearing in \cite{FrankSeiringer.2021} for $\nu=1$) 
is obtained by choosing $m=3$ and $\theta(t)=t^{-1}$ for $t\ge\lambda(1)$.
\end{example}

\subsubsection{Assumptions on electrostatic and magnetic vector potentials}\label{sssec:hypVA}
We shall always assume that $A\in L^2_{\loc}(\Geb,\RR^{d})$, without further reference. 

With regards to the electrostatic potential $V:\Geb\to\RR$,
we always assume that $V=V_+-V_-\restr_{\Geb}$ where $V_\pm\ge0$, $V_+\in L_\loc^1(\Geb)$ and the negative part 
is the restriction to $\Geb$ of a function $V_-:\RR^d\to\RR$ 
belonging to the $d$-dimensional Kato class. The latter assumption means that $V_-$ is measurable and
\begin{align*}
	\lim_{r\downarrow0}\sup_{x\in\RR^{d}}\int_{\RR^d}\chi_{[0,r)}(|x-y|)E(x-y)V_-(y)\Id y&=0, 
\end{align*}
where, with $\omega_d$ denoting the hypersurface area of the unit sphere $S^{d-1}$,
\begin{align}\label{EFund}
E(0)\coloneq0,\qquad
\forall z\in\RR^d\setminus\{0\}:\quad
E(z)\coloneq \begin{cases}((d-2)\omega_d)^{-1}|z|^{2-d}, & d\ge 3,\\-(2\pi)^{-1}\ln(|z|), & d=2.\end{cases}
\end{align}
Then $V_-$ is infinitesimally form bounded with respect to $-\Delta$ and,
for every $p\in(0,\infty)$, there exists $c_p\in(0,\infty)$ such that
\begin{align}\label{expmbdVminus}
	\sup_{x\in\RR^{d}}\EE\Big[\eul^{p\int_0^tV_-(b_s^x)\Id x}\Big]&\le \eul^{c_p(1+t)},\quad t\ge0,
\end{align}
for any $d$-dimensional standard Brownian motion $b=(b_t)_{t\ge0}$. As before, $b_t^x\coloneq x+b_t$.
Proofs of these facts on $V_-$ can be found in \cite{AizenmanSimon.1982}.
Given any $\ve>0$, we have in particular the quadratic form bound
\begin{align}\label{rfbVminusDirlaplace}
V_-\restr_{\Geb}&\le-\frac{\ve}{2}\Delta_{\Geb}+c_\ve,
\end{align}
for some $c_\ve\in(0,\infty)$ also depending on $V_-$, of course,
where $\Delta_{\Geb}$ denotes the Dirichlet Laplacian on $\Geb$.

\subsection{Definition of the Hamiltonian via quadratic forms}\label{sssec:defHam}

We now construct polaron type Hamiltonians for systems confined to $\Geb$.

For all $j\in\{1,\ldots,d\}$, we first define a symmetric operator $w_j$ in $L^2(\Geb,\Fock)$ by
\begin{align*}
	w_j\Psi&\coloneq-\ii\partial_j\Psi-A_j\Psi,
	\quad\Psi\in\dom(w_j)\coloneq\mathrm{span}\{f\phi|\,f\in C_0^\infty(\Geb),\,\phi\in\Fock\}.
\end{align*}
Then we introduce a ``maximal'' non-negative quadratic form setting
\begin{align}\label{def:posform}
	\fr{q}^{\max}[\Psi]&\coloneq\frac{1}{2}\sum_{j=1}^{d}\|w_j^*\Psi\|^2
	+\int_{\Geb}\|\NO^{1/2}\Psi(x)\|_{\Fock}^2\Id x+
	\int_{\Geb}V_+(x)\|\Psi(x)\|_{\Fock}^2\Id x,
	\quad
\end{align}
for all $\Psi\in\dom(\fr{q}^{\max})\coloneq 
L^2(\Geb,\fdom(\NO))\cap\fdom(V_+\id_{\Fock})\cap\bigcap_{j=1}^{d}\dom(w_j^*)$.
The form $\fr{q}^{\max}$ is closed as a sum of non-negative closed quadratic forms.
The selfadjoint operator representing it corresponds to Neumann boundary conditions, 
whereas the form corresponding to Dirichlet boundary conditions is the ``minimal'' form
\begin{align*}
	\fr{q}^{\min}\coloneq\ol{\fr{q}^{\max}\restr_{\operatorname{span}\{f\phi\mid f\in C_0^\infty(\Geb),\phi\in\fdom(N)\}}}.
\end{align*}
\begin{rem}\label{rem:qminqmax}
In the case $\Geb=\RR^{d}$, we know that $\fr q^{\min}=\fr q^{\max}$  \cite{Simon.1979b,Matte.2017}.
\end{rem}
As we know thanks to suitable diamagnetic inequalities (see, e.g., \cite[\textsection4]{Matte.2017}), the quadratic
form defined by $V_-$ on $L^2(\Geb,\Fock)$ is again infinitesimally bounded with respect to $\fr{q}^{\min}$.
More precisely, for every $\ve>0$, the form bound \cref{rfbVminusDirlaplace} implies
\begin{align}\label{qfbVminusq}
\int_{\Geb}V_-(x)\|\Psi(x)\|_{\Fock}^2\Id x&\le \ve\fr{q}^{\min}[\Psi]+c_\ve\|\Psi\|^2,
\quad \Psi\in\dom(\fr{q}^{\min}).
\end{align}
In general and in particular in the physically most interesting cases, the functions
$v_x= v(x,\cdot)$ are not square-integrable, whence the field operators $\vp(v_x)$, $x\in\Geb$, that heuristically
should describe the electron-phonon interaction are ill-defined. As is well-known (at least for some $A$)
the interaction term is, however, meaningful when considered as a quadratic form with domain $\dom(\fr{q}^{\min})$.
This is the content of the next theorem, which follows from an adaption and minor elaboration
of a well-known argument by Lieb and Yamazaki \cite{LiebYamazaki.1958}.
For the reader's convenience, we present its proof in Appendix~\ref{app:rfbinteraction}.
\begin{thm}\label{thm:LY}
	For every $\Psi\in\dom(\fr{q}^{\min})$, the iterated integral
	\begin{align}\label{def:frw}
		\fr{w}(v)[\Psi]&\coloneq2\Re\int_{\mc{K}}
		\int_{\Geb}\ol{v}(x,k)\langle\Psi(x)|\ap(k)\Psi(x)\rangle_{\Fock}\Id x\, \Id\mu(k)
	\end{align}
	is well-defined. The so-obtained quadratic form $\fr{w}(v)$ is infinitesimally
	$\fr{q}^{\min}$-bounded. In fact, for all $E\ge1$,
	\begin{align}\label{qfbdinteraction}
		|\fr{w}(v)[\Psi]|&\le 2 L_E(v)\fr{q}^{\min}[\Psi]+2 L_E(v)E\|\Psi\|^2,\quad \Psi\in\dom(\fr{q}^{\min}),
	\end{align}
	where $L_E(v)\xrightarrow{E\to\infty}0$ by assumption. Finally,
	\begin{align}\label{eq:fieldinteraction}
		v_x\in \HP, \;x\in\Geb\quad\Rightarrow\quad\fr{w}(v)[\Psi]&=
		\int_{\Geb}\langle\Psi(x)|\vp(v_{x})\Psi(x)\rangle_{\Fock}\Id x,
		\quad\Psi\in\dom(\fr{q}^{\min}).
	\end{align}
\end{thm}
\begin{rem}
	Note that the order of integration in \eqref{def:frw} matters, 
	as the integrand in general is not simultaneously integrable with respect to $(x,k)\in\Geb\times\mc{K}$.
\end{rem}
\begin{proof}
	The well-definedness of \cref{def:frw} follows from \cref{prop:LY}, 
	whereas \cref{qfbdinteraction} is proved in \cref{cor:relbound}.
	The identity \cref{eq:fieldinteraction} follows from \cref{eq:adavp,eq:pwsmeared}.
\end{proof}
\begin{defn}
The polaron Hamiltonian corresponding to the coupling function $v$ is
the unique selfadjoint operator $H(v)$ representing the following quadratic form,
which is closed and semibounded by the infinitesimal $\fr{q}^{\min}$-boundedness of both $V_-$ and $\fr{w}(v)$,
\begin{align*}
	\fr{h}(v)[\Psi]\coloneq \fr{q}^{\min}[\Psi]-\int_{\Geb}V_-(x)\|\Psi(x)\|_{\Fock}^2\Id x+\fr{w}(v)[\Psi],\quad
	\Psi\in\dom(\fr{h}(v))=\dom(\fr{q}^{\min}).
\end{align*}
\end{defn}

Let us discuss the relation of our definition of the polaron Hamiltonian with a more direct one for ultraviolet regular electron-phonon interactions described by $\vt$, cf. \cref{def:vt}.
In view of the first bound in \cref{rbNvp} the direct integral of $\vp(\vt_x)$, $x\in\Geb$, 
is infinitesimally operator bounded with respect to $H(0)$. Hence, by the last implication in Theorem~\ref{thm:LY}
we do not run into notational conflicts setting
\begin{align}\label{Hvt}
(H(\vt)\Psi)(x)&\coloneq(H(0)\Psi)(x)+\vp(\vt_x)\Psi(x),\quad\text{a.e. $x\in\Geb$}, 
\end{align}
for all $\Psi\in\dom(H(\vt))=\dom(H(0))$. While we introduce $H(\vt)$ mainly to work with it in
technical proof steps, we point out  that no regularity assumptions other than $\vt\in \mc{L}^\infty(\Geb,\HP)$ are imposed on the
$x$-dependence of $\vt$. 

Sometimes polaron Hamiltonians are defined by approximating $v$ by a sequence of 
coupling functions in $\mc{L}^\infty(\Geb,\HP)$ and observing resolvent convergence of the so-obtained sequence
of regularized Hamiltonians to some limiting Hamiltonian. The latter then must agree with $H(v)$ in view of
 \cref{Hvt} and the next corollary. We shall need its statement
in an approximation step in our proof of the Feynman--Kac formula for $H(v)$. The uniformity
in $A$ of the convergence \cref{Aunifresconv} is exploited in \cref{sec:rescontA}.

\begin{cor}\label{cor:normresconvUV}
	Let also $v_1,v_2,\ldots:\Geb\times\mc{K}\to\CC$ be measurable and satisfy the assumptions {\rm(a)} and {\rm(b)}
	of \cref{sssec:coupling} with the same $\lambda$.
	Assume that $L_1(v_n-v)\xrightarrow{n\to\infty}0$. Then $H(v_n)$ converges to $H(v)$ in the norm
	resolvent sense as $n\to\infty$. In fact, we find $A$-independent numbers $c,n_0>0$ such that
	$H(v)+c\ge1/2$ and $H(v_n)+c\ge1/2$ for all integers $n\ge n_0$ in the quadratic form sense and such that
	\begin{align}\label{Aunifresconv}
	\lim_{n\to\infty}\sup_{A\in L^2_{\loc}(\Geb,\RR^{d})} \sup_{\Psi\in\dom(\fr{q}^{\min}):\|\Psi\|=1}Q(n,A,\Psi)&=0,
	\end{align}
	where we abbreviate, recalling that both $H(v_n)$ and $H(v)$ depend on $A$,
	\begin{align*}
	Q(n,A,\Psi)
	&\coloneq\big\|(H(v)+c)^{1/2}\big((H(v_n)+c)^{-1}-(H(v)+c)^{-1}\big)(H(v)+c)^{1/2}\Psi\big\|.
	\end{align*}
\end{cor}

\begin{proof}
	We know that all forms $\fr{h}(v_n)$, $n\in\NN$, and $\fr{h}(v)$ have the common domain $\dom(\fr{q}^{\min})$.
	We pick some $E\ge1$ such that $2L_E(v)<1/2$.
	Since $L_E(v_n)\to L_E(v)$, $n\to\infty$, we can apply \cref{qfbVminusq,qfbdinteraction} 
	to find $A$-independent $c,n_0\ge0$ 
	such that $\fr{h}(v_n)[\Psi]+c\|\Psi\|^2\ge (E\|\Psi\|^2+\fr{q}^{\min}[\Psi])/2$ for all $\Psi\in\dom(\fr{q}^{\min})$ 
	and $n\ge n_0$ and analogously for $\fr{h}(v)$.
	The bound \eqref{qfbdinteraction} with $v_n-v$ put in place of $v$ now entails
	\begin{align*}
		|\fr{h}(v_n)[\Psi]-\fr{h}(v)[\Psi]|
		\le 4L_E(v_n-v)\big(\fr{h}(v)[\Psi]+c\|\Psi\|^2\big),\quad \Psi\in\dom(\fr{q}),\,n\ge n_0.
	\end{align*}
	For all $n\ge n_0$ with $4 L_E(v_n-v)\le 1/2$ we now infer directly from
	\cite[Lemma~D.1]{HiroshimaMatte.2019} that
	$Q(n,A,\Psi)\le 8 L_E(v_n-v)\|\Psi\|$,
	which implies all assertions.
\end{proof}


\section{Presentation of the Feynman--Kac formulas}\label{ssec:FK}

\noindent
We now move to the presentation of our Feynman--Kac formulas for the operators 
$H(\vt)$ and $H(v)$ defined in \cref{sssec:defHam}.
These formulas comprise several stochastic processes that we shall introduce step by step in what follows.

\subsection{Brownian motions and their time-reversals}

In the whole article we fix some filtered probability space
$(\Omega,\fr{F},(\fr{F}_t)_{t\ge0},\PP)$ satisfying the usual assumptions, i.e.,
the measure space $(\Omega,\fr{F},\PP)$ is complete and, for all $t\ge0$, the 
sub-$\sigma$-algebra $\fr{F}_t$ contains the set $\fr{N}$ of all $\PP$-zero sets and satisfies
$\fr{F}_t=\bigcap_{r>t}\fr{F}_r$. Expectations with respect to $\PP$ will be denoted
by $\EE$, conditional expectations given $\fr{F}_t$ by $\EE^{\fr{F}_t}$ for any $t\ge0$.
Furthermore, $b=(b_t)_{t\ge0}$ always denotes a $d$-dimensional $(\fr{F}_t)_{t\ge0}$-Brownian motion. 
For any $\fr{F}_0$-measurable $q:\Omega\to\RR^{d}$, we set $b^q_t\coloneq q+b_t$. 

Let $t>0$, $x\in\RR^{d}$ and consider the reversed Brownian motion
\begin{align*}
	b^{t;x}&\coloneq(b^x_{t-s})_{s\in[0,t]}.
\end{align*}
We know from the theory of reversed diffusion processes developed in \cite{HaussmannPardoux.1986,Pardoux.1986} 
that $b^{t;x}$ is a continuous semimartingale on $(\Omega,\fr{F},\PP)$ with respect
to the filtration $(\fr{G}^t_s)_{s\in[0,t]}$, where $\fr{G}^t_s$ is the smallest sub-$\sigma$-algebra of $\fr{F}$
containing $\fr{N}$ such that $b_{t-s}$ and all increments $b_t-b_{t-r}$
with $r\in[0,s]$ are $\fr{G}^t_s$-measurable.

\subsection{Path integrals involving $A$ and $V$}\label{ssecStrat}

As is well-known from the theory of Schr\"{o}dinger operators, the vector potential should contribute 
to the Feynman--Kac integrand via the Stratonovic integral of $A(b^x)$ along $b^x$.
A canonical generalization of this integral for our merely locally square-integrable $A$ in the case $\Geb=\RR^d$ is
\begin{align}\label{defPhitx}
	\Phi_t(x)&\coloneq \frac{1}{2}\int_0^tA(b_s^x)\Id b_s^x-\frac{1}{2}\int_0^tA(b_s^{t;x})\Id b_s^{t;x},\quad t\ge0.
\end{align}
As shown in \cite[Lemma~9.1]{Matte.2021},
the process $\Phi(x)=(\Phi_t(x))_{t\ge0}$ is well-defined and adapted to $(\fr{F}_t)_{t\ge0}$ for a.e. $x\in\RR^{d}$.
The first stochastic integral in \eqref{defPhitx} is constructed using the filtration $(\fr{F}_t)_{t\ge0}$,
the second one by means of $(\fr{G}^t_s)_{s\in[0,t]}$ for each fixed $t>0$.
This type of generalized Stratonovic integral has been used in \cite{FoellmerProtter.2000} to derive It\^{o} formulas for
functions of low regularity. Unaware of \cite{FoellmerProtter.2000}, 
the second author employed the definition \eqref{defPhitx}
in \cite{Matte.2021} to derive Feynman--Kac formulas for Pauli--Fierz Hamiltonians with singular coefficients.
The idea behind \eqref{defPhitx} is simple: On the one hand, it is common to define the Stratonovic integral of $A(b^x)$ 
along $b^x$ over the time interval $(0,t]$ as the limit in probability of the arithmetic mean of Riemann sums 
corresponding to partitions $0=t_0<t_1<\ldots t_n=t$ using initial and end point evaluations, respectively:
\begin{align*}
	&\sum_{i=1}^n\frac{1}{2}(A(b^x_{t_{i-1}})+A(b^x_{t_i}))(b_{t_i}-b_{t_{i-1}})
	\\
	&=\frac{1}{2}\sum_{i=1}^nA(b^x_{t_{i-1}})(b_{t_i}-b_{t_{i-1}})
	-\frac{1}{2}\sum_{j=1}^nA(b^{t;x}_{s_{j-1}})(b^{t;x}_{s_j}-b^{t;x}_{s_{j-1}}),
\end{align*}
where $s_j\coloneq t-t_{n-j}$, $j\in\{0,\ldots,n\}$. On the other hand, by the general theory of stochastic integration
with respect to continuous semimartingales we know that the two sums in the second line converge in
probability to the respective terms in \eqref{defPhitx} as the mesh of the partition goes to zero.
The idea to construct Feynman--Kac integrands for Schr\"{o}dinger operators with
very singular $A$ by combining ``forwards and backwards''
integrals was already present but technically implemented differently in \cite{Hundertmark.1996}.

If, for instance, $A\in C_b^1(\RR^{d},\RR^{d})$, then we know that $\Phi(x)$ 
is well-defined for all $x\in\RR^{d}$ and, $\PP$-a.s., we obtain the familiar expression
\begin{align}\label{Phitxreg}
	\Phi_t(x)&=\int_0^tA(b_s^x)\Id b_s+\frac{1}{2}\int_0^t\mathrm{div}A(b_s^x)\Id s,\quad t\ge0;
\end{align}
see, e.g., \cite[Lemma~8.3]{Matte.2021}.

Still considering the case $\Geb=\RR^d$, we next set
\begin{align}\label{defStx}
	S_t(x)&\coloneq\int_0^tV(b_s^x)\Id s-\ii\Phi_t(x),\quad t\ge0,
\end{align}
for every $x\in\RR^{d}$ for which $\Phi(x)$ is defined.
Here we should remark that, for any given $x\in\RR^{d}$, we only know $\PP$-a.s. that $V(b^x):[0,\infty)\to\RR$ is 
locally integrable \cite[Lemma~2]{FarisSimon.1975}. 
We therefore introduce the convention that the path integrals of $V$ in \eqref{defStx}
have to be read as $0$ at every $\gamma\in\Omega$ for which $V(b_\bullet^x(\gamma))$ is not locally integrable.

Finally, we consider general open $\Geb$.
If $A$ and $V_+$ have extensions to locally square-integrable and locally integrable functions on all of $\RR^d$, 
respectively, the above construction carries over.
Otherwise, we pick open sets $\Geb_n\subset\Geb$, $n\in\NN$, satisfying
$\ol{\Geb_n}\subset\Geb_{n+1}$ for all $n\in\NN$ and $\bigcup_{n=1}^\infty\Geb_n=\Geb$.
We define $S^n_t(x)$ by putting $\chi_{\Geb_n}A$ and $\chi_{\Geb_n}V_+$, extended by $0$
to functions on $\RR^d$, in place of $A$ and $V_+$ in the above formulas.
 Introducing the first exit times
\begin{align}\label{exittaun}
	\tau_n(x)&\coloneq \inf\{ t\ge0|\,b_t^x\in\Geb_n^c\},\quad x\in\RR^d,
\end{align}
we then know from \cite[\textsection9.2]{Matte.2021} that $S_t^n(x)=S^m_t(x)$, on
$\{t<\tau_n(x)\}\setminus\mc{N}$ for all $m,n\in\NN$ with $n<m$ and some possibly $(t,x)$-dependent
$\PP$-zero set $\mc{N}$. For fixed $t\ge0$ and a.e. $x\in\RR^d$,
it therefore makes sense to define $S_t(x)\coloneq S^n_t(x)$ on $\{t<\tau_n(x)\}\setminus\mc{N}$ for every $n\in\NN$.
For convenience we set $S_t(x)=0$ on $\{t\ge\tau_{\Geb}(x)\}\cup\mc{N}$.
Since $\bigcup_{n=1}^\infty\{t<\tau_n(x)\}=\{t<\tau_{\Geb}(x)\}$ with $\tau_{\Geb}(x)$ given by \cref{defexitGeb},
we thus obtain a well-defined $\fr{F}_t$-measurable random variable $S_t(x)$, whose
construction does, up to changes on $\PP$-zero sets, 
not depend on the chosen sequence $(\Geb_n)_{n\in\NN}$.

\subsection{Processes appearing in the interaction terms}\label{ssec:Upm}

For every $x\in\RR^d$, we introduce the pathwise well-defined $\HP$-valued Bochner-Lebesgue integrals
\begin{align}\label{def:UpmUV}
	U_{\reg,t}^-(\vt;x)&\coloneq\int_0^{t}\eul^{-s}\vt_{b_s^x}\Id s,\quad 
	U_{\reg,t}^+(\vt;x)\coloneq\int_0^{t}\eul^{-(t-s)}\vt_{b_s^x}\Id s,\quad t\ge0.
\end{align}
In general, when $\vt$ is replaced by $v$, these expressions can no longer be defined as $\HP$-valued integrals.
Assuming $v\in \mc{L}^\infty(\Geb,\HP)$ in addition to our standing hypotheses on $v$
and employing It\^{o}'s formula, we shall, however,
find alternative expressions that stay meaningful when the additional assumption on $v$ is dropped again. 
We are thus led to the definitions \cref{defUinftyminus,defUinftyplus}. There, we split $v$ into two parts separated by the level
set $\{\lambda=\sigma\}$ for some $\sigma\in[2,\infty)$. More precisely, we set
\begin{align}\label{def:cutoffv}
v_{\sigma,x}&\coloneq\chi_{\{\sigma\le\lambda\}}v_x,\quad \tilde{v}_{\sigma,x}\coloneq\chi_{\{\lambda<\sigma\}}v_x,
\quad\text{if $x\in\Geb$,}
\end{align}
and furthermore, noting that $\lambda-1\ge1$ on $\{v_{\sigma,x}\not=0\}$ since $\sigma\ge2$,
\begin{align}\label{defbetainfty}
	\beta_{\sigma,x}^\pm&\coloneq (\mp1+\lambda)^{-1}v_{\sigma,x},\quad 
	\alpha^\pm_{\sigma,x}\coloneq\nabla_x\beta^\pm_{\sigma,x},\quad\text{if $x\in\Geb$.}
\end{align}
Here the gradient is {\em a priori} computed pointwise on $\mc{K}$, i.e., by definition
$\alpha_{\sigma,x}^\pm(k)=(\mp1+\lambda(k))^{-1}\nabla_{x}v(x,k)$, $x\in\Geb$, for each fixed $k\in\mc{K}$.
According to \cref{lembetaC1HP}, however, the map $x\mapsto\beta^\pm_{\sigma,x}$ 
is in $C^1(\Geb,\HP)$ and $\alpha_{\sigma,x}^\pm$ is its gradient at $x\in\Geb$ computed with respect to the
norm on $\HP$. It is convenient to extend the above functions by
\begin{align}\label{vbetaext}
v_{\sigma,x}\coloneq\tilde{v}_{\sigma,x}\coloneq\beta^\pm_{\sigma,x}\coloneq0,
\quad\alpha^\pm_{\sigma,x}\coloneq0,
\quad\text{if $x\in\Geb^c$,}
\end{align}
so that $\tilde{v}_{\sigma},\beta_{\sigma}^\pm\in \mc{L}^\infty(\RR^d,\HP)$ and
$\alpha_{\sigma}^\pm\in \mc{L}^\infty(\RR^d,\HP^d)$ for all $\sigma\in[2,\infty)$.
Then the $\HP$-valued isometric stochastic integrals
\begin{align}\label{def:Msigmapmx}
M_{\sigma,t}^\pm(x)&\coloneq \int_0^t\eul^{\pm s}\alpha^\pm_{\sigma,b_s^x}\Id b_s,\quad t\ge0,
\end{align}
are manifestly well-defined $L^2$-martingales for every $x\in\RR^d$. Finally, we set
\begin{align}\label{defUinftyminus}
	U^-_{\sigma,t}(x)
	&\coloneq U_{\reg,t}^-(\tilde{v}_{\sigma};x)+\beta^-_{\sigma,x}-\eul^{-t}\beta^-_{\sigma,b_t^x}+M_{\sigma,t}^-(x),
	\\\label{defUinftyplus}
	U^+_{\sigma,t}(x)
	&\coloneq U_{\reg,t}^+(\tilde{v}_{\sigma};x)+\eul^{-t}\beta^+_{\sigma,x}-\beta_{\sigma,b_t^x}^+
	+\eul^{-t}M_{\sigma,t}^+(x),
\end{align}
for all $t\ge0$ and $x\in\RR^d$.

The following statement provides the legitimization for the above definitions.
\begin{lem}\label{lemUregU}
Additionally assume that $(1+\lambda)v\in \mc{L}^\infty(\Geb,\HP)$ and set $v_x\coloneq0$ whenever $x\in\Geb^c$.
Let $x\in\Geb$ and $\sigma\in[2,\infty)$.
Then, $\PP$-a.s., $U_{\reg,t}^\pm(v;x)$ and $U^\pm_{\sigma,t}(x)$ 
agree on $\{t<\tau_{\Geb}(x)\}$ for all $t\ge0$.
\end{lem}
\begin{proof}
	The assertion follows from \cref{def:UpmUV,lemItoM}\,(iii).
\end{proof}
Especially, this yields the independence of our processes of the parameter $\sigma$.
\begin{cor}\label{lemsigmaUpm}
Let $x\in\Geb$ and $\sigma,\kappa\in[2,\infty)$. Then, $\PP$-a.s., 
$U^\pm_{\sigma,t}(x)$ and $U^\pm_{\kappa,t}(x)$ agree on $\{t<\tau_{\Geb}(x)\}$ for all $t\ge0$.
\end{cor}
\begin{proof}
Assuming $\sigma<\kappa$, the assertion follows from \cref{lemUregU} with $\tilde{v}_\kappa$
put in place of $v$ as well as the formulas \cref{defUinftyminus,defUinftyplus} applied to both $\sigma$ and $\kappa$.
\end{proof}

\subsection{The complex action}\label{ssec:complexaction}

We now introduce the analogue of Feynman's complex action in the model treated here. 
Once more, we start with $\vt\in \mc{L}^\infty(\RR^d,\HP)$ for which we set
\begin{align}\label{def:uUV}
	u_{\reg,t}(\vt;x)&\coloneq\int_0^t\langle \vt_{b_s^x}|U_{\reg,s}^+(\vt;x)\rangle_{\HP}\Id s,\quad t\ge0,\,x\in\RR^d.
\end{align}
Again this expression is ill-defined in general when $\vt$ is substituted by $v$.
The complex action associated with the possibly ultraviolet singular $v$ is defined by
\begin{align}\nonumber
	u_{\sigma,t}(x)&\coloneq u_{\reg,t}(\tilde{v}_\sigma;x)+a_{\sigma,t}(x)+w_{\sigma,t}(x) 
	- \langle\beta_{\sigma,b_t^x}^- | \eul^{-t}M_{\sigma,t}^+(x) \rangle_{\HP} 
	+ \langle M_{\sigma,t}^-(x)| \beta_{\sigma,x}^+ \rangle_{\HP}
	\\\label{def:actioninfty}
	&\quad -\ii\int_0^t\Im\langle\alpha_{\sigma,b_s^x}^-|\beta_{\sigma,b_s^x}^+\rangle_{\HP}\Id b_s
	+m_{\sigma,t}(x),\quad t\ge0,
\end{align}
for all $x\in\RR^d$, with
\begin{align}\label{def:asigma}
a_{\sigma,t}(x)&\coloneq
\frac{1}{2}\langle\beta_{\sigma,x}^-|\beta_{\sigma,x}^+\rangle_{\HP}
+\frac{1}{2}\langle\beta_{\sigma,b_t^x}^-|\beta_{\sigma,b_t^x}^+\rangle_{\HP}
-\eul^{-t}\langle\beta_{\sigma,b_t^x}^-|\beta_{\sigma,x}^+\rangle_{\HP},
\\\label{def:wsigma}
	w_{\sigma,t}(x) &\coloneq 
	\int_0^t \bigg(\frac{1}{2}\langle\alpha_{\sigma,b_s^x}^-|\alpha_{\sigma,b_s^x}^+\rangle_{\HP}-
	\langle\beta_{\sigma,b_s^x}^-|\beta_{\sigma,b_s^x}^+\rangle_{\HP}\bigg)\Id s,
	\\\label{def:msigma}
	  m_{\sigma,t}(x) &\coloneq  \int_0^t \langle \alpha_{\sigma,b_s^x}^- | \eul^{-s}M_{\sigma,s}^+(x) \rangle_{\HP} \Id b_{s}.
\end{align}
Our definition \cref{def:actioninfty} is motivated by the following result:

\begin{lem}\label{lemuregu}
Additionally assume that $(1+\lambda)v\in \mc{L}^\infty(\Geb,\HP)$. Let $x\in\Geb$ and $\sigma\in[2,\infty)$.
Then, $\PP$-a.s., $u_{\reg,t}(v;x)$ and $u_{\sigma,t}(x)$ agree on $\{t<\tau_{\Geb}(x)\}$ for all $t\ge0$.
\end{lem}
\begin{proof}
The proof of this lemma can be found at the end of \cref{ssec:regforu}.
\end{proof}
Again, the choice of $\sigma$ in \cref{def:actioninfty,def:asigma,def:wsigma,def:msigma} 
is immaterial for all $x$ of actual relevance:
\begin{cor}\label{lemsigmau}
Let $x\in\Geb$ and $\sigma,\kappa\in[2,\infty)$. Then, $\PP$-a.s., 
$u_{\sigma,t}(x)$ and $u_{\kappa,t}(x)$ agree on $\{t<\tau_{\Geb}(x)\}$ for all $t\ge0$.
\end{cor}
\begin{proof}
Assuming $\sigma<\kappa$, we consider the definition of $u_{\kappa,t}(x)$, i.e., the
right hand side of \cref{def:actioninfty} with $\kappa$ put in place of $\sigma$.
Then we apply \cref{lemuregu} with $\tilde{v}_\kappa$ put in place of $v$ to re-write
the expression $u_{\reg,t}(\tilde{v}_\kappa;x)$ on $\{t<\tau_{\Geb}(x)\}$.
Then we see that, on $\{t<\tau_{\Geb}(x)\}$, $u_{\kappa,t}(x)$ is equal to $u_{\reg,t}(\tilde{v}_\sigma;x)$
plus a sum of terms that combine to the remaining six members on
right hand side of \cref{def:actioninfty} with parameter $\sigma$.
\end{proof}
\begin{rem}\label{rem:Feynmansaction}
In the situation of \cref{exFpolaron}, $u_{\sigma,t}(x)$ agrees $\PP$-a.s. with Feynman's famous
expressions for the complex action in the polaron model on $\RR^3$. This is shown in \cref{app:action},
where we also find a more compact expression for $u_{\sigma,t}(x)$ in the situation of \cref{exFS}
analogous to Feynman's formula. Our formula \cref{def:actioninfty} is useful in the general setting treated here
for deriving exponential moment bounds on $u_{\sigma,t}(x)$ and convergence theorems for 
sequences of complex actions corresponding to different coupling functions.
\end{rem}
\begin{rem}\label{rem:HPRR}
In physical applications, the vectors $\beta^\pm_{\sigma,x}$ and the components of $\alpha_{\sigma,x}^\pm$
with $x\in\Geb$ all belong to a certain {\em real} subspace $\HP_{\RR}$ of $\HP$ and, hence, the first,
purely imaginary term in the second line of \cref{def:actioninfty} is zero and $u_{\sigma,t}(x)$ is real-valued.
In the situation of \cref{exFpolaron}, $\HP_{\RR}=\{f\in L^2(\RR^3)|\,f(-k)=\ol{f(k)},\,\text{a.e. $k$}\}$.
Further, since all eigenfunctions $\phi_n$ of the Dirichlet Laplacian can be chosen real-valued, 
we can choose $\HP_{\RR} = \{f\in \ell^2(\NN)| f(n)\in\RR,\,n\in\NN\}$ in \cref{exFS}.
\end{rem}


\subsection{Feynman--Kac integrands and formulas}

The Fock space operator-valued parts of our Feynman--Kac integrands 
(sometimes called multiplicative functionals) involve a last building block, namely the operator norm convergent series
\begin{align*}
	F_t(h)&\coloneq\sum_{n=0}^\infty\frac{1}{n!}\ad(h)^n\eul^{-tN},\quad h\in\HP,\,t>0.
\end{align*}
From \cite[Appendix~6]{GueneysuMatteMoller.2017} we know indeed that $F_t:\HP\to\LO(\Fock)$ is analytic and
\begin{align}\label{bdFt}
	\|F_t(h)\|&
	\le2^{1/2}\eul^{4(1+t^{-1})\|h\|_{\HP}^2},
	\\\label{bddFt}
	\|F_t'(h)g\|&\le 2^{3/2}(1+t^{-1})^{1/2}\|g\|_{\HP} \eul^{4(1+t^{-1})\|h\|_{\HP}^2},
\end{align}
for all $t>0$ and $g,h\in \HP$.

\subsubsection{Ultraviolet regular coupling functions}
Starting once more with the ultraviolet regular $\vt\in \mc{L}^\infty(\RR^d,\HP)$, we define
\begin{align}\label{def:Wreg}
	W_{\reg,t}(x)&\coloneq \eul^{u_{\reg,t}(\vt;x)}F_{t/2}(-U_{\reg,t}^{+}(\vt;x))F_{t/2}(-U_{\reg,t}^{-}(\vt;x))^*,
\end{align}
for all $t>0$ and $x\in\RR^d$, as well as $W_{\reg,0}(x)\coloneq\id_{\Fock}$.

\begin{rem}\label{remWreg}
Let $t\ge0$ and $x\in\Geb$. Since $F_{t/2}:\HP\to\LO(\Fock)$ is analytic, we see that
	$W_{\reg,t}(x):\Omega\to\LO(\Fock)$ is separably valued and $\fr{F}_t$-measurable. Further,
	we shall see in \cref{prop:knownGMM} that $\|W_{\reg,t}(x)\|\le \eul^{c_{\vt}t}$ on $\Omega$
	with a solely $\vt$-dependent $c_\vt\in(0,\infty)$. A similar bound with a slightly worse right hand side
	also follows from \cref{bdFt} and the bounds
\begin{align}\label{Usigmabd}
		\|U_{\reg,t}^\pm(\vt;x)\|_{\HP}&\le
		\sup_{y\in\RR^d} \|\vt_{y}\|_\HP(1-\eul^{-t}),
		\\\label{brureg}
		|u_{\reg,t}(\vt;x)|&\le t\sup_{y\in\RR^d} \|\vt_{y}\|_\HP^2,
	\end{align}
	which are evident from the definitions \cref{def:UpmUV,def:uUV}.
\end{rem}
\begin{thm}\label{thm:mainFKreg}
	Let $t\ge0$ and $\Psi\in L^2(\Geb,\Fock)$. Then
	\begin{align}\label{eq:FKmainreg}
		(\eul^{-tH(\vt)}\Psi)(x)&=\EE\big[\chi_{\{t<\tau_{\Geb}(x)\}}\eul^{-\ol{S}_t(x)}W_{\reg,t}(x)^*\Psi(b_t^x)\big],
		\quad\text{a.e. $x\in\Geb$.}
	\end{align}
\end{thm}

\begin{proof}
	This theorem is proven in \cref{sec:FKproof}.
\end{proof}

Since $\|W_{\reg,t}(x)\|$ is bounded for all $t\ge0$ and $x\in\RR^d$,
\cref{thm:mainFKreg} can actually be generalized so as to cover all 
non-negative $V_-\in L^1_\loc(\Geb,\RR)$ which are form small
with respect to one-half times the negative Dirichlet-Laplacian on $\Geb$;
see \cite[Proof of Corollary~1.4 in \textsection9.4]{Matte.2021}.

\subsubsection{Ultraviolet singular coupling functions}
Passing to the possibly ultraviolet singular coupling function $v$, we define
\begin{align}\label{def:W}
	W_{\sigma,t}(x)&\coloneq \eul^{u_{\sigma,t}(x)}F_{t/2}(-U_{\sigma,t}^{+}(x))F_{t/2}(-U_{\sigma,t}^{-}(x))^*,
\end{align}
for all $t>0$, $x\in\RR^d$ and $\sigma\in[2,\infty)$, as well as $W_{\sigma,0}(x)\coloneq\id_{\Fock}$.

To control $W_{\sigma,t}(x)$, we will need the following additional assumption on $\Geb$:
Let $\ell(\gamma)$ denote the length of a rectifiable $\RR^d$-valued path $\gamma$ and
put $d_\Geb(x,y) \coloneq \inf\{\ell(\gamma) | \text{$\gamma$ is a rectifiable path in $\Geb$ from $x$ to $y$} \}$, 
if $x$ and $y$ are in the same connected component of $\Geb$, and $d_\Geb(x,y)\coloneq\infty$ otherwise.
Then we introduce the condition
	\begin{align}\label{tailbound}
	\exists\:a_\Geb\ge1,\,C_\Geb>0\;\forall\:r,t>0: \ \
		\sup_{x\in\Geb}\PP[ \chi_{\{t<\tau_\Geb(x)\}}d_\Geb(b_t^x,x)\ge r] \le a_\Geb\eul^{-C_\Geb r^2/t}.
	\end{align}


\begin{example}\label{rem:Gebs}
We discuss a few examples based on the well-known bound
		\begin{align}\label{tailboundRRd} 
		\PP\Big[ \sup_{0\le s\le t}|b_s|\ge r \Big] \le 4d\eul^{-r^2/2dt},
		\end{align}
		cf., e.g., \cite[Lemma~5.2.1]{DemboZeitouni.2010}, which shows that $\Geb=\RR^d$ satisfies \cref{tailbound}.
\begin{enumerate}
\item[(a)] Assume there exists $c_\Geb>0$ such that 
		$d_\Geb(x,y)\le c_\Geb|x-y|$,
		whenever $x$ and $y$ belong to the same connected component of $\Geb$.
		Then \cref{tailbound} directly follows from \cref{tailboundRRd}.
		In particular, \cref{tailbound} holds for convex $\Geb$, where $c_\Geb=1$.
\item[(b)] More generally, assume there exists $b_\Geb>0$ such that,
	for all $x,y\in\Geb$ and every $\gamma\in C([0,1],\Geb)$ with $\gamma(0)=x$ and $\gamma(1)=y$,
	\begin{align*}
	\exists \:t_1,t_2\in[0,1]:\quad|\gamma(t_1)-\gamma(t_2)| \ge b_\Geb d_\Geb(x,y). 
	\end{align*}
	Then \cref{tailbound} holds with $a_\Geb=4d$ and $C_{\Geb}=b_\Geb^2/8d$, since the above condition entails
	$\chi_{\{t<\tau_\Geb(x)\}}d_\Geb(b_t^x,x)\le (2/b_{\Geb})\sup_{s\in[0,t]}|b_s|$, which can be combined with
	 \cref{tailboundRRd}.
\item[(c)]
	The open slit plane $\RR^2\setminus\{(x,0):x\le0\}$ satisfies the condition in (b) (with $b_\Geb=1/2$), but not the one in (a).
	\end{enumerate}
\end{example}
\begin{rem}\label{rem:Wmeas}
	Again $W_{\sigma,t}(x):\Omega\to\LO(\Fock)$ is always separably valued and $\fr{F}_t$-measurable.
	Employing \cref{tailbound} we shall further verify
	that $\chi_{\{t<\tau_\Geb(x)\}}\|W_{\sigma,t}(x)\|$ has finite moments of any order in \cref{lem:mbconvW} .
\end{rem}

We are now in a position to formulate our main result:

\begin{thm}\label{thm:mainFK}
Assume that $\Geb$ fulfills \cref{tailbound}. 
	Pick any $\sigma\in[2,\infty)$ and let $t\ge0$ and $\Psi\in L^2(\Geb,\Fock)$. Then 
	\begin{align}\label{eq:FKmain}
		(\eul^{-tH(v)}\Psi)(x)&=\EE\big[\chi_{\{t<\tau_{\Geb}(x)\}}\eul^{-\ol{S}_t(x)}W_{\sigma,t}(x)^*\Psi(b_t^x)\big],
		\quad\text{a.e. $x\in\Geb$.}
	\end{align}
\end{thm}
\begin{proof}
	This theorem is proven at the end of \cref{sec:FKproof}.
\end{proof}



\section{Feynman--Kac formula for regular coefficients}\label{sec:FKregcoeff}

\noindent
This \lcnamecref{sec:FKregcoeff} is devoted to the proof of \cref{thm:mainFKreg} in the special case
where $V\in C_b(\RR^d,\RR)$ and $A\in C^1_b(\RR^d,\RR^d)$. 
Moreover, with the only exception of \cref{cor:FKUVregAV} at its very end,
we shall always consider the case $\Geb=\RR^d$ in this \lcnamecref{sec:FKregcoeff}.

The main strategy in the case $\Geb=\RR^d$ is to show that the right hand side of \cref{eq:FKmainreg} defines
a strongly continuous semigroup of bounded operators on $L^2(\RR^d,\Fock)$ and to verify that $H(\vt)$ is its
generator. For both tasks we employ a certain stochastic differential equation satisfied
by $(\eul^{-S_t(x)}W_{\reg,t}(\vt;x)\phi)_{t\ge0}$ with $\phi\in\dom(N)$ as a starting point.

In the case $A=0$ the results of this \lcnamecref{sec:FKregcoeff} are actually known from \cite{GueneysuMatteMoller.2017}.

\subsection{Stochastic differential equation}
\label{ssec:Markov}

In the next proposition we find a differential equation pathwise satified by $(W_{\reg,t}(\vt;x)\phi)_{t\ge0}$.
In the succeeding \cref{lem:SDEX} we include $\eul^{-S_t(x)}$ and obtain a true stochastic differential equation,
for non-vanishing $A$ at least.
In both \lcnamecrefs{prop:knownGMM} we shall, for each $x\in\RR^{d}$, employ the operator
\begin{align}\label{eq:Ht}
	\wt{H}^{A,V}(x)&\coloneq\frac{1}{2}|A(x)|^2-\frac{\ii}{2}\mathrm{div}A(x)+V(x)+N+\vp(\vt_{x}),
\end{align}
which by the first bound in \eqref{rbNvp} and the Kato--Rellich theorem 
is a closed operator in $\Fock$ with domain $\dom(N)$.

To prove Markov and resulting semigroup properties later on in this \lcnamecref{sec:FKregcoeff},
we shall also exploit a flow equation associated with the differential equation. 
For this purpose, we denote by $W_{\reg,s,t}(\vt;x)$ the operator obtained from the definition 
\cref{def:Wreg} upon replacing the Brownian motion $b$ by its time shifted version 
\begin{align}\label{def:hatb}
\hat b \coloneq (b_{s+t}-b_s)_{t\ge 0}.
\end{align}

\begin{prop}\label{prop:knownGMM}
		Let $\phi\in\mathrm{span}\{\expv{f}|\,f\in\HP\}$, let $q:\Omega\to\RR^{d}$ be $\fr{F}_0$-measurable and
		abbreviate $Y\coloneq(W_{\reg,t}(\vt;q)\phi)_{t\ge0}$. Then all paths of $Y$ attain values
		in $\dom(N)$ and belong to $C^1([0,\infty),\Fock)$. Moreover, $Y$ is the only
		such process pathwise satisfying the initial value problem
		\begin{align}\label{IVPWUV}
			\frac{\Id}{\Id t}Y_t&=-\wt{H}^{0,0}(b_t^q)Y_t,\quad t\ge0;\qquad Y_0=\phi.
		\end{align}
		Finally,
		\begin{align}\label{nbWreg}
		\sup_{x\in\RR^d}\|W_{\reg,t}(\vt;x)\|&\le \eul^{c_\vt t},\quad t\ge0,
		\quad\text{with}\quad c_\vt\coloneq \sup_{y\in\RR^d}\|\vt_y\|_{\HP}^2,
		\end{align}
		and the following flow equation holds for all $x\in\RR^d$ and $t\ge s\ge 0$,
		\begin{align}\label{flowW}
			W_{\reg,t}(\vt;x)&=W_{\reg,s,t}(\vt;b_s^x)W_{\reg,s}(\vt;x).
		\end{align}
\end{prop}
\begin{proof}
In this proof, we drop the reference to $\vt$ in the notation, so that $W_{\reg,t}(q)=W_{\reg,t}(\vt;q)$ and
$u_{\reg,t}(q)=u_{\reg,t}(\vt;q)$, and so on.
We start by considering the case where $\phi=\expv{f}$ for some $f\in\HP$. Then the formula
	\begin{align}\label{def:Yft}
		Y_{f,t}\coloneq W_{\reg,t}(q)\epsilon(f)&=\eul^{u_{\reg,t}(q)
			-\langle U_{\reg,t}^{-}(q)|f\rangle_{\HP}}
		\expv{\eul^{-t}f-U_{\reg,t}^{+}(q)},\quad t\ge0,
	\end{align}
	can be inferred from \cref{def:ada,eq:adavp}. Then 
	$Y_{f,t}$ is manifestly continuously differentiable and straightforward computations 
	using \cref{def:ada,def:UpmUV,def:uUV} reveal that
	\begin{align*}
		\frac{\Id}{\Id t}Y_{f,t}&=
		-\ad(\eul^{-t}f-U^{+}_{\reg,t}(q))Y_{f,t}
		-\ad(\vt_{b_t^q})Y_{f,t}-\langle \vt_{b_t^q}|\eul^{-t}f-U^{+}_{\reg,t}(q)\rangle_{\HP} Y_{f,t},
	\end{align*}
	for all $t\ge0$. Comparing with \cref{def:ada,Nexpv}, we recognize the action of
	the number and annihilation operators in the first and third terms on the right hand side, respectively.
	In view of \eqref{eq:adavp} this proves \eqref{IVPWUV} for $\phi=\expv{f}$.

Clearly, \eqref{IVPWUV} also holds when $\phi$ is a linear combination of exponential vectors, which we assume 
in the rest of this proof. Differentiating $\|Y_t\|_{\Fock}^2$ and using that, by \eqref{rbNvp}, 
	the spectrum of $\wt{H}_{\UV}^{0,0}(x)$ is bounded from
	below uniformly in $x$ by $-c_\vt$, 
	we deduce that $\|Y_t\|_{\Fock}\le\eul^{c_\vt t}\|\phi\|_{\Fock}$, 
	which entails \cref{nbWreg} and unique solvability of \eqref{IVPWUV}.
	
	Next, let $x\in\RR^d$, $s\ge0$, and set $Z_t\coloneq W_{\reg,s+t}(x)\phi$. 
	Then \eqref{IVPWUV} implies
	\begin{align*}
		\frac{\Id}{\Id t}Z_t&=-\wt{H}^{0,0}(b_{s+t}^x)Z_t,\quad t\ge0;\qquad Z_0=W_{\reg,s}(x)\phi.
	\end{align*}
	Further, applying \eqref{IVPWUV} to the time-shifted Brownian motion
	$\hat{b}$ and filtration $(\fr{F}_{s+t})_{t\ge0}$
	and observing that $\hat{b}_t^{b_s^x}=b_{s+t}^x$, we see that, pathwise, the processes
	$[0,\infty)\ni t\mapsto W_{\reg,s,s+t}(b_s^x)W_{\reg,s}(x)\phi$ and $Z$ both solve
	the same uniquely solvable initial value problem. Since $\phi$ can be chosen in a dense
	subset of $\Fock$, this implies \cref{flowW}.
\end{proof}

\begin{prop}\label{lem:SDEX}
	Assume that $V\in C_b(\RR^{d},\RR)$ and $A\in C_b^1(\RR^{d},\RR^d)$.
	Let $x\in\RR^{d}$ and $\phi\in\mathrm{span}\{\expv{f}|\,f\in\HP\}$. Then 
	$(\eul^{-S_t(x)}W_{\reg,t}(\vt;x)\phi)_{t\ge0}$ is a continuous $\Fock$-valued semimartingale
	whose paths belong $\PP$-a.s. to $C([0,\infty),\dom(N))$ and 
	which $\PP$-a.s. solves
	\begin{align}\label{SDEX}
		X_t&=\phi+\ii\int_0^t A(b_s^x) X_s\Id b_s-\int_0^t\wt{H}^{A,V}(b_s^x)X_s\Id s,\quad t\ge0.
	\end{align}
\end{prop}

\begin{proof}
	Under the present assumption on $A$, $\Phi_t(x)$ is given by \eqref{Phitxreg}.
	The assertion thus follows from \cref{defStx}, \cref{prop:knownGMM} and It\^{o}'s product formula.
\end{proof}

\subsection{Markov and semigroup properties for regular $A$ and $V$}

In this subsection we still restrict our discussion to the case where $\Geb=\RR^d$, $V\in C_b(\RR^{d},\RR)$ and $A\in C_b^1(\RR^{d},\RR^d)$. 
Our goal is to derive a Markov property involving the Feynman--Kac operators given by
\begin{align}\label{def:Treg}
(T_{\reg,t}\Psi)(x)&\coloneq\EE[\eul^{-\ol{S}_t(x)}W_{\reg,t}(\vt;x)^*\Psi(b_t^x)],\quad x\in\RR^d,
\end{align}
for all $\Psi\in L^p(\RR^d,\Fock)$ with $p\in[1,\infty]$. In view of \cref{nbWreg} and since
convolution with the probability density function of $b_t^x$ is a contraction on $L^p(\RR^d)$,
it is clear that the expectation in \cref{def:Treg} is well-defined and 
\begin{align}\label{LpbdTreg}
\|T_{\reg,t}\Psi\|_p&\le\eul^{(c_\vt +\|V\|_\infty)t}\|\Psi\|_p,\quad t\ge0, 
\end{align}
again for all $p\in[1,\infty]$ and $\Psi\in L^p(\RR^d,\Fock)$. As a corollary of the Markov property
proven in the next lemma, $(T_{\reg,t})_{t\ge0}$ turns out to be a semigroup on $L^p(\RR^d,\Fock)$.


\begin{lem}
	Let $\Psi:\RR^{d}\to\Fock$ be measurable and bounded, $x\in\RR^{d}$ and $t\ge s\ge0$. Then, $\PP$-a.s.,
	\begin{align}\label{Mark1}
		\EE^{\fr{F}_s}[\eul^{-\ol{S}_t(x)}W_{\reg,t}(\vt;x)^*\Psi(b_t^x)]
		&=\eul^{-\ol{S}_s(x)}W_{\reg,s}(\vt;x)^*(T_{\reg,t-s}\Psi)(b_s^x).
	\end{align}
\end{lem}

\begin{proof}
       Let $\PP_{\textsc{w}}:\fr{F}_{\textrm{w}}\to[0,1]$ denote the completed Wiener measure on
       $\Omega_{\textrm{w}}\coloneq C([0,\infty),\RR^d)$. Further, let $(\fr{F}_{\textrm{w},t})_{t\ge0}$
       denote the automatically right-continuous
       completion of the natural filtration associated with the evaluation maps
       $\mathrm{ev}_t:\Omega_{\textrm{w}}\to\RR^d$, $\gamma\mapsto\gamma(t)$. 
       Then $\mathrm{ev}=(\mathrm{ev}_t)_{t\ge0}$ is an $(\fr{F}_{\textrm{w},t})_{t\ge0}$-Brownian motion.
       Denote by $(S_t[x,\cdot])_{t\ge0}$ and $(W_{\reg,t}[x,\cdot])_{t\ge0}$ the processes obtained
       upon choosing $b=\mathrm{ev}$ in the construction of $(S_t(x))_{t\ge0}$ and $(W_{\reg,t}(\vt;x))_{t\ge0}$, respectively. 
       Thanks to the assumptions on $A$ and $V$ we may assume that 
	$[0,\infty)\times\RR^{d}\ni(t,x)\mapsto S_{t}[x,\gamma]$ is continuous at every $\gamma\in\Omega_{\textrm{w}}$,
	so that in particular $(x,\gamma)\mapsto S_{t}[x,\gamma]$ is product measurable for all $t\ge0$.
	
	Now fix $x\in\RR^d$ and $t\ge s\ge0$. Employing the notation \cref{def:hatb} we then find
	\begin{align*}
	S_t(x)=S_{t-s}[b_s^x,\hat{b}]+S_s(x), \ \
	\text{$\PP$-a.s., and} \ \ W_{\reg,s,t}(\vt;b_s^x)=W_{\reg,t-s}[b_s^x,\hat{b}] \ \ \text{on $\Omega$.}
	\end{align*}
	In fact, the first relation is standard and the second one is quite obvious as the involved operator-valued processes
	are defined pathwise. These remarks in conjunction with the flow relation \eqref{flowW}
	and the pull-out property of conditional expectations imply
	\begin{align*}
	\EE^{\fr{F}_s}[\eul^{-\ol{S}_t(x)}W_{\UV,t}(\vt;x)^*\Psi(b_t^x)]
	&=\eul^{-\ol{S}_s(x)}W_{\reg,s}(\vt;x)^*\EE^{\fr{F}_s}[\Theta(b_s^x,\hat{b})],\quad\text{$\PP$-a.s.,}
	\end{align*}
	where $\Theta(y,\gamma)\coloneq \eul^{-\ol{S}_{t-s}[y,\gamma]}W_{\reg,t-s}[y,\gamma]^*
	\Psi(y+\mathrm{ev}_{t-s}(\gamma))$ defines a bounded product measurable $\Fock$-valued function on
	$\RR^d\times\Omega_{\textrm{w}}$. Since $b_s^x$ is $\fr{F}_s$-measurable and $\hat{b}$ is
	$\fr{F}_s$-independent, the ``useful rule'' for conditional expectations $\PP$-a.s. yields
	\begin{align*}
		\EE^{\fr{F}_s}[\Theta(b_s^x,\hat{b})]&=
		\EE[\Theta(y,\hat{b})]\big|_{y=b_s^x}=\EE[\Theta(y,b)]\big|_{y=b_s^x}
		=(T_{\reg,t-s}\Psi)(b_s^x).
	\end{align*}
	Here the second equality holds since $\hat{b}$ and $b$ have the same distribution. For each $y\in\RR^d$,
	we used in the third equality that $S_{t-s}[y,b]=S_{t-s}(y)$, $\PP$-a.s., and
	$W_{\reg,t-s}[y,b]=W_{\reg,t-s}(\vt;y)$ on $\Omega$.
	Altogether these remarks prove \eqref{Mark1}. 
\end{proof}

\begin{cor}\label{prop:semigroup}
	Let $p\in[1,\infty]$, $\Psi\in L^p(\RR^{d},\Fock)$ and $t\ge s\ge0$. Then 
	\begin{align*}
	T_{\reg,t}\Psi=T_{\reg,s}T_{\reg,t-s}\Psi.
	\end{align*}
\end{cor}

\begin{proof}
	By \cref{LpbdTreg} every $T_{\reg,r}$ with $r\ge0$ is bounded on $L^p(\RR^{d},\Fock)$. 
	Hence, also for $p<\infty$, we may assume in addition that $\Psi$ is bounded, 
	by density of $L^p(\RR^d,\Fock)\cap L^\infty(\RR^d,\Fock)$ in $L^p(\RR^d,\Fock)$.
	The asserted identity then follows by taking expectations in \eqref{Mark1}. 
\end{proof}

\subsection{Strong continuity}

By our next \lcnamecref{prop:cont}, the semigroup $(T_{\reg,t})_{t\ge0}$ is strongly continuous, so that we can study
its generator in the next \lcnamecref{sec:FKproof}.

\begin{prop}\label{prop:cont}
Assume that $V\in C_b(\RR^{d},\RR)$ and $A\in C_b^1(\RR^{d},\RR^d)$.
	Let $p\in[1,\infty)$. Then $(T_{\reg,t})_{t\ge0}$ seen as a semigroup on $L^p(\RR^{d},\Fock)$ is strongly continuous.
\end{prop}
\begin{proof}
	By the semigroup relation and \eqref{LpbdTreg} it suffices to show that
	$T_{\UV,t}\Psi\xrightarrow{t\downarrow0}\Psi$ in $L^p(\RR^{d},\Fock)$
	for every $\Psi$ in a total subset of $L^p(\RR^{d},\Fock)$.
	Thus, we only consider $\Psi\coloneq \rho\expv{f}$ with $\rho\in C_0(\RR^{d})$ and $f\in\HP$.
	By Minkowski's inequality, $\|T_{\reg,t}\Psi-\Psi\|_p\le\mc{N}_1(t)+\mc{N}_2(t)$, $t>0$, with
	\begin{align*}
		\mc{N}_1(t)^p&\coloneq\int_{\RR^{d}}\big\|\EE[\eul^{-\ol{S}_t(x)}W_{\reg,t}(x)^*-1]\expv{f}\big\|_{\Fock}^p|\rho(x)|^p\Id x,
		\\
		\mc{N}_2(t)^p&\coloneq\int_{\RR^{d}}\EE\big[
		\|\eul^{-\ol{S}_t(x)}W_{\reg,t}(x)^*\expv{f}\|_{\Fock}|\rho(b_t^x)-\rho(x)|\big]^p\Id x.
	\end{align*}
	In view of \eqref{nbWreg}, 
	$\|\eul^{-\ol{S}_t(x)}W_{\reg,t}(x)^*\expv{f}\|_{\Fock}\le \eul^{(c_{\vt}+\|V\|_\infty)t}\|\expv{f}\|_{\Fock}$ 
	on $\Omega$ for all $t\ge0$ and $x\in\RR^{d}$.
	Hence, standard estimations employing that $\rho$ is compactly supported and uniformly continuous show that
	$\mc{N}_2(t)\to0$, $t\downarrow0$. Fix $x\in\RR^d$ for the moment. Analogously to \eqref{def:Yft} we then find
	\begin{align}\label{Wregstarexpv}
		W_{\reg,t}(x)^*\epsilon(f)&=\eul^{\ol{u_{\reg,t}(x)}-\langle U_{\reg,t}^{+}(x)|f\rangle_{\HP}}
		\expv{\eul^{-t}f-U_{\reg,t}^{-}(x)},\quad t\ge0.
	\end{align}
	The process defined by the right hand side of \cref{Wregstarexpv} is continuous and the same holds for $S(x)$, so that 
	$\eul^{-S_t(x)}\to1$, $t\downarrow0$, on $\Omega$.
	Thus, by dominated convergence,
	$\EE[\eul^{-\ol{S}_t(x)}W_{\reg,t}(x)^*\expv{f}-\expv{f}]\to0$, $t\downarrow0$.
	Invoking the dominated convergence theorem once more, we deduce that $\mc{N}_1(t)\to0$ as $t\downarrow0$.
\end{proof}

\subsection{Proof of the Feynman--Kac formula for regular coefficients}

By means of the stochastic differential equation proven in \cref{ssec:Markov},
we shall now verify in the case $\Geb=\RR^d$ that $H(\vt)$ generates 
$(T_{\reg,t})_{t\ge0}$ seen as a semigroup on $L^2(\RR^{d},\Fock)$.
\begin{prop}\label{prop:FKUVregAV}
	Assume that $\Geb=\RR^d$,  $V\in C_b(\RR^{d},\RR)$ and $A\in C^1_b(\RR^{d},\RR^{d\nu})$. Then
	$\eul^{-tH(\vt)}=T_{\reg,t}$ for all $t\ge0$.
\end{prop}
\begin{proof}
	By \cref{prop:semigroup,prop:cont}, we know that $(T_{\reg,t})_{t\ge 0}$ is a strongly continuous semigroup of 
	bounded operators on $L^2(\RR^{d},\Fock)$ and hence has a closed generator, which we denote by $G$.
	
	Under the present assumptions on $V$ and $A$, we know from \cite[Remark~5.8 \& Example~6.4]{Matte.2017}
	that $H(\vt)$ is essentially selfadjoint on $\mathrm{span}\{\rho\expv{f}|\,\rho\in C_0^\infty(\RR^{d}),\,f\in\HP\}$.
	Pick $\rho\in C_0^\infty(\RR^{d})$ and $f\in\HP$. Using the notation \cref{eq:Ht}, we have, for a.e. $x\in\RR^{d}$,
	\begin{align}\nonumber
		&(H(\vt)\rho\expv{f})(x)
		\\\label{eq:HUVexpv}
		&=-\frac{1}{2}(\Delta\rho)(x)\expv{f}+\ii A(x)\cdot\nabla\rho(x)\expv{f}+\rho(x)\wt{H}^{A,V}(x)^*\expv{f}.
	\end{align}
	Let $x\in\RR^{d}$ and $\phi\in\mathrm{span}\{\expv{h}|\,h\in\HP\}$. 
	Then \cref{lem:SDEX} in conjunction with It\^{o}'s formula 
	\begin{align*}
		\ol{\rho}(b^x_t)&=\ol{\rho}(x)+\int_0^t\nabla\ol{\rho}(b_s^x)\Id b_s+\frac{1}{2}\int_0^t\Delta\ol{\rho}(b_s^x)\Id s,
		\quad t\ge0,\;\text{$\PP$-a.s.},
	\end{align*}
	and It\^{o}'s product formula implies
	\begin{align*}
		&\langle\expv{f}|\eul^{-S_t(x)}W_{\reg,t}(x)\phi\rangle_{\Fock}\ol{\rho}(b_t^x)
		-\langle\expv{f}|\phi\rangle_{\Fock}\ol{\rho}(x)
		\\
		&=-\int_0^t\langle (H(\vt)\rho\expv{f})(b_s^x)|\eul^{-S_s(x)}W_{\reg,s}(x)\phi\rangle_{\Fock}\Id s
		\\
		&\quad+\int_0^t\langle \expv{f}|\eul^{-S_s(x)}W_{\reg,s}(x)\phi\rangle_{\Fock}(\nabla+\ii A(b_s^x))\ol{\rho}(b_s^x)\Id b_s,
		\quad t\ge0,\;\text{$\PP$-a.s.}
	\end{align*}
	In view of \cref{nbWreg}, 
	the stochastic integral in the last line is a martingale and hence drops out upon taking expectations. This yields
	\begin{align}\label{genG1}
		\langle (T_{\reg,t}\rho\expv{f})(x) | \phi\rangle_{\Fock}
		 - \langle \rho(x)\expv f | \phi \rangle_{\Fock} & = 
		 -\int_0^t \langle (T_{\reg,s}H(\vt)\rho\expv{f})(x) | \phi \rangle_{\Fock} \Id s. 
	\end{align}
	Since $\phi$ can be chosen in a dense subset of $\Fock$, \cref{genG1} extends to all $\phi\in\Fock$.
	In fact, to pass to general $\phi$ under the integral in \cref{genG1} we employ dominated convergence
	taking into account that 
	$\|(T_{\reg,s}H(\vt)\rho\expv{f})(x)\|_{\Fock}\le\eul^{(c_{\vt}+\|V\|_\infty)t}\|H(\vt)\rho\expv{f})\|_\infty$ for all $s\in[0,t]$;
	recall \cref{LpbdTreg}.
	Setting $\phi = \Phi(x)$ in \cref{genG1} for any $\Phi\in L^2(\RR^{d},\Fock)$, integrating with respect to $x$, 
	applying Fubini's theorem and observing that the right hand side of the next identity is well-defined 
	as an $L^2(\RR^d,\Fock)$-valued Bochner--Lebesgue integral, by the continuity of its integrand, we find
	\begin{align*}
		T_{\reg,t}\rho\expv f - \rho \expv f = - \int_0^t T_{\reg,s}H(\vt) \rho \expv f \Id s,\quad t\ge0.
	\end{align*}
	This shows that $\rho \expv f \in\dom(G)$ and $G \rho \expv f = H(\vt)\rho \expv f$.
	By the observation prior to \cref{eq:HUVexpv}, this implies $H(\vt)\subset G$.
	By the bound \cref{LpbdTreg} and the Hille--Yosida theorem, $(-\infty,-c_{\vt}-\|V\|_\infty)$ is contained in the resolvent set of
	$G$ and in particular the intersection of the resolvent sets of $H(\vt)$ and $G$ is non-empty. 
	Combined with the second resolvent identity, this implies $G = H(\vt)$, which finishes the proof.
\end{proof}

In the next corollary we implicitly
employ a standard procedure due to Simon \cite{Simon.1978b} 
(see also \cite{BroderixHundertmarkLeschke.2000}) to infer Feynman--Kac formulas
on proper subsets $\Geb$ of $\RR^d$ from the previous proposition.
The procedure from \cite{Simon.1978b} has been adapted to models in non-relativistic quantum field theory
in \cite{Matte.2021}. We refer to the latter two papers for any further explanations of Simon's procedure. Here
we shall merely argue that technical criteria given in \cite{Matte.2021} are satisfied in the present setting.

\begin{cor}\label{cor:FKUVregAV}
Assume that $V\in C_b(\RR^{d},\RR)$ and $A\in C^1_b(\RR^{d},\RR^{d})$. Consider the
Hamiltonian $H(\vt)$ on a general open subset $\Geb\subset\RR^d$. Then \cref{eq:FKmain} holds
for all $t\ge0$ and $\Psi\in L^2(\Geb,\Fock)$.
\end{cor}

\begin{proof}
Let $\fr{h}_{\RR^d}(\vt)$ be the polaron form 
on $\RR^d$ and $\fr{h}_{\Geb}(\vt)$ the one on $\Geb$.
To infer \cref{eq:FKmain} from \cref{prop:FKUVregAV} we only have verify that these quadratic forms
satisfy certain criteria permitting to apply \cite[Lemma~3.4]{Matte.2021}.

Let $K_n$, $n\in\NN$, be compact sets exhausting $\Geb$ in the sense that $K_n\subset\mathring{K}_{n+1}$
for all $n\in\NN$ and $\bigcup_{n=1}^\infty K_n=\Geb$. Further, let $\chi_n\in C_0^\infty(\RR^d)$
with $0\le \chi_n\le1$ satisfy $\chi_n=1$ on $K_n$ and $\chi_n=0$ on $K_{n+1}^c$. Define
$Y:\RR^d\to[0,\infty]$ by 
$Y(x)\coloneq\dist(x,\Geb^c)^{-3}+\sum_{n=1}^\infty|\nabla\chi_n(x)|^2$ for all $x\in\Geb$ and
$Y(x)\coloneq\infty$ for all $x\in\Geb^c$.
Set $\scr{D}_Y\coloneq \{\Psi\in\dom(\fr{h}_{\RR^d}(\vt))|\,Y\|\Psi\|^2_{\Fock}\in L^1(\RR^d)\}$.
Then $\Psi=0$ a.e. on $\Geb^c$ for every $\Psi\in\scr{D}_Y$, whence we can interpret $\scr{D}_Y$
as a subspace of $L^2(\Geb,\Fock)$ in the canonical fashion. By virtue of \cite[Lemma~3.4]{Matte.2021}
it then suffices to verify:
\begin{enumerate}
\item[(a)] $\scr{D}_Y\subset\dom(\fr{h}_{\Geb}(\vt))$.
\item[(b)] The closure of $\scr{D}_Y$ with respect to the norm associated with $\fr{h}_{\Geb}(\vt)$ 
is equal to $\dom(\fr{h}_{\Geb}(\vt))$.
\item[(c)] $\fr{h}_{\Geb}(\vt)[\Psi]=\fr{h}_{\RR^d}(\vt)[\Psi]$ for all $\Psi\in\scr{D}_Y$.
\end{enumerate}
To verify (a) and (b) we recall that $\dom(\fr{h}_{\Geb}(\vt))=\dom(\fr{q}^{\min})$ and that the
norms associated with $\fr{h}_{\Geb}(\vt)$ and $\fr{q}^{\min}$ are equivalent. In other words,
to prove (a) and (b) we can assume without loss of generality that $V_-=0$ and $\vt=0$.
But then (a) and (b) are special cases of \cite[Proposition~5.13]{Matte.2021}.
Furthermore, $\fr{h}_{\Geb}(\vt)[\Psi]=\fr{h}_{\RR^d}(\vt)[\Psi]$ obviously holds for all 
$\Psi\in\mathrm{span}\{f\phi|\,f\in C_0^\infty(\Geb),\,\phi\in\fdom(N) \}$, i.e., for all $\Psi$
in a core for $\fr{h}_{\Geb}(\vt)[\Psi]$. By (a) and the closedness of $\fr{h}_{\RR^d}(\vt)$,
this entails (c).
\end{proof}


\section{Bounds on the interaction processes}\label{sec:Upm}

\noindent
The objective of this \lcnamecref{sec:Upm} is to prove \cref{lemUregU} as well as
the following theorem on the $\HP$-valued processes defined in \cref{defUinftyminus,defUinftyplus}.
Readers who wish to jump over technical details 
can move on to the next section after reading the theorem.

\begin{thm}\label{lemexpmomentU}
Assume that $\Geb$ fulfills \cref{tailbound}. 
	Let $p>0$ and define
	\begin{align}\label{def:sp}
		\sigma_p \coloneq \inf\big\{\sigma \ge 2 \big| \ 32\sqrt{2p}L_1(v_{\sigma})\le 1\wedge\sqrt{4C_\Geb}\big\} < \infty.
	\end{align}
	Then there exists $c_{\Geb}\in[1,\infty)$, solely depending on $\Geb$, such that
	\begin{align}\label{expmomentUminus}
		\sup_{t>0}\sup_{x\in\Geb}\EE\Big[\chi_{\{t<\tau_{\Geb}(x)\}}\eul^{p\|U_{\sigma,t}^\pm(x)\|_{\HP}^2/(1\wedge t)}\Big]
		&\le c_{\Geb}\eul^{p\sup_{y\in\Geb}\|\tilde{v}_{\sigma_p,y}\|_\HP^2},
		\quad\sigma\in[2,\infty).
	\end{align}
	Furthermore, let $v^1,v^2,\ldots$ be coupling functions fulfilling the same hypotheses as $v$ such
	that $L_1(v^n-v)\to0$ as $n\to\infty$. Define $U_{\sigma,t}^{n,\pm}(x)$ by putting $v^n$ in place of $v$ in
	\cref{defUinftyminus,defUinftyplus}. Then, for all $p>0$ and $\sigma\in[2,\infty)$,
	\begin{align}\label{convUpmLp}
		\lim_{n\to\infty}\sup_{t>0}\sup_{x\in\Geb}
		\EE\big[(1+t^{-1})^{p/2}\chi_{\{t<\tau_{\Geb}(x)\}}\|U_{\sigma,t}^{n,\pm}(x)-U_{\sigma,t}^\pm(x)\|_{\HP}^p\big]&=0.
	\end{align}
\end{thm}

In the remainder of this \lcnamecref{sec:Upm}, we first
discuss the martingale term from \cref{defUinftyminus,defUinftyplus}
thus finishing the proof of \cref{lemUregU} (\cref{ssec:Upmmartingale}) and then prove 
the above theorem (\cref{ssec:expmom}).
We will employ the exponential moment bound in the following remark multiple times:

\begin{rem}\label{rem:var}
	Assume that $(Z_t)_{t\ge0}$ is a predictable $\RR^d$-valued process such that $\int_0^t\EE[|Z_s|^2]\Id s<\infty$
	for all $t\ge0$. Then $M_t\coloneq \int_0^tZ_s\Id b_s$, $t\ge0$, defines a continuous real-valued
	$L^2$-martingale with quadratic variation given by 
	\begin{align*}
		[M]_t=\int_0^t|Z_s|^2\Id s,\quad t\ge0.
	\end{align*}
	We shall often use the bound (see, e.g., \cite[Remark~3.3]{MatteMoller.2018}): 
	\begin{align}\label{expmartest2}
		\EE\bigg[\sup_{s\in[0,t]}\eul^{M_s}\bigg]&\le (1+\pi)^{1/2}\EE\big[\eul^{4[M]_t}\big]^{1/2},\quad t\ge0.
	\end{align}
	This bound also applies to stopped versions of $M$, since
	$M_{\tau\wedge t}\coloneq \int_0^t\chi_{\{s\le\tau\}}Z_s\Id b_s$, $t\ge0$, holds $\PP$-a.s.
	for every stopping time $\tau:\Omega\to[0,\infty]$, where $(\chi_{\{t\le\tau\}}Z_t)_{t\ge0}$ is again predictable.
\end{rem}

\subsection{Discussion of the martingale part}\label{ssec:Upmmartingale}
In this \lcnamecref{ssec:Upmmartingale} we discuss the stochastic integral processes 
$M^\pm_\sigma(x)=(M^\pm_{\sigma,t}(x))_{t\ge0}$ given by \cref{def:Msigmapmx}.
Part~(iii) of the next lemma will in particular complete the proof of \cref{lemUregU}.

\begin{lem}\label{lemItoM}
	Let $\sigma\in[2,\infty)$.
	Then the following holds:
	\begin{enumerate}
		\item[{\rm(i)}] For all $x\in\RR^d$,
		$M_{\sigma}^\pm(x)$ is a continuous $\HP$-valued $L^2$-martingale and its quadratic variation satisfies
		\begin{align}\label{qvMkappa}
			[M_\sigma^\pm(x)]_t&\le\frac{1}{2}|1-\eul^{\pm 2t}|\sup_{y\in\Geb}\|\alpha^\pm_{\sigma,y}\|_{\HP}^2,
			\quad t\ge0.
		\end{align}
		\item[{\rm(ii)}] For every $p>0$, we find solely $p$-dependent $c_p,c_p'\in(0,\infty)$ such that
		\begin{align*}
		\sup_{t>0}
			\sup_{x\in\RR^d}\EE\big[(1+t^{-1})^{p/2}\|\eul^{(\mp t)\wedge0}M_{\sigma,t}^\pm(x)\|_{\HP}^p\big]
			\le c_p\sup_{y\in\Geb}\|\alpha^\pm_{\sigma,y}\|_{\HP}^p\le c_p'L_1(v_\sigma)^p.
		\end{align*}
		\item[{\rm(iii)}] Assume in addition that $(1+\lambda)v\in \mc{L}^\infty(\Geb,\HP)$. 
		Let $x\in\Geb$ and abbreviate
		\begin{align*}
		I_{\sigma,t}^\pm(x)\coloneq\int_0^t\eul^{\pm s}v_{\sigma,b_s^x}\Id s,\quad t\ge0.
		\end{align*}
		Then, $\PP$-a.s., for all $t\ge0$,
		\begin{align}\label{ItoM}
			I_{\sigma,t}^\pm(x)
			&=\beta^{\pm}_{\sigma,x}-\eul^{\pm t}\beta^\pm_{\sigma,b_t^x}+ M_{\sigma,t}^\pm(x)
			\quad\text{on $\{t<\tau_{\Geb}(x)\}$.}
		\end{align}
	\end{enumerate}
\end{lem}

\begin{proof}
	(i):
	 The right hand side of \eqref{qvMkappa} is an upper bound on
	$J_t\coloneq\int_0^t\|\eul^{\pm s}\alpha^\pm_{\sigma,b_s^x}\|_{\HP}^2\Id s$.
	Since $\EE[J_t]<\infty$, $t\ge0$, we know that $M_\sigma^\pm(x)$ is a well-defined, continuous
	$L^2$-martingale with quadratic variation $(J_t)_{t\ge0}$. 
	
	\smallskip
	
	\noindent(ii): Employing a Burkholder inequality (see, e.g., \cite[Theorem~4.36]{DaPratoZabczyk.2014}),
	we find a solely $p$-dependent $c_p'\in(0,\infty)$ such that
	\begin{align}\label{BurkM}
		\EE\Big[\sup_{s\in[0,t]}\|M_{\sigma,s}^\pm(x)\|_{\HP}^p\Big]&\le c_p'\EE\big[[M_\sigma^\pm(x)]_t^{p/2}\big],
	\quad t\ge0,
	\end{align}
	which together with (i) implies the first asserted bound. The second one follows from  
	(b) in \cref{sssec:coupling} and \cref{defbetainfty}, since $\lambda-1\ge1$ holds on $\{v_{\sigma,x}\not=0\}$.
	
	\smallskip
	
	\noindent(iii): We choose the
	open subsets $\Geb_n\subset\Geb$, $n\in\NN$, 
	and corresponding exit times $\tau_n(x)$ as at the end of \cref{ssecStrat}.
	We further pick $\rho_n\in C^\infty(\RR^d)$ such that 
	$\rho_n=1$ on $\ol{\Geb}_n$ and $\rho_n=0$ on $\Geb_{n+1}^c$ for all $n\in\NN$.
	The additional assumption on $v$ and \cref{lem:green1} ensure that $v$ satisfies the hypotheses of \cref{lembetaC2HP}
	with $\ell=2$. Thus, by the latter lemma, the maps 
	$x\mapsto f_n^\pm(x)\coloneq \rho_n(x)\beta_{\sigma,x}^\pm$ belong to $C^2(\RR^d,\HP)$. 
	Together with Assumption (a) in \cref{sssec:coupling}, \cref{lembetaC2HP} further entails,
	with a Laplacian acting on $\HP$-valued functions,
	\begin{align}\label{vbetarel}
	\bigg(\mp 1-\frac{1}{2}\Delta\bigg)f_n^\pm(x)=(\mp1+\lambda)\beta_{\sigma,x}^\pm=
	v_{\sigma,x},\quad x\in \ol{\Geb}_n,\,n\in\NN.
	\end{align}
	Now let $x\in\Geb$ and pick some $n_0\in\NN$ such that $x\in \Geb_{n_0}$.
	Then It\^{o}'s formula (see, e.g., \cite[Theorem~4.32]{DaPratoZabczyk.2014}) $\PP$-a.s. yields
	\begin{align*}
	\eul^{\pm t}f_n^\pm(b_t^x)&=f_n^\pm(x)+\int_0^t\eul^{\pm s}\nabla f_n^\pm(b_s^x)\Id b_s
	-\int_0^t\eul^{\pm s}\bigg(\mp 1-\frac{1}{2}\Delta\bigg)f_n^\pm(b_s^x)\Id s,
	\end{align*}
	for all $t\ge0$ and integers $n\ge n_0$. Putting $t\wedge\tau_n(x)$ in place of $t$, using
	that $\rho_n(b_s^x)=1$ for all $s\in[0,\tau_n(s)]$ and taking \cref{vbetarel} into account, we $\PP$-a.s. find
	\begin{align}\nonumber
	\eul^{\pm(t\wedge\tau_n(x))}\beta^\pm_{\sigma,b_{t\wedge\tau_n(x)}^x}&=\beta^\pm_{\sigma,x}
	+\int_0^t\chi_{\{\tau_n(x)\ge s\}}\eul^{\pm s}\alpha^\pm_{\sigma,b_s^x}\Id b_s
	\\\label{zerletaeminustbeta}
	&\quad
	-\int_0^{t}\chi_{\{\tau_n(x)\ge s\}}\eul^{\pm s}v_{\sigma,b_s^x}\Id s,
	\end{align}
	for all $t\ge0$ and integers $n\ge n_0$,
	where we also used a standard stopping rule for stochastic integrals.
	Since $\tau_n(x)\uparrow\tau_{\Geb}(x)$ as $n\to\infty$, 
	this proves (iii).
\end{proof}

\subsection{Convergence and exponential moment bound}\label{ssec:expmom}

Now we move to the proof of Theorem~\ref{lemexpmomentU}, where it is convenient to use the quantities 
\begin{align}\label{def:gs}
	g_\sigma \coloneq\max\bigg\{\sup_{y\in\Geb}\|\beta_{\sigma,y}^\pm\|_{\HP}\,,\,
	\sup_{y\in\Geb}\|\nabla_y\beta_{\sigma,y}^\pm\|_{\HP}\bigg\}\le \sqrt{2}\cdot4L_1(v_{\sigma}),
	\quad \sigma\in[2,\infty),
\end{align}
so that $g_\sigma\to0$ as $\sigma\to\infty$. We further abbreviate
\begin{align*}
	\bbeta_{\sigma,t}^\pm(x)\coloneq 
	\eul^{(\mp t)\wedge 0}\beta_{\sigma,x}^\pm-\eul^{(\pm t)\wedge 0}\beta_{\sigma,b_t^x}^\pm.
\end{align*}
for all $t\ge0$, $x\in\Geb$ and $\sigma\in[2,\infty)$. Then the definitions
\cref{defUinftyminus,defUinftyplus} correspond to the two cases in
\begin{align}\label{forUinftypm}
	U_{\sigma,t}^\pm(x)= U_{\reg,t}^\pm(\tilde{v}_\sigma;x)
	+\bbeta_{\sigma,t}^\pm(x)+\eul^{(\mp t)\wedge0}M_{\sigma,t}^\pm(x),\quad t\ge0.
\end{align}
Recalling the definition of $d_\Geb$ below \cref{def:W} and
taking \cref{lembetaC1HP} into account we observe that
	\begin{align}
		\|\bbeta_{\sigma,t}^\pm(x)\|_\HP   \le |1-\eul^{-t}|g_{\sigma} + d_\Geb(b_t^x,x)g_\sigma,
		\quad\text{where $|1-\eul^{-t}|\le t^{1/2}$}.
		\label{bdgammapm}
	\end{align}
Let $\Theta:\RR\to\RR$ be non-decreasing and right-continuous with $\Theta(0)=0$ and denote the associated
Stieltjes-Borel measure by $\theta:\fr{B}(\RR)\to[0,\infty]$. Using $\Theta(0)=0$ in the first relation
and \cref{tailbound} in the last one, we then find
\begin{align}\label{lagkage1}
\begin{aligned}
	\EE\big[\chi_{\{t<\tau_{\Geb}(x)\}}\Theta(d_\Geb(b_t^x,x))\big]
	&=\EE\big[\Theta(\chi_{\{t<\tau_{\Geb}(x)\}}d_\Geb(b_t^x,x))\big]
	\\
	&=\int_{(0,\infty)}\PP\big[\chi_{\{t<\tau_{\Geb}(x)\}}d_\Geb(b_t^x,x)\ge s\big]\Id\theta(s)
	\\
	&\le a_\Geb \int_{(0,\infty)}\eul^{-C_\Geb s^2/t}\Id\theta(s),\quad x\in\Geb,\,t>0.
\end{aligned}
\end{align}
When $\Theta(s)=(g_\sigma s)^p$ for all $s\ge0$ and some $p>0$, this together with \cref{bdgammapm} yields
	\begin{align}\label{eq:bconv}
		\begin{aligned}
			\sup_{x\in\Geb}\EE\big[  \chi_{\{t<\tau_{\Geb}(x)\}}  \|\bbeta_{\sigma,t}^\pm(x)\|_{\HP}^p\big]
			&\le c_{\Geb,p}(1\wedge t)^{p/2}L_1(v_\sigma)^p,\quad t\ge0.
		\end{aligned}
	\end{align}
	Here we also used \cref{def:gs} for $t\ge1$ and $c_{\Geb,p}>0$ depends only on $\Geb$ and $p$.

\begin{proof}[Proof of the convergence relation \eqref{convUpmLp}.]
	Let $p>0$. Then \eqref{convUpmLp} follows directly from
	\cref{Usigmabd,forUinftypm,eq:bconv} as well as \cref{lemItoM}(ii),
	all applied to the coupling function $v^n-v$.
\end{proof}

\begin{proof}[Proof of the exponential moment bound \eqref{expmomentUminus}.]
	We pick $p>0$ and $\sigma\in[2,\infty)$. In the last step of this proof,
	$\sigma$ will be chosen sufficiently large depending on $p$.
	
	\smallskip
	\noindent
	{\em Step~1.} 
	Let $x\in\Geb$. Combining the trivial bound \cref{Usigmabd} with
	\begin{align*}
	\|U_{\sigma,t}^\pm(x)\|_{\HP}^2&=\|U_{\reg,t}^\pm(\tilde{v}_\sigma;x)\|_{\HP}^2
	+\|\bbeta_{\sigma,t}^\pm(x) - \eul^{(\mp t)\wedge0}M_{\sigma,t}^\pm(x)\|_{\HP}^2,	
	\end{align*} 
	 cf. \cref{forUinftypm}, and the Cauchy--Schwarz inequality, we find
	\begin{align}\label{adam7685}
		\begin{aligned}
			&\EE\big[\chi_{\{t<\tau_{\Geb}(x)\}}\eul^{p\|U_{\sigma,t}^\pm(x)\|_{\HP}^2/(t\wedge1)}\big]	
			\\
			&\le \eul^{p\sup_{y\in\Geb}\|\tilde{v}_{\sigma,y}\|_\HP^2}\\
			& \qquad \times
			\EE\big[\chi_{\{t<\tau_{\Geb}(x)\}}\eul^{4p\|\bbeta_{\sigma,t}^\pm(x)\|_{\HP}^2/(t\wedge1)}\big]^{1/2}
			\EE\big[\eul^{4p\eul^{(\mp2t)\wedge0}\|M_{\sigma,t}^\pm(x)\|_{\HP}^2/(t\wedge1)}\big]^{1/2},
		\end{aligned}
	\end{align}
	for all $t>0$.
In the next two steps we derive bounds on the two expectations on the right hand side of \cref{adam7685}.  

\smallskip
	\noindent
	{\em Step~2.} 
	Employing \cref{bdgammapm} first and
	choosing $\Theta(s)=\eul^{8p (g_\sigma s)^2/t}-1$, $s\ge0$, in \cref{lagkage1} we find
		\begin{align*}
			\EE\big[\chi_{\{t<\tau_{\Geb}(x)\}} \eul^{4p\|\bbeta_{\sigma,t}^\pm(x)\|_\HP^2/t}\big]
			& \le 
			16pg_\sigma^2a_\Geb\eul^{8p g_\sigma^2}\int_{0}^\infty \eul^{8p(g_\sigma s)^2/t-C_\Geb s^2/t}\frac{s\,\Id s}{t},
		\end{align*}
		for all $t\in(0,1]$,
		which together with \cref{def:gs} (applied when $t>1$) yields the implication
		\begin{align}\label{betabd2}
			16pg_\sigma^2\le C_\Geb
			\ \Rightarrow\ 
			\sup_{t>0}\sup_{x\in\Geb}\EE\big[\chi_{\{t<\tau_{\Geb}(x)\}}\eul^{4p\|\bbeta_{\sigma,t}^\pm(x)\|_{\HP}^2/(t\wedge1)}\big]
			&\le a_\Geb\eul^{C_\Geb}.
		\end{align}
	\noindent
	{\em Step~3.} 
	Let $x\in\Geb$.
	Employing It\^{o}'s formula and \eqref{qvMkappa} we find
	\begin{align}\label{susanne1}
		\|M_{\sigma,t}^\pm(x)\|_{\HP}^2
		&\le2N_{\sigma,t}^\pm(x)+\frac{1}{2}|1-\eul^{\pm2t}|g_\sigma^2,\quad t\ge0,
	\end{align}
	$\PP$-a.s., with the continuous local martingale $N_{\sigma}^\pm(x)$ defined by
	\begin{align*}
		N_{\sigma,t}^\pm(x)
		&\coloneq\int_0^t\Re\langle M_{\sigma,s}^\pm(x)|\eul^{\pm s}
		\alpha_{\sigma,b_s^x}^\pm\rangle_{\HP}\Id b_s,\quad t\ge0.
	\end{align*}
	The $\PP$-a.s. bound
	\begin{align}\nonumber
		[N_{\sigma}^\pm(x)]_t&=
		\int_0^t|\Re\langle M_{\sigma,s}^\pm(x)|\eul^{\pm s}
		\alpha_{\sigma,b_s^x}^\pm\rangle_{\HP}|^2\Id s
		\\\label{susanne2}
		&\le \frac{1}{2}|1-\eul^{\pm2t}|g_\sigma^2
		\sup_{s\in[0,t]}\|M_{\sigma,s}^\pm(x)\|_{\HP}^2,
		\quad t\ge0,
	\end{align}
	in conjunction with \cref{qvMkappa,BurkM} reveals that $N_{\sigma}^\pm(x)$ actually is a martingale. 
	Next, we define an increasing sequence of bounded stopping times $\tau_n^\pm(x):\Omega\to[0,\infty)$ 
	such that $\tau_n^\pm(x)\uparrow\infty$, $n\to\infty$, by
	\begin{align*}
	\tau_n^\pm(x)&\coloneq n\wedge \inf\big\{t\ge0\,\big|\;[N_{\sigma}^\pm(x)]_t\ge n \big\},\quad n\in\NN.
	\end{align*}
	Then the stopped processes given by
	$Q^{\pm,n}_{\sigma,t}(x)\coloneq N_{\sigma,\smash{\tau_n^\pm(x)}\wedge t}^\pm(x)$, $t\ge0$, are 
	martingales as well. Their quadratic variations $\PP$-a.s. satisfy 
	$[Q^{\pm,n}_{\sigma}(x)]_t=[N_{\sigma}^{\pm}(x)]_{\tau_n^\pm(x)\wedge t}$ for all $t\ge0$.
	Invoking \cref{expmartest2,susanne1}, setting $c_0\coloneq(1+\pi)^{1/2}$ and observing 
	$|1-\eul^{\pm2(\tau_n^\pm(x)\wedge t)}|\le|1-\eul^{\pm2t}|$, we thus find 
	\begin{align*}
		&\EE\bigg[\sup_{s\in[0,\tau_n^\pm(x)\wedge t]}\eul^{a\|M_{\sigma,s}^\pm(x)\|_{\HP}^2}\bigg]
		\\
		&\le \eul^{|1-\eul^{\pm2t}|ag_\sigma^2/2}\EE\bigg[\sup_{s\in[0,t]}
		\eul^{2aQ_{\sigma,s}^{\pm,n}(x)}\bigg]
		\\
		&\le c_0\eul^{|1-\eul^{\pm2t}|ag_\sigma^2/2}
		\EE\Big[\eul^{4(2a)^2[N_{\sigma}^\pm(x)]_{\tau_n^{\pm}(x)\wedge t}}\Big]^{1/2}
		\\
		&\le c_0\eul^{|1-\eul^{\pm2t}|ag_\sigma^2/2}\EE\bigg[\sup_{s\in[0,\tau_n^\pm(x)\wedge t]}
		\eul^{(8a|1-\eul^{\pm 2t}|g_\sigma^2)a\|M_{\sigma,s}^\pm(x)\|_{\HP}^2}\bigg]^{1/2},
	\end{align*}
	for all $a,t\ge0$ and $n\in\NN$. Here we also used \cref{susanne2} in the last step.
	Since $[N_{\sigma}^\pm(x)]_{\tau_n^{\pm}(x)\wedge t}\le n$
	by the choice of $\tau_n^\pm(x)$, we see that the leftmost expectation in this chain of inequalities is finite.
	This proves the implication
	\begin{align}\label{impl01}
	8a|1-\eul^{\pm 2t}|g_\sigma^2\le1\quad\Rightarrow\quad\sup_{x\in\Geb}\sup_{n\in\NN}
	\EE\bigg[\sup_{s\in[0,\tau_n^\pm(x)\wedge t]}\eul^{a\|M_{\sigma,s}^\pm(x)\|_{\HP}^2}\bigg]^{1/2}\le 
	c_0\eul^{1/16},
	\end{align}
	for fixed $a,t\ge0$.
	In view of \cref{adam7685}, we wish to chose
	$a=4p\eul^{(\mp2t)\wedge0}/(t\wedge 1)$.
	Observing that $\eul^{(\mp2t)\wedge0}|1-\eul^{\pm 2t}|/(t\wedge 1)=(1-\eul^{-2t})/(t\wedge1)\le2$
	for all $t>0$ and applying the monotone convergence theorem for each $x$, we arrive at the implication
	\begin{align}\label{impl02}
	64pg_\sigma^2\le1\quad\Rightarrow\quad\sup_{t>0}\sup_{x\in\Geb}
	\EE\bigg[\sup_{s\in[0,t]}\eul^{4p\eul^{(\mp2t)\wedge0}\|M_{\sigma,s}^\pm(x)\|_{\HP}^2/(t\wedge 1)}\bigg]^{1/2}\le 
	c_0\eul^{1/16}.
	\end{align}
	{\em Step~4.} The remarks in Steps~2 and~3 show that the product of the two expectations on the right hand
	side of \cref{adam7685} is less than or equal to some constant solely depending on $\Geb$ provided that
	$64pg_\sigma^2\le1\wedge(4C_\Geb)$. In view of \cref{def:gs,lemsigmaUpm}
	this implies \cref{expmomentUminus} with $\sigma_p$ as in \cref{def:sp}. 
\end{proof}


\section{Bounds on the complex action}\label{sec:action}

\noindent
In this \lcnamecref{sec:action} we prove \cref{lemuregu} as well as the next theorem.
Again the reader can move on to the next section after reading the theorem if he or she wishes
to jump over technical proofs.

\begin{thm}\label{prop:conveu}
Defining
\begin{align}\label{def:vsp}
\vs_p&\coloneq\inf\big\{\sigma\ge2\big|\,16\sqrt{p}L_1(v_\sigma)\le1\big\},\quad p>0,
\end{align}
we find universal constants $c,c'\in(0,\infty)$ such that, for all $p>0$ and $\sigma\in[2,\infty)$,
	\begin{align}\label{expmbdu}
		\sup_{x\in\Geb}\EE\big[\chi_{\{t<\tau_\Geb(x)\}}|\eul^{u_{\sigma,t}(x)}|^p\big]
		&\le c'\eul^{ct+p\sup_{y\in\Geb}\|\tilde{v}_{\vs_p,y}\|^2t},\quad t\ge0.
	\end{align}
	Moreover, let $v^1,v^2,\ldots$ be coupling functions fulfillung the same hypotheses as $v$ and
	assume that $L_1(v^n-v)\to0$ as $n\to\infty$. Denote by $u_{\sigma,t}^n(x)$ the complex action defined
	by means of $v^n$. Then
	\begin{align}\label{conveu}
		\sup_{t\in[0,r]}\sup_{x\in\Geb}
		\EE\big[\chi_{\{t<\tau_\Geb(x)\}}|\eul^{u_{\sigma,t}^n(x)}-\eul^{u_{\sigma,t}(x)}|^p\big]&\xrightarrow{\;\;n\to\infty\;\;}0,
		\quad p,r>0,\,\sigma\in[2,\infty).
	\end{align}
\end{thm}

\subsection{Regular expressions for the complex action}\label{ssec:regforu}

To establish \cref{lemuregu} in the present \lcnamecref{ssec:regforu} we proceed in two steps that both involve
applications of It\^{o}'s formula. The first step is taken in the next lemma, the second one in the succeeding
proof of \cref{lemuregu}. We shall employ the processes given by
\begin{align}\label{def:cqsigma}
c_{\sigma,t}(x)&\coloneq\langle\beta^-_{\sigma,b_t^x}|U_{\sigma,t}^+(x)\rangle_{\HP},
\qquad
q_{\sigma,t}(x)\coloneq
	\int_0^t\langle\alpha_{\sigma,b_s^x}^-|U_{\sigma,s}^+(x)\rangle_{\HP}\Id b_s,
\end{align}
for all $t\ge0$ and $x\in\RR^d$. Both of them are well-defined under our general hypothesis on $v$,
and $q_{\sigma}(x)$ is an $L^2$-martingale. (The process given by \cref{def:dsigma} is well-defined for
the physically most relevant choices of $v$, but not necessarily under our general hypotheses.)

\begin{lem}\label{lem:Itoaction}
Additionally assume that $(1+\lambda)v\in \mc{L}^\infty(\Geb,\HP)$.
	Let $\sigma\in[2,\infty)$ and 
	$x\in\Geb$. Then, $\PP$-a.s., we know  for all $t\ge0$ that
	\begin{align*}
		u_{\reg,t}(v;x) = u_{\reg,t}(\tilde{v}_\sigma;x) -c_{\sigma,t}(x) + d_{\sigma,t}(x) 
		+q_{\sigma,t}(x)\quad\text{on $\{t<\tau_\Geb(x)\}$,}
	\end{align*}
	with
	\begin{align}\label{def:dsigma}
	d_{\sigma,t}(x)&\coloneq \int_0^t\int_{\mc{K}}\ol{\beta_{\sigma,b_s^x}^-}v_{\sigma,b_s^x}\Id\mu\,\Id s
	=\int_0^t\int_{\mc{K}}\frac{|v_{\sigma,b_s^x}|^2}{1+\lambda}\Id\mu\,\Id s.
	\end{align}
\end{lem}
\begin{proof}
Thanks to \cref{lemItoM}(iii) we $\PP$-a.s. know that 
	$c_{\sigma,t}(x)=\langle\eul^{-t}\beta_{\sigma,b_t^x}^-| I_{\sigma,t}^+(x)\rangle_{\HP}$ on $\{t<\tau_\Geb(x)\}$ for all $t\ge0$.
	Again employing the exit times $\tau_n(x)$ defined in \cref{exittaun}, we further know
	from the proof of \cref{lemItoM}(iii) that
	$(\eul^{-(t\wedge\tau_n(x))}\beta_{\sigma,b_{t\wedge\tau_n(x)}^x}^-)_{t\ge0}$ 
	is a continuous semimartingale that can $\PP$-a.s. be written as in \cref{zerletaeminustbeta}.
	Also writing $I_{\sigma,t\wedge\tau_n(x)}^+(x)=\int_0^t\chi_{\{\tau_n(x)\ge s\}}\eul^sv_{\sigma,b_s^x}\Id s$,
	we infer from It\^{o}'s product rule that, $\PP$-a.s.,
	\begin{align*}
	c_{\sigma,t\wedge\tau_n(x)}(x)
	&=\int_0^t\chi_{\{\tau_n(x)\ge s\}}\langle\eul^{-s}\beta_{\sigma,b_{s}^x}^-| 
	\eul^sv_{\sigma,b_s^x}\rangle_{\HP}\Id s
	\\
	&\quad-\int_0^t\chi_{\{\tau_n(x)\ge s\}}\langle\eul^{-s}v_{\sigma,b_s^x}|I_{\sigma,s}^+(x)\rangle_{\HP}\Id s
	\\
	&\quad+\int_0^t\chi_{\{\tau_n(x)\ge s\}}\langle\eul^{-s}\alpha_{\sigma,b_s^x}^-|I_{\sigma,s}^+(x)\rangle_{\HP}\Id b_s,
	\quad t\ge0,\,n\in\NN.
	\end{align*}
	Since $\eul^{-s}I^+_{\sigma,s}(x)=U_{\reg,s}^+(v_\sigma;x)=\chi_{\{\lambda\ge\sigma\}}U^+_{\sigma,s}(x)$ 
	on $\{s<\tau_{\Geb}(x)\}$ for all $s\ge0$, $\PP$-a.s., the above identity is equivalent to
	\begin{align*}
	&c_{\sigma,t\wedge\tau_n(x)}(x)
	\\
	&=d_{\sigma,t\wedge\tau_n(x)}(x)-u_{\reg,t\wedge\tau_n(x)}(v;x)
	+u_{\reg,t\wedge\tau_n(x)}(\tilde{v}_\sigma;x)
	+q_{\sigma,t\wedge\tau_n(x)}(x).
	\end{align*}
	We conclude by recalling that $\tau_n(x)\uparrow\tau_\Geb(x)$ as $n\to\infty$.
\end{proof}

\begin{proof}[Proof of \cref{lemuregu}.]
Let $x\in\Geb$. Combining \cref{defUinftyplus,def:msigma,def:cqsigma} and recalling that
$\chi_{\{\lambda<\sigma\}}U^+_{\reg,t}(\tilde{v}_\sigma;x)=0$, we $\PP$-a.s. find
\begin{align*}
q_{\sigma,t}(x)&=\int_0^t\langle\alpha_{\sigma,b_s^x}^-|\eul^{-s}\beta^+_{\sigma,x}\rangle_{\HP}\Id b_s
-\int_0^t\langle\alpha_{\sigma,b_s^x}^-|\beta_{\sigma,b_s^x}^+\rangle_{\HP}\Id b_s
+m_{\sigma,t}(x),\quad t\ge0,
\end{align*}
where the first member on the right hand can be written as
\begin{align*}
\int_0^t\langle\alpha_{\sigma,b_s^x}^-|\eul^{-s}\beta^+_{\sigma,x}\rangle_{\HP}\Id b_s
&=\langle M_{\sigma,t}^-(x)|\beta^+_{\sigma,x}\rangle_{\HP}.
\end{align*}

In the next step we again employ the open subsets $\Geb_n$ exausting $\Geb$, the corresponding
first exit times $\tau_n(x)$ and the localization functions $\rho_n$ defined in the proof of \cref{lemItoM}(iii).
As noted in the proof of the latter lemma, we know under the present assumptions 
that $x\mapsto\beta_{\sigma,x}^\pm$ are twice continuously
differentiable on $\Geb$ as $\HP$-valued functions.  
We define $f_n\in C^2(\RR^d,\RR)$ by
 \begin{align*}
 f_n(x)&\coloneq\rho_n(x)\langle\beta_{\sigma,x}^-|\beta_{\sigma,x}^+\rangle_{\HP}, 
 \quad x\in\RR^d,\, n\in\NN.
 \end{align*} 
Using $\Delta_x\beta_x^\pm=-2\lambda\beta_x^\pm$, $x\in\Geb$, we then observe that
 \begin{align*}
\frac{1}{2} \Delta f_n(x)&=\langle\alpha_{\sigma,x}^-|\alpha_{\sigma,x}^+\rangle_{\HP}-
2\int_{\mc{K}}\lambda\ol{\beta_{\sigma,x}^-}\beta^+_{\sigma,x}\Id\mu,
 \quad x\in\Geb_n,\,n\in\NN.
 \end{align*}
 We fix $x\in\Geb$ again and pick $n_0\in\NN$ such that $x\in\Geb_{n_0}$.
Employing It\^{o}'s formula for $f_n(b^x)$, we then deduce that, $\PP$-a.s.,
\begin{align}\nonumber
-&\int_0^{t\wedge\tau_n(x)}\Re\langle\alpha_{\sigma,b_s^x}^-|\beta_{\sigma,b_s^x}^+\rangle_{\HP}\Id b_s
\\\nonumber
&=-\frac{1}{2}\int_{0}^t\chi_{\{s\le\tau_n(x)\}}\nabla f_n(b_s^x)\Id b_s
\\\nonumber
&=\frac{1}{2}(f_n(x)-f_n(b_{t\wedge\tau_n(x)}^x))
+\frac{1}{4}\int_0^{t\wedge\tau_n(x)}\Delta f_n(b_s^x)\Id s
\\\nonumber
&=\frac{1}{2}\big(\langle\beta_{\sigma,x}^-|\beta_{\sigma,x}^+\rangle_{\HP}
-\langle\beta_{\sigma,b_{t\wedge\tau_n(x)}^x}^-|\beta_{\sigma,b_{t\wedge\tau_n(x)}^x}^+\rangle_{\HP}\big)
+\frac{1}{2}\int_0^{t\wedge\tau_n(x)}\langle\alpha_{\sigma,b_s^x}^-|\alpha_{\sigma,b_s^x}^+\rangle_{\HP}\Id s
\\\label{ItotI}
&\quad
-\int_0^{t\wedge\tau_n(x)}\int_{\mc{K}}\lambda\ol{\beta_{\sigma,b_s^x}^-}\beta^+_{\sigma,b_s^x}\Id\mu\,\Id s,
\end{align}
for all $t\ge0$ and $n\in\NN$ with $n\ge n_0$. Since $\lambda\beta^{+}_{\sigma,y}=v_{\sigma,y}+\beta^+_{\sigma,y}$,
the term in the last line of \cref{ItotI} satisfies
\begin{align*}
-\int_0^{t\wedge\tau_n(x)}\int_{\mc{K}}\lambda\ol{\beta_{\sigma,b_s^x}^-}\beta^+_{\sigma,b_s^x}\Id\mu\,\Id s
&=-d_{\sigma,t\wedge\tau_n(x)}(x)
-\int_0^{t\wedge\tau_n(x)}\langle\beta_{\sigma,b_s^x}^-|\beta_{\sigma,b_s^x}^+\rangle_{\HP}\Id s.
\end{align*}
Finally, we combine \cref{defUinftyplus,def:cqsigma}, again
using that $U^+_{\reg,t}(\tilde{v}_\sigma;x)$ equals $0$ on $\{\lambda\ge\sigma\}$, to get
\begin{align*}
-c_{\sigma,t}(x)
&=\langle\beta^-_{\sigma,b_t^x}|\beta_{\sigma,b_t^x}^+\rangle_{\HP}
-\eul^{-t}\langle\beta^-_{\sigma,b_t^x}|\beta_{\sigma,x}^+\rangle_{\HP}-
\langle\beta^-_{\sigma,b_t^x}|\eul^{-t}M_{\sigma,t}^+(x)\rangle_{\HP},\quad t\ge0.
\end{align*}

Taking all these remarks, \cref{def:asigma,def:wsigma,lem:Itoaction} into account,
we $\PP$-a.s. arrive at
\begin{align*}
	u_{\reg,t\wedge\tau_n(x)}(v;x)&= u_{\reg,t\wedge\tau_n(x)}(\tilde{v}_\sigma;x)
	+a_{\sigma,t\wedge\tau_n(x)}(x)
	+w_{\sigma,t\wedge\tau_n(x)}(x) 
	\\
	&\quad
	- \langle\beta_{\sigma,b_{t\wedge\tau_n(x)}^x}^- | \eul^{-(t\wedge\tau_n(x))}
	M_{\sigma,t\wedge\tau_n(x)}^+(x) \rangle_{\HP} 
	+ \langle M_{\sigma,t\wedge\tau_n(x)}^-(x)| \beta_{\sigma,x}^+ \rangle_{\HP}
	\\
	&\quad -\ii\int_0^{t\wedge\tau_n(x)}\Im\langle\alpha_{\sigma,b_s^x}^-|\beta_{\sigma,b_s^x}^+\rangle_{\HP}\Id b_s
	+m_{\sigma,t\wedge\tau_n(x)}(x),
\end{align*}
for all $t\ge0$ and $n\ge n_0$. Since $\tau_n(x)\uparrow\tau_{\Geb}(x)$ as $n\to\infty$, this
proves \cref{lemuregu}.
\end{proof}

\subsection{Convergence and exponential moment bound}
We now turn to the proof of \cref{prop:conveu}. The only non-obvious missing ingredient is treated in the next lemma first.

\begin{lem}\label{lemexpmbmsigmax}
If $p>0$ and $\sigma\in[2,\infty)$, then
\begin{align*}
4p g_\sigma^2\le1\quad\Rightarrow\quad\forall\,t\ge0:\quad
\sup_{x\in\Geb}\EE\Big[\sup_{s\in[0,t]}|\eul^{m_{\sigma,s}(x)}|^p\Big]\le (1+\pi)\eul^{t/8}.
\end{align*}
\end{lem}

\begin{proof}
Let $x\in\Geb$ and $\sigma\in[2,\infty)$. For both choices of the sign we define
\begin{align*}
J^\pm_{\sigma,t}(x)&\coloneq\int_0^t\Re\langle\alpha_{\sigma,b_s^x}^\pm|\eul^{-s}M_{\sigma,s}^+(x)\rangle_{\HP}\Id b_s,
\quad t\ge0,
\end{align*}
Employing It\^{o}'s product formula we observe that, $\PP$-a.s.,
\begin{align*}
\eul^{-2t}\|M_{\sigma,t}^+(x)\|_{\HP}^2&=-2\int_{0}^t\eul^{-2s}\|M_{\sigma,s}^+(x)\|_{\HP}^2\Id s
\\
&\quad
+2\int_0^t\eul^{-2s}\Re\langle M_{\sigma,s}^+(x)|\eul^{s}\alpha_{\sigma,b_s^x}^+\rangle_{\HP}\Id b_s
+\int_0^t\eul^{-2s}\|\eul^{s}\alpha_{\sigma,b_s^x}^+\|_{\HP}^2\Id s,
\end{align*}
for all $t\ge0$, so that
\begin{align*}
\int_{0}^t\eul^{-2s}\|M_{\sigma,s}^+(x)\|_{\HP}^2\Id s
&\le J^+_{\sigma,t}(x)+\frac{t}{2}\cdot\sup_{y\in\Geb}\|\alpha_{\sigma,y}^+\|_{\HP}^2,\quad t\ge0.
\end{align*}
This $\PP$-a.s. implies
\begin{align}\nonumber
[J_\sigma^\pm(x)]_t&=\int_0^t|\Re\langle\alpha_{\sigma,b_s^x}^\pm|\eul^{-s}M_{\sigma,s}^+(x)\rangle_{\HP}|^2\Id s
\\\label{derek1}
&\le \sup_{y\in\Geb}\|\alpha_{\sigma,y}^\pm\|_{\HP}^2\int_{0}^t\eul^{-2s}\|M_{\sigma,s}^+(x)\|_{\HP}^2\Id s
\le g_\sigma^2\bigg(J^+_{\sigma,t}(x)+\frac{g_\sigma^2t}{2}\bigg),\quad t\ge0.
\end{align}
Pick some $a>0$ and fix $t>0$ for the moment. In view of \cref{impl01} (where $\tau^+_n(x)\to\infty$, $n\to\infty$)
we find some cutoff parameter $\Sigma(a,t)\in[2,\infty)$ such that the random variable
$\sup_{s\in[0,t]}\exp(a\|\chi_{\{\lambda\ge\Sigma(a,t)\}}M_{\sigma,s}^+(x)\|_{\HP}^2)$ has finite expectation. 
Furthermore, by \cref{lemItoM}(iii) (applied to the coupling function $\chi_{\{\lambda<\Sigma(a,t)\}}v_\sigma$)
and trivial estimations we find some $c(a,t)\in(0,\infty)$ such that
$\sup_{s\in[0,t]}\|\chi_{\{\lambda<\Sigma(a,t)\}}M_{\sigma,s}^+(x)\|_{\HP}\le c(a,t)$, $\PP$-a.s.
These remarks prove the {\em a priori} bound
\begin{align*}
\forall\, a,t>0:\quad\EE\Big[\sup_{s\in[0,t]}\eul^{a\|M_{\sigma,s}^+(x)\|_{\HP}^2}\Big]<\infty.
\end{align*}
Also taking \cref{derek1} into account we see that all exponential moments of $[J_{\sigma}^\pm(x)]_t$ exist for every $t\ge0$.
In particular, $J_{\sigma}^\pm(x)$ are $L^2$-martingales and, given $p>0$, we infer from \cref{expmartest2,derek1} that
\begin{align}\label{derek2}
\EE\Big[\sup_{s\in[0,t]}\eul^{pJ^\pm_{\sigma,s}(x)}\Big]&
\le c_0\EE\big[\eul^{4p^2[J_{\sigma}^\pm(x)]_t}\big]^{1/2}
\le c_0\eul^{p^2g_\sigma^4t}
\EE\big[\eul^{4p^2g_\sigma^2J^+_{\sigma,t}(x)}\big]^{1/2},
\end{align}
for all $t\ge0$ with $c_0=(1+\pi)^{1/2}$. In particular, the leftmost expectation in \cref{derek2} is finite
for both choices of the sign and all $t\ge0$.
Thus, we obtain, first for the plus sign, afterwards for the minus sign, the following implication
\begin{align*}
4pg_\sigma^2\le1\quad\Rightarrow\quad\forall\,t\ge0:\quad \EE\Big[\sup_{s\in[0,t]}\eul^{pJ^\pm_{\sigma,s}(x)}\Big]
\le c_0^2\eul^{2p^2g_\sigma^4t}\le c_0^2\eul^{t/8}.
\end{align*}
Since $\Re m_{\sigma,s}(x)=J_{\sigma,s}^-(x)$, $s\ge0$, this proves the assertion.
\end{proof}

We can now prove the first part of \cref{prop:conveu}.
\begin{proof}[Proof of the exponential moment bound \cref{expmbdu}]
Let $p,t>0$ and $x\in\Geb$. We pick some $\sigma\ge2$ satisfying $256p L_1(v_\sigma)^2\le1$, which
in view of \cref{def:gs} ensures that $8pg_\sigma^2\le 1$.
Hence, the exponential moment bound of \cref{lemexpmbmsigmax} is available with $2p$ put in place of $p$.
Further, we have the trivial bounds $|a_{\sigma,t}(x)|\le c_1g_\sigma^2$
and $|w_{\sigma,t}(x)|\le c_1g_\sigma^2t$ with some universal constant $c_1\in(0,\infty)$,
while $|u_{\reg,t}(\tilde{v}_\sigma;x)|$ can be estimated (trivially) by means of \cref{brureg}.
Finally, we infer from \cref{impl02} (with $p/8$ put in place of $p$) that
\begin{align*}
\EE\big[\eul^{4pg_{\sigma}\eul^{(\mp t)\wedge0}\|M^{\pm}_{\sigma,t}(x)\|_{\HP}}\big]
&\le\eul^{8pg_\sigma^2}
\EE\big[\eul^{(p/2)\eul^{(\mp2 t)\wedge0}\|M^{\pm}_{\sigma,t}(x)\|_{\HP}^2}\big]
\le c_2,
\end{align*}
with another universal constant $c_2\in(0,\infty)$.
Since $\|\beta_{\sigma,y}^\pm\|_{\HP}\le g_\sigma$, $y\in\RR^d$, these remarks in conjunction with
\cref{def:actioninfty} and the generalized H\"{o}lder inequality (with exponents $2$, $4$, $4$, $\infty$)
imply the bound
\begin{align}\label{expmbduprel}
\sup_{x\in\Geb}\EE\big[|\eul^{u_{\sigma,t}(x)}|^p\big]
		&\le c_2^{2/4}(1+\pi)^{1/2}\eul^{c_1pg_\sigma^2(1+t)+t/16+p\sup_{y\in\Geb}\|\tilde{v}_{\sigma,y}\|^2t},
\end{align}
where $pg_\sigma^2\le1/8$.
By virtue of \cref{lemsigmau} we can replace $\sigma$ by any fixed $\kappa\in[2,\infty)$ on the left hand side
of \cref{expmbduprel} provided that we insert the indicator function $\chi_{\{t<\tau_\Geb(x)\}}$ under the expectation at the same time.
After that we may pass to the limit $\sigma\downarrow\vs_p$ on the right hand side
of \cref{expmbduprel}, if necessary.
\end{proof}

The proof of the convergence relation~\cref{conveu} makes use of the following lemma:

\begin{lem}\label{lem:convuLp}
Under the assumptions of \cref{prop:conveu}, let $\sigma\in[2,\infty)$. Then
\begin{align*}
\sup_{x\in\Geb}\EE\bigg[\sup_{s\in[0,t]}|u_{\sigma,s}^n(x)-u_{\sigma,s}(x)|^{p}\bigg]\xrightarrow{\;\;n\to\infty\;\;}0,
\quad p,t>0.
\end{align*}
\end{lem}

\begin{proof}
Assume that $p\ge1$ without loss of generality.
By Minkowski's inequality for $\PP$, it suffices to treat the seven contributions to the
complex action in \cref{def:actioninfty} separately. In view of the bound in \cref{def:gs}, which can also be applied to the
coupling function $v^n-v$, and since
$\sup_{n\in\NN}\sup_{y\in\Geb}\|\tilde{v}^n_{\sigma,y}\|_{\HP}<\infty$ and
$\sup_{n\in\NN}L_1(v_n)<\infty$ as well as $\sup_{y\in\Geb}\|\tilde{v}^n_{\sigma,y}-\tilde{v}_{\sigma,y}\|_{\HP}\to0$
and $L_1(v_n-v)\to0$ as $n\to\infty$, it is obvious how to treat the terms
$u_{\reg,t}(\tilde{v};x)$, $a_{\sigma,t}(x)$ and $w_{\sigma,t}(x)$ when they are approximated by their analogues
for $v^n$. 

In the remaining part of this proof objects defined by means of $v^n$ get an additional superscript $n$.
To deal with the fourth member on the right hand side of \cref{def:actioninfty}, we combine the previous remarks
with \cref{qvMkappa,BurkM} obtaining
\begin{align}\nonumber
&\sup_{x\in\Geb}\EE\bigg[\sup_{s\in[0,t]}\big|\langle\beta^{n,-}_{\sigma,b_s^x}|\eul^{-s}M_{\sigma,s}^{n,+}(x)\rangle_{\HP}
-\langle\beta^-_{\sigma,b_s^x}|\eul^{-s}M_{\sigma,s}^+(x)\rangle_{\HP}\big|^p\bigg]^{1/p}
\\\nonumber
&\le \sqrt{2}\cdot4L_1(v^n_\sigma-v_\sigma)\sup_{x\in\Geb}\EE\Big[\sup_{s\in[0,t]}\|M_{\sigma,s}^+(x)\|_{\HP}^p\Big]^{1/p}
\\\nonumber
&\quad+\sqrt{2}\cdot4L_1(v^n_\sigma)\sup_{x\in\Geb}\EE\Big[\sup_{s\in[0,t]}\|M_{\sigma,s}^{n,+}(x)-M_{\sigma,s}^+(x)\|_{\HP}^p\Big]^{1/p}
\\\label{convterm4}
&\le c_p(\eul^{2t}-1)^{1/2}L_1(v^n_\sigma-v_\sigma)\big(L_1(v_{\sigma})+\sup_{m\in\NN}L_1(v^m_\sigma)\big)
\xrightarrow{\;\;n\to\infty\;\;}0,\quad t>0,
\end{align}
with a solely $p$-dependent $c_p\in(0,\infty)$.
Here we estimated the expectation involving the difference $M_{\sigma,s}^{n,+}(x)-M_{\sigma,s}^+(x)$ by
\cref{qvMkappa,BurkM} with $v^n-v$ put in place of $v$.
 The fifth member on the right hand side of \cref{def:actioninfty} can be treated in the same way.
 
 Finally, Burkholder's inequality and the remarks in the first paragraph of this proof take care of
 the sixth term on the right hand side of \cref{def:actioninfty}, the purely imaginary martingale. So we are
 left with the martingale $m_{\sigma}(x)$ and its analogues $m_{\sigma}^n(x)$ defined by means of $v^n$.
Here Burkholder's inequality yields
\begin{align*}
&\EE\bigg[\sup_{s\in[0,t]}|m_{\sigma,s}^n(x)-m_{\sigma,s}(x)|^{p}\bigg]
\\
&\le c_p'
\EE\bigg[\bigg(\int_0^t\big|\langle\alpha^{n,-}_{\sigma,b_s^x}|\eul^{-s}M_{\sigma,s}^{n,+}(x)\rangle_{\HP}
-\langle\alpha^-_{\sigma,b_s^x}|\eul^{-s}M_{\sigma,s}^+(x)\rangle_{\HP}\big|^2\Id s\bigg)^{p/2}\bigg]
\\
&\le
c_p't^{p/2}\EE\bigg[\sup_{s\in[0,t]}\big|\langle\alpha^{n,-}_{\sigma,b_s^x}|\eul^{-s}M_{\sigma,s}^{n,+}(x)\rangle_{\HP}
-\langle\alpha^-_{\sigma,b_s^x}|\eul^{-s}M_{\sigma,s}^+(x)\rangle_{\HP}\big|^p\bigg],\quad t>0,
\end{align*}
with a solely $p$-dependent $c_p'\in(0,\infty)$.
Here the term in the last line converges to zero uniformly in $x$ as $n\to\infty$, because
\cref{convterm4} still holds true when the symbol $\beta$ is replaced by $\alpha$ in its first line.
\end{proof}

We can now complete the proof of \cref{prop:conveu}:

\begin{proof}[Proof of the convergence relation~\cref{conveu}]
Without loss of generality we may assume that $p\ge1$. 
As noted in the proof of \cref{lem:convuLp},
 $\sup_{n\in\NN}\sup_{y\in\Geb}\|\tilde{v}^n_{\vs_{2p},y}\|_{\HP}<\infty$ with $\vs_{2p}$ as defined in the statement of \cref{prop:conveu}.
 Thus, \cref{expmbdu} with $2p$ put in place of $p$ applies to $v$ and every $v^n$ and in particular
\begin{align*}
\sup_{\sigma\in[2,\infty)}\sup_{n\in\NN}\sup_{t\in[0,r]}\sup_{x\in\Geb}
\EE\big[\chi_{\{t<\tau_\Geb(x)\}}|\eul^{u_{\sigma,t}^n(x)}|^{2p}\big]&<\infty,\quad r>0.
\end{align*}
Moreover, by the fundamental theorem of calculus, Jensen's inequality and the generalized H\"{o}lder inequality,
\begin{align*}
&\EE\big[\chi_{\{t<\tau_\Geb(x)\}}|\eul^{u_{\sigma,t}^n(x)}-\eul^{u_{\sigma,t}(x)}|^p\big]
\\
&\le \int_0^1\EE\big[\chi_{\{t<\tau_\Geb(x)\}}|u_{\sigma,t}^n(x)-u_{\sigma,t}(x)|^p
|\eul^{\theta u_{\sigma,t}^n(x)}|^p|\eul^{(1-\theta)u_{\sigma,t}(x)}|^p\big]\Id \theta
\\
&\le\EE\big[|u_{\sigma,t}^n(x)-u_{\sigma,t}(x)|^{2p}\big]^{1/2}
\\
&\quad\cdot\sup_{\theta\in(0,1)}
\EE\big[\chi_{\{t<\tau_\Geb(x)\}}|\eul^{u_{\sigma,t}^n(x)}|^{2p}\big]^{\theta/2}
\EE\big[\chi_{\{t<\tau_\Geb(x)\}}|\eul^{u_{\sigma,t}(x)}|^{2p}\big]^{(1-\theta)/2},
\end{align*}
for all $t\ge0$, $x\in\Geb$ and $\sigma\in[2,\infty)$. Now \cref{conveu} follows from \cref{lem:convuLp}.
\end{proof}


\section{Weighted $L^p$ to $L^q$ bounds and convergence theorems for Feynman--Kac operators}\label{sec:FKsemigroups}

\noindent
The objective of the following \cref{ssec:FKI,ssec:FKops} is to analyze the right hand sides of our Feynman--Kac formulas 
\cref{eq:FKmainreg,eq:FKmain} considered as bounded operators from $L^p(\Geb,\Fock)$ to 
$L^q(\Geb,\Fock)$ with $1<p\le q\le\infty$. The convergence theorems established for these operators
in \cref{ssec:convthmsFKO} are used in the final \cref{sec:FKproof} to complete
the proofs of \cref{eq:FKmainreg,eq:FKmain} in a series of approximation steps.

\subsection{Feynman--Kac integrands: moment bounds and convergence}\label{ssec:FKI}

Let us first collect some bounds on the Feynman--Kac integrands defined in \cref{def:W} for the possibly ultraviolet singular
coupling function $v$.
\begin{lem}\label{lem:mbconvW}
	Assume that $\Geb$ fulfills \cref{tailbound}. Then
	\begin{align}\label{momentbdW}
		\sup_{x\in\Geb}\EE\big[\chi_{\{t<\tau_\Geb(x)\}}\|W_{\sigma,t}(x)\|^p\big]\le c_{\Geb}\eul^{p\nu_p(v,t)},\quad t\ge0,\,p>0,
	\end{align}
	with a solely $\Geb$-dependent $c_{\Geb}\in(0,\infty)$ and
	\begin{align*}
	\nu_p(v,t)&\coloneq \ln(2)+24\sup_{y\in\Geb}\|\tilde{v}_{\sigma_{48p},y}\|_{\HP}^2
	+\bigg(\frac{c}{2p}+\sup_{y\in\Geb}\|\tilde{v}_{\vs_{2p},y}\|_{\HP}^2\bigg)t,
	\end{align*}
	where $c$ is the universal constant appearing in \cref{expmbdu} and $\sigma_{32p}$ and $\vs_{2p}$ are given by
	\cref{def:sp,def:vsp}, respectively.
	Furthermore, if $v^1,v^2,\ldots$ are coupling functions fulfilling the same hypotheses as $v$
	such that $L_1(v^n-v)\to0$ as $n\to\infty$ and if $W_{\sigma,s}^n(x)$ is defined by putting $v^n$
	in place of $v$ in \cref{def:W}, then
	\begin{align}\label{convW}
		\sup_{s\in[0,t]}\sup_{x\in\Geb}
		\EE\big[\chi_{\{s<\tau_\Geb(x)\}}\|W_{\sigma,s}^n(x)-W_{\sigma,s}(x)\|^p\big]
		\xrightarrow{\;\;n\to\infty\;\;}0,\quad p,t>0.
	\end{align}
\end{lem}

\begin{proof}
	The moment bound \eqref{momentbdW} follows from \cref{bdFt,def:W,expmomentUminus,expmbdu} 
	as well as H\"{o}lder's inequality with exponents $2$, $4$ and $4$.
	The convergence relation \eqref{convW} is proved in a straightforward fashion
	taking also \cref{conveu,bddFt,convUpmLp}
	into account in addition to \cref{bdFt,def:W,expmomentUminus,expmbdu}.
\end{proof}

\subsection{Feynman--Kac operators: definitions and weighted $L^p$ to $L^q$ bounds}\label{ssec:FKops}

Next, we treat the Feynman--Kac operators given by the right hand sides of \cref{eq:FKmain}.

For $t\ge0$ and $x\in\Geb$, let $D_{t}(x)$ be either $W_{\reg,t}(\vt;x)$, $W_{\sigma,t}(x)$ or
$W_{\sigma,t}^n(x)-W_{\sigma,t}(x)$ where the $n$-dependent operator-valued processes
are defined as in \cref{lem:mbconvW}.
In the latter two cases, we assume $\Geb$ satisfies \cref{tailbound}.
 Let $a\ge0$ and $\Psi:\Geb\to\Fock$ be measurable. 
Then $\EE[\eul^{\vr a|b_t|}]\le 2^{d/2}\eul^{\vr^2a^2t}$, $\vr\ge0$, and H\"{o}lder's inequality imply
\begin{align}\nonumber
	&\EE\big[\chi_{\{t<\tau_\Geb(x)\}}\eul^{a|b_t|}\|D_{t}(x)\|\eul^{-\Re S_t(x)}\|\Psi(b_t^x)\|_{\Fock}\big]
	\\\label{HoelderT0}
	&\le 2^{d/6}\eul^{3a^2t} \sup_{y\in\Geb}\EE\big[\chi_{\{t<\tau_\Geb(y)\}}\|D_{t}(y)\|^{3}\big]^{1/3}
	\sup_{z\in\RR^{d}}\EE\big[\eul^{3\int_0^tV_-(b_s^z)\Id s}\big]^{1/3}\|\Psi\|_\infty,
\end{align}
if $\Psi$ is essentially bounded, as well as
\begin{align}\nonumber
	&\EE\big[\chi_{\{t<\tau_\Geb(x)\}}\eul^{a|b_t|}\|D_{t}(x)\|\eul^{-\Re S_t(x)}\|\Psi(b_t^x)\|_{\Fock}\big]
	\\\nonumber
	&\le2^{d/6p'}\eul^{3p'a^2t} \sup_{y\in\Geb}\EE\big[\chi_{\{t<\tau_\Geb(y)\}}\|D_{t}(y)\|^{3p'}\big]^{1/3p'}
	\\\label{HoelderT}
	&\quad\cdot\sup_{z\in\RR^{d}}\EE\big[\eul^{3p'\int_0^tV_-(b_s^z)\Id s}\big]^{1/3p'}
	\EE\big[\chi_{\{t<\tau_\Geb(x)\}}\|\Psi(b_t^x)\|_{\Fock}^p\big]^{1/p},
\end{align}
whenever $\Psi$ is $p$-integrable for some $p\in(1,\infty)$. Here $p'$ is the exponent conjugate to $p$.
Furthermore, we recall that
\begin{align}\label{LpLinftyLaplace}
		&\sup_{x\in\Geb}\EE\big[\chi_{\{t<\tau_\Geb(x)\}}\|\Psi(b_t^x)\|_{\Fock}^p\big]\le(2\pi t)^{-d/2} \|\Psi\|_p^p,
		\\\label{LpLqLaplace}
		&\int_{\Geb}\EE\big[\chi_{\{t<\tau_\Geb(x)\}}\|\Psi(b_t^x)\|_{\Fock}^p\big]^{q/p}\Id x
		\le c_{p,q,d} t^{-\frac{d}{2}(\frac{q}{p}-1)}\|\Psi\|_p^q,
\end{align}
for all $t>0$, $p\in(1,\infty)$, $p$-integrable $\Psi$ and $q\in[p,\infty)$.
In view of the above bounds and Remark~\ref{rem:Wmeas} the following definitions are meaningful:

\begin{defn}\label{defn:TUV}
	Let $t\ge0$, $p\in(1,\infty]$ and $\Psi\in L^p(\Geb,\Fock)$. 
	For every $x\in\Geb$ for which $S_t(x)$ is defined (thus for a.e. $x$) we generalize \cref{def:Treg} by
	\begin{align*}
		(T_{\reg,t}\Psi)(x)&\coloneq\EE\big[\chi_{\{t<\tau_\Geb(x)\}}\eul^{-\ol{S}_t(x)}W_{\reg,t}(\vt;x)^*\Psi(b_t^x)\big],
	\end{align*}
	and define
	\begin{align}\label{def:Tt}
		(T_{t}\Psi)(x)&\coloneq\EE\big[\chi_{\{t<\tau_\Geb(x)\}}\eul^{-\ol{S}_t(x)}W_{\sigma,t}(x)^*\Psi(b_t^x)\big].
	\end{align}
	Note that, by \cref{lemsigmaUpm,lemsigmau}, the right hand side of \cref{def:Tt}
	does not depend on the choice of $\sigma\in[2,\infty)$.
\end{defn}

\begin{thm}
	Let $p\in(1,\infty]$, $a\ge0$ and $F:\RR^d\to\RR$ be Lipschitz continuous with Lipschitz constant $\le a$.
	Let $\Psi\in L^p(\Geb,\Fock)$ be such that $\eul^F\Psi$ is $p$-integrable over $\Geb$ as well.
	If $\Geb$ satisfies \cref{tailbound},
	then there exists $c_*\in(0,\infty)$, solely depending on $p$, $a$, $V_-$ and $\Geb$ such that
	\begin{align}\label{LpLinftyTUV}
		\|(\eul^FT_{t}\Psi)(x)\|_{\Fock}&\le t^{-d/2p}\eul^{c_*(1+t)+\nu_{3p'}(v,t)}\|\eul^F\Psi\|_p,\quad\text{a.e. $x\in\Geb$,}
	\end{align}
	for every $t>0$.
	Furthermore, for every $q\in[p,\infty]$ we find some 
	$c_\diamond\in(0,\infty)$, solely depending on $p$, $q$, $a$, $V_-$ and $\Geb$, such that
	\begin{align}\label{LpLqTUV}
		\|\eul^FT_{t}\Psi\|_{q}&\le t^{-\frac{d}{2}(\frac{1}{p}-\frac{1}{q})}\eul^{c_\diamond(1+t)+\nu_{3p'}(v,t)}
		\|\eul^F\Psi\|_p,\quad t>0.
	\end{align}
	The same bounds hold for $T_{\reg,t}$ without the assumption \cref{tailbound} provided that, on the right hand sides, 
	$\nu_{3p'}(v,t)$ is replaced by $c_\vt t$ with $c_{\vt}$ given by \cref{nbWreg}. In this case the constants $c_*$ and
	$c_\diamond$ do not depend on any properties of $\Geb$ other than its dimension~$d$.
\end{thm}

\begin{proof}
	Combine 
	\cref{expmbdVminus,nbWreg,momentbdW,defn:TUV,HoelderT0,HoelderT,LpLinftyLaplace,LpLqLaplace}.
	In fact, \cref{LpLinftyTUV} holds for every $x\in\Geb$ for which the generalized Stratonovich integral contributing to $S_t(x)$
	is well-defined; see \cref{ssecStrat}.
\end{proof}

\subsection{Convergence theorems for Feynman--Kac operators}\label{ssec:convthmsFKO}

To infer \cref{thm:mainFK} from \cref{thm:mainFKreg} we employ the following result, where
$\|\cdot\|_{p,q}$ denotes the operator norm from $L^p(\Geb,\Fock)$ to $L^q(\Geb,\Fock)$.

\begin{thm}\label{thm:approxvn}
Assume $\Geb$ satisfies \cref{tailbound}.
Let $v^1,v^2,\ldots$ be coupling functions fulfilling the same hypotheses as $v$ and
assume that $L_1(v^n-v)\to0$ as $n\to\infty$. 
Denote by $T^n_{t}$ the maps obtained by putting $v^n$ in place of $v$ in Definition~\ref{defn:TUV}.
Let $p\in(1,\infty]$. Then
\begin{align*}
		\sup_{x\in\Geb\setminus\scr{N}_t}\|(T_{t}^n\Psi-T_{t}\Psi)(x)\|_{\Fock}\xrightarrow{\;\;n\to\infty\;\;}0,
	\end{align*}
	for all $t>0$ and $\Psi\in L^p(\Geb,\Fock)$, 
	where $\scr{N}_t\subset\Geb$ is any Borel set of measure zero such that $S_t(x)$ is defined for all $x\in\Geb\setminus\scr{N}_t$.
	Furthermore, for every $q\in[p,\infty]$,
		\begin{align}\label{LpLqconvTvn}
		\|T_{t}^n-T_{t}\|_{p,q}\xrightarrow{\;\;n\to\infty\;\;}0,\quad t>0.
	\end{align}
	In the case $q=p$ we actually have the locally uniform convergence
	\begin{align}\label{LpLpconvTUV}
		\sup_{s\in[0,t]}\|T_{s}^n-T_{s}\|_{p,p}\xrightarrow{\;\;n\to\infty\;\;}0,\quad t>0.
	\end{align}
\end{thm}

\begin{proof}
	The assertions follow from
	\cref{expmbdVminus,convW,defn:TUV,HoelderT0,HoelderT,LpLinftyLaplace,LpLqLaplace}.
\end{proof}

To obtain Feynman--Kac formulas for possibly singular magnetic vector potentials $A$ and electrostatic potentials $V$,
we need to approximate them by regular ones. 
The appropriate convergence properties of the Feynman--Kac operators then are secured by the next two theorems.

\begin{thm}\label{lem:TapproxA}
	Assume that the magnetic vector potential has an extension $A\in L_\loc^2(\RR^d,\RR^d)$. Let
	$A_1,A_2,\ldots\in L^2_{\loc}(\RR^{d},\RR^{d})$ and $\alpha\in L^2_{\loc}(\RR^{d},\RR)$
	be such that $A_n\xrightarrow{n\to\infty} A$ a.e. on $\RR^d$, as well as $|A_n|\le\alpha$ 
	a.e. on $\RR^{d}$ for every $n\in\NN$. Let $p\in(1,\infty)$, $\Psi\in L^p(\Geb,\Fock)$ and $t\ge0$.
	Denote by $T^n_{t}$ and $T_{\reg,t}^n$
	the maps obtained by putting $A_n$ in place of $A$ in Definition~\ref{defn:TUV}. If $\Geb$ satisfies \cref{tailbound}, then
	\begin{align}\label{convA1}
		&\lim_{n\to\infty}\|(T_{t}^n\Psi-T_{t}\Psi)(x)\|_{\Fock}=0,\quad \text{a.e. $x\in\Geb$},
		\\\label{convA2}
		&T_{t}^n\Psi\xrightarrow{\;\;n\to\infty\;\;} 
		T_{t}\Psi\quad\text{in every $L^q(\Geb,\Fock)$ with $q\in[p,\infty)$.}
	\end{align}
	The same holds without the assumption \cref{tailbound} when $T_{\reg,t}^n$ and $T_{\reg,t}$ 
	are put in place of $T_t^n$ and $T_t$, respectively.
\end{thm}

\begin{proof}
	Let $\Phi^n_t(x)$ denote the process obtained by putting $A_n$ in place of $A$ in \eqref{defPhitx}.
	Then the assumptions and \cite[Theorem~9.2]{Matte.2021} directly imply the existence
	of a Borel zero set $N\subset\RR^{d}$ such that 
	$\Phi_t(x)$ and all $\Phi_t^n(x)$ are well-defined and
	$\Phi_t^n(x)\xrightarrow{n\to\infty}\Phi_t(x)$ in probability for all $x\in\RR^{d\nu}\setminus N$.
	Now let $x\in\Geb\setminus N$.
	Using $|\eul^{\ii r}-\eul^{\ii s}|\le 2\wedge|r-s|$, $r,s\in\RR$, we then find
	\begin{align*}
		\|(T_{t}^n\Psi-T_{t}\Psi)(x)\|_{\Fock}
		&\le\sup_{y\in\RR^{d}}\EE\Big[\eul^{3p'\int_0^tV_-(b_s^y)\Id s}\Big]^{1/3p'}
		\sup_{z\in\RR^{d}}\EE[\chi_{\{t<\tau_\Geb(z)\}}\|W_{\UV,t}(z)\|^{3p'}]^{1/3p'}
		\\
		&\quad\cdot \EE[(2\wedge|\Phi_t^n(x)-\Phi_t(x)|)^{3p'}]^{1/3p'}\EE[\chi_{\{t<\tau_\Geb(x)\}}\|\Psi(b_t^x)\|_{\Fock}^p]^{1/p}.
	\end{align*}
	Since the functions $(2\wedge|\Phi_t^n(x)-\Phi_t(x)|)^{3p'}\le 2^{3p'}$ go to $0$ in probability as $n\to\infty$ and
	obviously are bounded in $L^2(\PP)$ uniformly in $n$, Vitali's theorem implies that
	$\EE[(2\wedge|\Phi_t^n(x)-\Phi_t(x)|)^{3p'}]\xrightarrow{n\to\infty}0$. 
	Together with \cref{expmbdVminus,momentbdW} these remarks imply \eqref{convA1}.
	Now \cref{convA2} follows by dominated convergence because of \eqref{LpLqLaplace}.
	For $T_{\reg,t}^n$ and $T_{\reg,t}$ the proof uses \cref{nbWreg} instead of \cref{momentbdW} and 
	is identical otherwise.
\end{proof}

A similar approximation result for sequences of potentials can easily be proved by 
applications of the dominated convergence theorem and \eqref{HoelderT}. We refrain from giving a separate proof here.

\begin{thm}\label{lem:TapproxV}
	Let $V_1,V_2,\ldots\in L^1_{\loc}(\Geb,\RR)$
	and assume that $V_n\xrightarrow{n\to\infty} V$ a.e. on $\Geb$, as well as $V_n\ge -V_-$ 
	a.e. on $\Geb$ for every $n\in\NN$. Let $p\in(1,\infty)$, $\Psi\in L^p(\Geb,\Fock)$ and $t\ge0$.
	Denote by $T^n_{t}$ and $T_{\reg,t}^n$ the maps obtained by putting $V_n$ in place of $V$ in \cref{defn:TUV}. Then
	\cref{convA1,convA2} hold true, provided that $\Geb$ fulfills \cref{tailbound}.
	The same convergence relations hold when $T_{\reg,t}^n$ and $T_{\reg,t}$ are put in place of $T_t^n$ and $T_t$, respectively,
	 in which case the assumption \cref{tailbound} on $\Geb$ is unnecessary.
\end{thm}

\begin{rem}\label{rem:approxAV}
	There exist vector potentials $A_1,A_2,\ldots\in C_0^\infty(\RR^{d},\RR^{d})$ 
	and some dominating function $\alpha\in L^2_{\loc}(\RR^{d},\RR)$ 
	as well as electrostatic potentials $V_1,V_2,\ldots\in C_b(\RR^{d},\RR)$
	fulfilling the hypotheses of \cref{lem:TapproxV,lem:TapproxA}, respectively.
	
	In fact, the existence of $A_n$ and $\alpha$ has been shown in 
	\cite[Lemma~9.3 \& Step~1 of the proof of Proposition~9.4]{Matte.2021}. The construction of $V_n$ is standard.
\end{rem}


\subsection{Feynman--Kac formulas in the general case}\label{sec:FKproof}

We are now in a position to complete the proofs of the Feynman--Kac formulas stated in \cref{ssec:FK}.

\begin{proof}[Proof of \cref{thm:mainFKreg}]
	Assume first that $V\in C_b(\RR^{d},\RR)$ and $A\in L_{\loc}^2(\RR^d,\RR^d)$. Pick 
	$A_n\in C_0^\infty(\RR^{d},\RR^{d})$ as in Remark~\ref{rem:approxAV}.
	Denote by $H_n(\vt)$ and $T_{\reg,t}^n$ the Hamiltonian and Feynman--Kac operators, respectively,
	defined by means of $A_n$ in place of $A$. Let $\Psi\in L^2(\Geb,\Fock)$ and $t\ge0$.
	Then $\eul^{-tH_n(\vt)}\Psi=T_{\reg,t}^n\Psi$ for all $n\in\NN$ by \cref{cor:FKUVregAV}.
	However, $T_{\reg,t}^n\Psi\to T_{\reg,t}\Psi$, $n\to\infty$, in $L^2(\Geb,\Fock)$ by Lemma~\ref{lem:TapproxA}, while
	$\eul^{-tH_n(\vt)}\Psi\to\eul^{-tH(v_\UV)}\Psi$, $n\to\infty$, by Theorem~\ref{thm:rescontA} and since
	strong resolvent convergence of semibounded operators entails strong convergence 
	of their semigroup members. 
	Thus, $\eul^{-tH(\vt)}=T_{\reg,t}$.
	
	Still assuming $V\in C_b(\RR^{d},\RR)$, we can copy the first part of the proof of
	\cite[Theorem~1.1 in \textsection9.4]{Matte.2021} to extend this result to $A\in L_{\loc}^2(\Geb,\RR^d)$.
	Here we set $\Geb_n\coloneq\{x\in\Geb|\,\dist(x,\Geb^c)>1/n\}$, $n\in\NN$, so that
	$A_n\coloneq\chi_{\Geb_n}A$, extended by $0$ to $\RR^d$, belongs to $L_\loc^2(\RR^d,\RR^d)$.
	Denote by $\fr{h}_n(\vt)$ and $H_n(\vt)$ the polaron quadratic form and polaron Hamiltonian {\em on $\Geb_n$}
	defined by means of $A_n$. Let, as usual, $\fr{h}(\vt)$ and $H(\vt)$ be the ones on $\Geb$ defined by means of $A$.
	Tacitly extending functions on $\Geb_n$ by $0$ to larger subsets of $\Geb$,
	we then have $\dom(\fr{h}_n(\vt))\subset\dom(\fr{h}_m(\vt))\subset\dom(\fr{h}(\vt))$, $m>n$,
	and $\fr{h}(\vt)[\Psi]=\lim_{n<m\to\infty}\fr{h}_m[\Psi]$ for all $\Psi\in\dom(\fr{h}_n(\vt))$ and $n\in\NN$.
	Let $t\ge0$ and $\Psi\in L^2(\Geb,\Fock)$.
	By \cite[Theorems~4.1 and~4.2]{Simon.1978c} the previous remarks imply that
	$\lim_{n\to\infty}\eul^{-tH_n(\vt)}(\Psi\restr_{\Geb_n})=\eul^{-tH(\vt)}\Psi$ in $L^2(\Geb,\Fock)$
	and, hence, a.e. on $\Geb$ along a subsequence.
	Define $S^n_t(x)$ and $\tau_n(x)$ as in the end of \cref{ssecStrat}. Then 
	$S_t^n(x)=S_t(x)$ holds $\PP$-a.s. on $\{t<\tau_n(x)\}$, whence
	the result of the first paragraph of this proof yields
	\begin{align}\label{extAGeb}
	(\eul^{-tH_n(\vt)}(\Psi\restr_{\Geb_n}))(x)
	&=\EE\big[\chi_{\{t<\tau_{n}(x)\}}\eul^{-\ol{S}_t(x)}W_{\reg,t}(\vt;x)^*\Psi(b_t^x)\big],
	\end{align}
	for a.e. $x\in\Geb_n$. Since $\chi_{\{t<\tau_n(x)\}}\to\chi_{\{t<\tau_{\Geb}(x)\}}$
	on $\Omega$, the right hand side of \cref{extAGeb} converges to
	$(T_{\reg,t}\Psi)(x)$ as $n\to\infty$ for every $x\in\Geb$. Altogether this proves
	\cref{eq:FKmainreg} for $V\in C_b(\RR^d,\RR)$.
	
	The extension to general $V$ is standard and we shall not give any details. We just mention that
	the extension proceeds in three steps: First, bounded $V$ are approximated by continuous and
	bounded $V_n$ by mollification. After that $V$ which are bounded from below are approximated
	by $V\wedge n$. Finally, general $V$ are approximated by $V\vee(-n)$.
	In all three steps strong convergence of the Feynman--Kac operators is ensured by
	Lemma~\ref{lem:TapproxV}. In the first step (mollification) strong resolvent convergence
	of the Hamiltonians can be checked directly using the second resolvent identity.
	In the last two steps monotone convergence theorems for quadratic forms 
	\cite[Theorems~S.14 and~S.16]{ReedSimon.1980}
	are invoked to show strong resolvent convergence of the Hamiltonians.
\end{proof}

Finally, we prove the Feynman--Kac formula for the polaron approximating $v$ by its
ultraviolet cutoff versions $\tilde{v}_n=\chi_{\{\lambda<n\}}v$, $n\in\NN$ (that we extended
by $0$ to elements of $\mc{L}^\infty(\RR^d,\HP)$).

\begin{proof}[Proof of \cref{thm:mainFK}]
Let $t\ge0$ and define $T_{\reg,t}^n$ and $T_t^n$ as in \cref{defn:TUV} with $\tilde{v}_n$
put in place of $\vt$ and $v$, respectively.
Thanks to \cref{thm:mainFKreg} (applied in the first equality) and \cref{lemUregU,lemuregu} (applied in the second one)
we know that $\eul^{-tH(\tilde{v}_n)}=T_{\reg,t}^n=T_t^n$ 
for all $n\in\NN$. In the limit $n\to\infty$, \cref{thm:approxvn,cor:normresconvUV} imply, however,
that $T_t^n\to T_t$ and $\eul^{-tH(\tilde{v}_n)}\to\eul^{-tH(v)}$ in operator norm. Here we also use that
norm resolvent convergence of semibounded operators entails norm convergence of the corresponding 
semigroup members. In conclusion, $\eul^{-tH(v)}=T_t$.
\end{proof}

\appendix

\section{Lieb--Yamazaki type bounds on the polaron interaction}\label{app:rfbinteraction}

\noindent
In this \lcnamecref{app:rfbinteraction}, we prove \cref{thm:LY} in the spirit of Lieb and Yamazaki \cite{LiebYamazaki.1958}.
In our presentation the Lieb-Yamazaki commutator argument is somewhat hidden, though: integrating
\cref{gottfried1} with respect to $x$ against the complex density $\ol{v}(x,k)$
we formally obtain the quadratic form of the commutator between the covariant derivative $-\ii\partial_{j}-A_j$ 
and $\ii\ol{v}(x,k)\ap(k)$ on the right hand side, and the usefulness of \cref{gottfried1} is revealed by the integration by parts
argument in the proof of  \cref{prop:LY}. Notice that our arguments yield integral formulas for the interaction
$\fr w(v)[\Psi]$ that apply to all $\Psi$ in the form domain of $H(v)$.

Given $\Phi_1,\Phi_2\in L^2(\Geb,\Fock)$, 
we write $\langle \Phi_1|\Phi_2\rangle_{\Fock}$ for the integrable function 
$\Geb\ni x\mapsto \langle\Phi_1(x)|\Phi_2(x)\rangle_{\Fock}$. 
The weak partial derivative
$\partial_j$ appearing repeatedly is acting on the variable $x_j$. 
As usual, $\mr{W}^{1,1}(\Geb)$ is the closure of $C_0^\infty(\Geb)$ in the Sobolev space $W^{1,1}(\Geb)$.

\begin{lem}\label{lemgottfried1}
	Let $\Psi,\Upsilon\in\dom(\fr{q}^{\min})$. Then
		$\langle(\NO+1)^{1/2}\Psi|\Upsilon\rangle_{\Fock}\in \mr{W}^{1,1}(\Geb)$
	and 
	\begin{align*}
		\partial_j\langle(\NO+1)^{1/2}\Psi|\Upsilon\rangle_{\Fock}
		= &
		\langle\ii w_j^*\Psi|(\NO+1)^{1/2}\Upsilon\rangle_{\Fock} 
		+\langle(\NO+1)^{1/2}\Psi|\ii w_j^*\Upsilon\rangle_{\Fock},
	\end{align*}
	for all $j\in\{1,\ldots,d\}$.
\end{lem}
\begin{proof}	
	Pick $\Psi_n,\Upsilon_n\in\operatorname{span}\{f\phi|f\in C_0^\infty(\Geb),\phi\in\fdom(N)\}$, $n\in\NN$, 
	such that $\Psi_n\to\Psi$
	in $L^2(\Geb,\Fock)$ and $\fr{q}^{\min}[\Psi_n-\Psi]\to0$ as $n\to\infty$ and
	analogously for $\Upsilon_n$. In particular,
	\begin{align}\label{gottfriedD0}
		\left.
		\begin{array}{r}
			w_j^*\Psi_n\to w_j^*\Psi,\ j=1,\ldots,d
			\\
			\\
			(\NO+1)^{1/2}\Psi_n\to(\NO+1)^{1/2}\Psi
		\end{array}
		\right\}\quad \text{in $L^2(\Geb,\Fock)$ as $n\to\infty$},
	\end{align}
	and again analogously for $\Upsilon_n$. Then
	\begin{align}\label{gottfriedD1}
		C_0^\infty(\Geb)\ni
		\langle(\NO+1)^{1/2}\Psi_n|\Upsilon_n\rangle_{\Fock}\xrightarrow{\;\;n\to\infty\;\;}
		\langle(\NO+1)^{1/2}\Psi|\Upsilon\rangle_{\Fock}\quad\text{in $L^1(\Geb)$.}
	\end{align}
	For all $n\in\NN$ and $j\in\{1,\ldots,d\}$, the definition of $w_j$, the Leibniz rule and the fact that $\ii A$ is purely imaginary entail
	\begin{align*}
		\partial_j\langle(\NO+1)^{1/2}\Psi_n|\Upsilon_n\rangle_{\Fock}
		=
		\langle\ii w_j\Psi_n|(\NO+1)^{1/2}\Upsilon_n\rangle_{\Fock}
		+\langle(\NO+1)^{1/2}\Psi_n|\ii w_j\Upsilon_n\rangle_{\Fock}.
	\end{align*}
	Recalling that $w_j\subset w_j^*$ as well as \eqref{gottfriedD0} 
	and its analog for $\Upsilon_n$, we conclude that
	\begin{align*}
		\partial_j\langle(\NO+1)^{1/2}\Psi_n|\Upsilon_n\rangle_{\Fock}
		\xrightarrow{\;\;n\to\infty\;\;}
		\langle\ii w_j^*\Psi|(\NO+1)^{1/2}\Upsilon\rangle_{\Fock}
		+\langle(\NO+1)^{1/2}\Psi|\ii w_j^*\Upsilon\rangle_{\Fock}
	\end{align*}
	in $L^1(\Geb)$ for all $j\in\{1,\ldots,d\}$.
	 Together with \eqref{gottfriedD1}, this proves the statement.
\end{proof}
The following statement is the major ingredient in our proof of \cref{thm:LY}. Therein, 
$I:L^2(\Geb,\Fock)\to L^2(\mu;L^2(\Geb,\Fock))$ (recall \cref{eq:pointwise0})
is the unique partial isometry with $\ker I =L^2(\Geb,\Fock_0)$ acting on vectors $g\ve(f)$ with $g\in L^2(\Geb)$
and $f\in\HP$ as
\begin{align}\label{eq:defI}
	(I (g\ve(f)))(k) \coloneq f(k)g(\NO+1)^{-1/2}\ve(f),\quad \text{$\mu$-a.e. $k\in\mc{K}$.}
\end{align}
We easily infer from \cref{eq:pointwise,eq:defI} that
\begin{align}\label{eq:apandI}
	\ap = (\NO+1)^{1/2}I = I\NO^{1/2}.
\end{align}
\begin{lem}\label{lemgottfried2}
	Let $\Psi\in\dom(\fr{q}^{\min})$. Then, for $\mu$-a.e. $k\in\mc{K}$, it holds
	$\langle\Psi|\ap(k)\Psi\rangle_{\Fock}\in\mr{W}^{1,1}(\Geb)$ and, for all $j\in\{1,\ldots,d\}$,
	\begin{align}\label{gottfried1}
		\partial_j\langle\Psi|\ap(k)\Psi\rangle_{\Fock}
		&=
		\ii\langle(\NO+1)^{1/2}\Psi|(I w_j^*\Psi)(k)\rangle_{\Fock}
		-\ii\langle w_j^*\Psi|\ap(k)\Psi\rangle_{\Fock}.
	\end{align}
\end{lem}
\begin{proof}
We may assume that $V_+=0$ throughout this proof.
	We pick some $j\in\{1,\ldots,d\}$ and shall first show that 
	\begin{align}\label{gottfried12}
		(I\Psi)(k)\in\dom(w_j^*)\quad\text{and}\quad
		w_j^*(I\Psi)(k)=(I w_j^*\Psi)(k)\quad\text{for $\mu$-a.e. $k$.}
	\end{align}
	Pick $\Psi_n\in\mathrm{span}\{f\phi|\,f\in C_0^\infty(\Geb),\,\phi\in\Fock\}$, $n\in\NN$,
	such that $\Psi_n\to\Psi$ in $L^2(\Geb,\Fock)$, $\fr{q}^{\min}[\Psi_n-\Psi]\to0$
	and in particular $w_j^*\Psi_n\to w_j^*\Psi$ in $L^2(\Geb,\Fock)$
	as $n\to\infty$. Evidently, $w_j^*(I\Psi_n)(k)=(I w_j^*\Psi_n)(k)$.
	Since $I$ is a partial isometry, $(I\Psi_n)_{n\in\NN}$ and 
	$(Iw_j^*\Psi_n)_{n\in\NN}$ converge in $L^2(\mu;L^2(\Geb,\Fock))$
	to $I\Psi$ and $Iw_j^*\Psi$, respectively. Hence, we
	find a subsequence $(\Psi_{n_\ell})_{\ell\in\NN}$ such that
	$(I\Psi_{n_\ell})(k)\to(I\Psi)(k)$ and 
	$(I w_j^*\Psi_{n_\ell})(k)\to(I w_j^*\Psi)(k)$ in $L^2(\Geb,\Fock)$ for $\mu$-a.e. $k$.
	Since $w_j^*$ is closed, this implies \eqref{gottfried12}.
	
	In conjunction with $I\Psi(k)\in L^2(\Geb,\fdom(\NO))$, $\mu$-a.e. $k$, which follows from 
	the second equality in \cref{eq:apandI}, \cref{gottfried12} shows that $I\Psi(k)\in\dom(\fr{q}^{\min})$, $\mu$-a.e. $k$.
	Thus, for $\mu$-a.e. $k$, we can apply \cref{lemgottfried1} with $\Upsilon=I\Psi(k)$. 
	Combining that lemma with \cref{gottfried12,eq:apandI} proves the statement.
\end{proof}
The next statement is now easily proven by a partial integration argument. This is a quadratic
form version of the commutator argument employed in \cite{LiebYamazaki.1958}.
\begin{lem}\label{prop:LY}
	Let $E\ge 1$ and let $\Psi\in\dom(\fr q^{\min})$. Then the iterated integral \cref{def:frw}
	is well-defined. Further, setting $\beta^E_{x}(k)\coloneq v(x,k)/(E+\lambda(k))$, $k\in\mc{K}$,
	and recalling that $\beta^E_{x},\partial_{x_1}\beta^E_{x},\ldots,\partial_{x_d}\beta^E_{x}$ all belong 
	to $\mc{L}^\infty(\Geb,\HP)$ as functions of $x$, we obtain
	\begin{align}\nonumber
		\fr w(v)[\Psi] &= 2\Re \int_\Geb \langle \Psi(x) | a(E\beta^E_{x}) \Psi (x) \rangle_\Fock \Id x 
		\\
		&\quad+  \Re\int_{\Geb} \sum_{j=1}^d \langle \ii (w_{j}^* \Psi) (x) | 
		\vp (\partial_{x_j}\beta^E_{x}) \Psi(x) \rangle_\Fock \Id x. \label{interaction1}
	\end{align}
\end{lem}
\begin{proof}
	Let $k\in\RR^d$ and $g\in\mr W^{1,1}(\Geb)$. Since $v(\cdot,k)\in C^\infty(\Geb)$
	is bounded with bounded first order partial derivatives and $-\Delta_xv(x,k)=2\lambda(k)v(x,k)$, $x\in\Geb$,
	\begin{align}\nonumber
	\int_{\Geb}\ol{v}(x,k)g(x)\Id x&=
	\int_{\Geb}\frac{(2E-\Delta_x)\ol{v}(x,k)}{2E+2\lambda(k)}g(x)\Id x
	\\\label{eq:LY}
	&=\int_{\Geb}E\ol{\beta^E_{x}}(k)g(x)\Id x
	+\frac{1}{2}\sum_{j=1}^d\int_{\Geb}\ol{(\partial_{x_j}\beta^E_{x})}(k)\partial_j g(x)\Id x.
	\end{align}
	According to \cref{lemgottfried2}, we can apply these remarks with
	$g\coloneq\langle\Psi|\ap(k)\Psi\rangle_{\Fock}\in \mr{W}^{1,1}(\Geb)$ for $\mu$-a.e. $k$. This reveals
	that the inner integral in \cref{def:frw} belongs to $L^1(\mu)$ as a function of $k$,
	and \cref{interaction1} follows by combining \cref{eq:pwsmeared,gottfried1,eq:LY}.
\end{proof}
We are now in a position to derive the desired form bound:
\begin{cor}\label{cor:relbound}
	For every $E\ge 1$, the relative form bound \cref{qfbdinteraction} holds true.
\end{cor}

\begin{proof}
	This follows upon combining \cref{interaction1}, \cref{rbNa,rbNvp}, Cauchy--Schwarz inequalities and 
	$\|E^{1/2}\beta^E_{x}\|^2+ \|\nabla_x\beta^E_{x}\|^2/2\le L_E(v)^2$, $x\in\Geb$.
\end{proof}

\section{Strong resolvent continuity w.r.t. vector potentials}\label{sec:rescontA}

\noindent
In our proof of the Feynman--Kac formula for ultraviolet regular coupling functions, 
we approximate the possibly singular magnetic vector potential $A$ in $L_{\loc}^2$ by a sequence of more regular vector potentials $A_n$, $n\in\NN$. We make use of the fact that this entails strong resolvent convergence of the corresponding polaron Hamiltonians. This is well-known for magnetic Schrödinger operators, the strongest results going back to \cite{LiskevichManavi.1997}, and it readily follows
from Feynman--Kac formulas {\em provided} that they are available for locally square-integrable $A$.
Since this is not yet the case in the situation the next theorem is employed, we give a purely functional analytic proof
based on \cite{LiskevichManavi.1997}; for curiosity we keep the assumptions more general and prove a stronger
statement than needed in the main text.

\begin{thm}\label{thm:rescontA}
	Let $A,A_n\in L^2_{\loc}(\Geb,\RR^{d})$, $n\in\NN$, and assume that $A_n\to A$ in $L_{\loc}^2$ as $n\to\infty$, 
	i.e., $1_K A_n\to 1_KA$ in $L^2$-sense for any compact $K\subset \Geb$.
	Let $f$ be either $\vt$ or $v$ and let
	$H_{n}(f)$, $n\in\NN$, denote the polaron Hamiltonians as defined in \cref{sssec:defHam} 
	with $A$ replaced by $A_n$. Abbreviate $H_n\coloneq H_{n}(f)$ and $H\coloneq H(f)$.
	Then $H_{n}$ converges to $H$ in the strong resolvent sense as $n\to\infty$. In fact,
	for some sufficiently large $\lambda>0$,
	\begin{align*}
	\lim_{n\to\infty}(\NO+1)^a(H_n+\lambda)^{-1}(\NO+1)^b\Psi=(\NO+1)^a(H+\lambda)^{-1}(\NO+1)^b\Psi,
	\end{align*}
for all $\Psi\in L^2(\Geb,\dom(\NO^b))$ where $a=b=1/2$. In the case $f=\vt$ we can also choose $a=1$ and $b=0$.
\end{thm}
\begin{proof}
Since we can approximate $v$ by its cutoff versions $v^\sigma\coloneq\tilde{v}_\sigma$, $\sigma\in\NN$, and
the corresponding weighted resolvent convergence \cref{Aunifresconv} is uniform in the vector potentials,
it suffices to treat the case $f=\vt$; note that $(\NO+1)^{1/2}(H(v)+c)^{-1/2}$
is bounded uniformly in the vector potentials when $c>0$ is chosen as in \cref{cor:normresconvUV}.

Putting $\dom(\fr{s}_{\max})\coloneq \fdom(V_+\id_{\Fock})\cap\bigcap_{j=1}^{d}\dom(w_j^*)$, we define a
closed maximal form
\begin{align*}
\fr{s}_{\max}[\Psi]\coloneq\frac{1}{2}\sum_{j=1}^{d}\|w_j^*\Psi\|^2+\int_{\Geb}V(x)\|\Psi(x)\|^2_{\Fock}\Id x,
\quad\Psi\in \dom(\fr{s}_{\max}).
\end{align*}
We let $S$ denote the non-negative selfadjoint operator on $L^2(\Geb,\Fock)$ representing the 
corresponding minimal form 
$\fr{s}_{\min}\coloneq\ol{\fr{s}_{\max}\restr_{\operatorname{span}\{f\phi\mid f\in C_0^\infty(\Geb),\phi\in\Fock\}}}$.
	Define $S_n$, $n\in\NN$, in the same way with $A_n$ replacing $A$.
	By \cite[Theorem~2.8]{LiskevichManavi.1997}, we know that
	\begin{align}\label{eq:strongresS}
		S_n \xrightarrow{\;\;n\to\infty\;\;} S \quad\text{in the strong resolvent sense.}
	\end{align}
	Next, we note that $S+N$ is selfadjoint on $\dom(S)\cap L^2(\Geb,\dom(N))$ and an analogous statement
	holds for  $S_n+N$, $n\in\NN$. The spectra of all these operators have a common lower bound $c_0\le0$ and we
	assume that $\lambda\ge1- c_0$. For $\ell\in\NN_0$, we denote by $P_\ell$
	the orthogonal projection in $\Fock$ onto the $\ell$'th direct sum component $\Fock_\ell$ in \cref{defFock}
	seen as a subspace of $\Fock$. Then 
	we further know that $S_n$, $S$ and $\NO$ commute with $P_\ell$ on their respective domains. 
	Applying $\NO P_\ell = \ell P_\ell$, $\ell\in\NN_0$, and the dominated convergence theorem (with
	dominating sequence $(4\|P_\ell\Psi\|^2)_{\ell\in\NN_0}\in\ell^1(\NN_0)$),  we find, 
	for all $\Psi\in L^2(\Geb,\dom(\NO^b))$,
	\begin{align}\nonumber
		\big\|( \NO+1)^a\big(&(S_n + \NO + \lambda)^{-1}  -(S+\NO+\lambda)^{-1}\big)(\NO+1)^b\Psi \big\|^2 \\&
		= \sum_{\ell=0}^\infty \big\|(\ell+1) P_{\ell}\big((S_n + \ell + \lambda)^{-1}-(S+\ell+\lambda)^{-1}\big)P_\ell\Psi\big\|^2
		\xrightarrow{\;\;n\to\infty\;\;}0,\label{strresconvN}
	\end{align}
	where $a,b\ge0$ satisfy $a+b=1$. 
	Put $D\Psi(x)\coloneq -\vp(\vt_x)\Psi(x)$, a.e. $x\in\Geb$, for every $\Psi\in L^2(\Geb,\fdom(N))$. 
	In view of \cref{rbNvp}, we know that $\Delta\coloneq D(\NO+1)^{-1/2}$ is bounded and, hence, 
	\begin{align*}
	\lim_{\lambda\to\infty}\sup_{n\in\NN}\|D(S_n+\NO+\lambda)^{-1}\|&=0.
	\end{align*}
	Choosing $\lambda$ sufficiently large, we thus obtain the expansions in Neumann series
	\begin{align*}
		(H_{n}+\lambda)^{-1} = \sum_{\ell=0}^\infty(S_n+\NO+\lambda)^{-1}\big(
		\Delta(\NO+1)^{1/2}(S_n+\NO+\lambda)^{-1}\big)^\ell,\quad n\in\NN.
	\end{align*}
	The same expansion holds with $H$ and $S$ put in place of $H_n$ and $S_n$, respectively.
	Moreover, after multiplying with it from the left, we can move $(\NO+1)^a$ with $a\in\{1/2,1\}$
	under the summation signs of the Neumann expansions. 
	After that, we apply both sides of the resulting identity to $(\NO+1)^{1-a}\Psi$ with $\Psi\in\dom(\NO^{1-a})$.
	The statement then follows from the dominated convergence theorem and \cref{strresconvN}; 
	recall also that products of strongly convergent bounded operators are strongly convergent. 
\end{proof}

\section{Differentiability properties of $\beta^\pm$ as $\HP$-valued maps}\label{app:HPdiff}

\noindent
By our assumptions, $v(\cdot,k)\in C^\infty(\Geb)$ when $k\in\mc{K}$ is fixed.
In this appendix we shall study the differentiability of the $\HP$-valued functions
$\Geb\ni x\mapsto\beta_{\sigma,x}^\pm$ with 
$\beta_{\sigma,x}^\pm(k)=\beta_\sigma^\pm(x,k)\coloneq\chi_{\{\lambda\ge\sigma\}}(k)v(x,k)/(\lambda(k)\mp1)$.

First, we note a slight variation of a standard argument which can used to verify the assumptions in the subsequent
\cref{lembetaC2HP}. 

In what follows, $\ball{r}(a)$ denotes the open ball of radius $r>0$ about $a$ in $\RR^d$.

\begin{lem}\label{lem:green1}
Assume that $f:\Geb\times\mc{K}\to\CC$ is measurable and $f(\cdot,k)$ is smooth on $\Geb$ for every $k\in\mc{K}$.
Assume further that $-\Delta_xf(x,k)=2\lambda(k)f(x,k)$ for all $x\in\Geb$ and $k\in\mc{K}$.
Finally, assume that $(1+\lambda)f(x,\cdot)\in\HP$ for all $x\in\Geb$. Then the 
partial derivatives $\partial_{x_j} f(\cdot,k)$ computed for fixed $k\in\mc{K}$ define functions from $\Geb$ to $\HP$ and
\begin{align*}
\sup_{x\in\ol{\ball{\rho}(a)}}\|\partial_{x_j}f(x,\cdot)\|_{\HP}&\le c_{\rho,r,d}\sup_{x\in\ol{\ball{r}(a)}}\|(1+\lambda)f(x,\cdot)\|_{\HP},
\end{align*}
for all $a\in\Geb$ and $r>\rho>0$ such that $\ol{\ball{r}(a)}\subset\Geb$ and with a constant
solely depending on $\rho$, $r$ and $d$.
\end{lem}

\begin{proof}
Let $\ol{\ball{r}(a)}\subset\Geb$. Denote by
$G$ the Green's function of the negative Dirichlet--Laplacian on the ball $\ball{r}(a)$
and by $P$ the Poisson kernel for $\ball{r}(a)$.
Let $\Sigma$ be the surface measure on $\partial\ball{r}(a)$. Then
the solution formula for Dirichlet problems for the Poisson equation on $\ball{r}(a)$ entails, for all $x\in\ball{r}(a)$,
\begin{align*}
\partial_{x_j}f(x,k)&=
2\lambda(k)\int_{\ball{r}(a)} \partial_{x_j}G(x,y)f(y,k)\Id y+\int_{\partial \ball{r}(a)} \partial_{x_j}P(x,y)f(y,k)\Id\Sigma(y).
\end{align*}
Here we also used 
the fact that one partial derivative $\partial_{x_j}$ can be computed under the integral sign in the integral
involving $G$; see, e.g., \cite[Satz~4.7.1 or Satz~4.8.2]{WienholtzKalfKriecherbauer.2009}. The asserted bound now is an easy
consequence of the generalized Minkowski inequality for the norm on $L^2(\mc{K},\mu)$ and the bound
$|G(x,y)|\le c_{d,r}E(x-y)/|x-y|$ where $E$ is given by \cref{EFund}; see, e.g., \cite[Satz~4.6.2]{WienholtzKalfKriecherbauer.2009}.
\end{proof}

\begin{lem}\label{lembetaC2HP}
Let $\ell\in\NN$ and assume in addition to the hypotheses in \cref{sssec:coupling}
that the partial derivative $\partial_x^\alpha v(\cdot,k)$, computed for fixed $k\in\mc{K}$, 
defines a locally bounded function from $\Geb$ to $\HP$
for every multi-index $\alpha\in\NN_0^d$ of length $|\alpha|\le\ell-1$. Let $\sigma\in[2,\infty)$.
Then the maps  $\Geb\ni x\mapsto\beta_{\sigma,x}^\pm$ belong to $C^\ell(\Geb,\HP)$ and, for all
$\alpha\in\NN_0^d$ with $|\alpha|\le\ell$ and $x\in\Geb$, the $\HP$-valued
partial derivative $\partial_x^\alpha\beta_{\sigma,x}^\pm$ is given by the expressions
$\partial_x^\alpha\beta_\sigma^\pm(x,k)$ computed for fixed $k\in\mc{K}$.
\end{lem}

\begin{proof}
We drop the superscript $\pm$ since both $\beta_\sigma^+$ and $\beta_\sigma^-$ can be treated in the same way.
We also drop the subscript $\sigma$. Then $\beta(\cdot,k)$ is smooth on $\Geb$ for every $k\in\mc{K}$.

Let $\ell\in\NN$, $a\in\Geb$ and pick some $r>0$ such that $\ol{\ball{r}(a)}\subset\Geb$.
Further, let $\alpha\in\NN_0^d$ be some multi-index, $j\in\{1,\ldots,d\}$ and $e_j$ be the
$j$'th canonical unit vector in $\RR^d$.
For all $k\in\mc{K}$, $x\in\ol{\ball{r/4}}(a)$ and $t\in[-r/4,r/4]\setminus\{0\}$, we write
\begin{align*}
\frac{1}{t}(\partial_x^\alpha\beta(x+te_j,k)-\partial_x^\alpha\beta(x,k))-\partial_{x_j}\partial_x^\alpha\beta(x,k)&
\\
=\int_0^1\big(\partial_{x_j}\partial_x^\alpha\beta(x+ste_j,k)-\partial_{x_j}\partial_x^\alpha\beta(x,k)\big)\Id s&.
\end{align*}
For $|\alpha|\le\ell-1$, we wish to show that the term in the second line
goes to zero in $\HP$ as a function of $k$, uniformly in $x\in\ol{\ball{r/4}}(a)$.

To this end we again employ the notation $G$, $P$ and $\Sigma$ introduced in the proof of \cref{lem:green1}.
Then we can apply the formula displayed in that proof with $f=\partial_x^\alpha\beta$. 
Thus, by Minkowski's inequality,
\begin{align*}
\bigg(\int_{\mc{K}}\bigg|\int_0^1\big(\partial_{x_j}\partial_x^\alpha\beta(x+ste_j,k)
-\partial_{x_j}\partial_x^\alpha\beta(x,k)\big)\Id s\bigg|^2\Id\mu(k)\bigg)^{\frac{1}{2}}
\le \mc{I}_G(x,t)+\mc{I}_P(x,t)&,
\end{align*}
for all $x\in\ol{\ball{r/4}}(a)$ and $0<|t|\le r/4$, with $\mc{I}_G(x,t),\mc{I}_P(x,t)\ge0$ given by
\begin{align*}
&\mc{I}_G(x,t)^2
\\
&\coloneq
\int_{\mc{K}}\bigg|
2\lambda(k)\int_0^t\int_{\ball{r}(a)} \big(\partial_{x_j}G(x+ste_j,y)-\partial_{x_j}G(x,y)\big)
\partial_y^\alpha\beta(y,k)\Id y\,\Id s\bigg|^2\Id\mu(k),
\end{align*}
and
\begin{align*}
&\mc{I}_P(x,t)^2
\\
&\coloneq\int_{\mc{K}}\bigg|\int_0^t\int_{\partial\ball{r}(a)} \big(\partial_{x_j}P(x+ste_j,y)-\partial_{x_j}P(x,y)\big)
\partial_y^\alpha\beta(y,k)\Id \Sigma(y)\,\Id s\bigg|^2\Id\mu(k).
\end{align*}
With $E$ given in \cref{EFund}, we set $h(x,y)\coloneq G(x,y)-E(x-y)$, 
\begin{align*}
\delta_f(t)&\coloneq
\sup_{x\in\ol{\ball{r/4}(a)}}\int_0^1\int_{\ball{r}(a)}\big|\partial_{x_j}f(x+ste_j,y)-\partial_{x_j}f(x,y)\big|\Id y\,\Id s,
\quad f\in\{G,E,h\}.
\end{align*}
Then $\delta_h(t)\to0$ as $t\to0$ since $\partial_{x_j}h$ is continuous on $\ball{r}(a)\times\ol{\ball{r}(a)}$;
see, e.g., \cite[Lemma~4.5.4]{WienholtzKalfKriecherbauer.2009}.
Moreover, it is straightforward to show that $\delta_E(t)\to0$ as $t\to0$. Thus, $\delta_G(t)\to0$ as $t\to0$.
Further, the generalized Minkowski inequality implies
\begin{align*}
&\mc{I}_G(x,t)
\\
&\le \int_0^1\int_{\ball{r}(a)}\bigg(\int_{\mc{K}}
\big|\partial_{x_j}G(x+ste_j,y)-\partial_{x_j}G(x,y)\big|^2|2\lambda(k)\partial_y^\alpha\beta(y,k)|^2
\Id\mu(k)\bigg)^{\frac{1}{2}}\Id y\,\Id s,
\end{align*}
for all $x\in\ol{\ball{r/4}}(a)$ and $0<|t|\le r/4$, so that
\begin{align*}
\sup_{x\in\ol{\ball{r/4}(a)}}\mc{I}_G(x,t)
&\le C\delta_G(t)\sup_{y\in\ol{\ball{r}(a)}}\|\partial_y^\alpha v(y,\cdot)\|_{\HP}\xrightarrow{\;\;t\to0\;\;}0.
\end{align*}
Here $C$ is a universal constant and we used that $\lambda-1\ge(1+\lambda)/4$ on $\{\lambda\ge\sigma\}$. 
Likewise, since $\partial_{x_j}P$ is uniformly continuous on $\ol{\ball{r/2}(a)}\times\partial\ball{r}(a)$,
\begin{align*}
\delta_P(t)&\coloneq
\sup_{x\in\ol{\ball{r/4}(a)}}\int_0^1\int_{\partial\ball{r}(a)}\big|\partial_{x_j}P(x+ste_j,y)-\partial_{x_j}P(x,y)\big|\Id\Sigma( y)\,\Id s
\xrightarrow{\;\;t\to0\;\;}0.
\end{align*}
Similarly as above the generalized Minkowski inequality entails
\begin{align*}
\sup_{x\in\ol{\ball{r/4}(a)}}\mc{I}_P(x,t)&\le \delta_P(t)\sup_{y\in\partial\ball{r}(a)}\|\partial_y^\alpha\beta(y,\cdot)\|_{\HP}
\xrightarrow{\;\;t\to0\;\;}0.
\end{align*}
We conclude by applying these remarks recursively with $|\alpha|=0$, $|\alpha|=1$ and so on up to $|\alpha|=\ell-1$.
\end{proof}

\begin{lem}\label{lembetaC1HP}
Let $\sigma\in[2,\infty)$.
Then the maps  $\Geb\ni x\mapsto\beta_{\sigma,x}^\pm$ belong to $C^1(\Geb,\HP)$ and, for all
$j\in\{1,\ldots,d\}$ and $x\in\Geb$, the $\HP$-valued partial derivative $\partial_{x_j}\beta_{\sigma,x}^\pm$ 
is given by the expressions $\partial_{x_j}\beta_\sigma^\pm(x,k)$ computed for fixed $k\in\mc{K}$.
\end{lem}

\begin{proof}
For every $n\in\NN$, we set $v^n_x\coloneq\chi_{\{\lambda<n\}}v_x$, $x\in\Geb$.
Then $v^n$ satisfy the hypotheses in \cref{lembetaC2HP} with $\ell=1$, whence the functions
$\Geb\ni x\mapsto\chi_{\{\lambda<n\}}\beta_{\sigma,x}^\pm$ belong to $C^1(\Geb,\HP)$ and their first order partial
derivatives can be computed holding $k$ fixed. On the other hand, by (b) in \cref{sssec:coupling},
we have the uniform convergences $\sup_{x\in\Geb}\|\chi_{\{\lambda\ge n\}}\beta_{\sigma,x}^\pm\|_{\HP}\to0$ and
$\sup_{x\in\Geb}\|\chi_{\{\lambda\ge n\}}\nabla_x\beta_{\sigma}^\pm(x,\cdot)\|_{\HP^d}\to0$ as $n\to\infty$.
\end{proof}


\section{On Feynman's expression for the complex action}\label{app:action}

\noindent
As promised in \cref{rem:Feynmansaction}, we verify in the next example that our formula for the complex action 
agrees with Feynman's famous expression \cref{uFeyn} when multi-polarons on $\RR^3$ are considered.
In \cref{exuconfpolaron} we consider confined multi-polarons and find a direct analogue of \cref{uFeyn}.
As mentioned earlier, our definition \cref{def:actioninfty} allows for a treatment of general coupling functions $v$
and is helpful in derivations of exponential moment bounds and suitable convergence theorems.
\begin{example}\label{exuFeyn}
In the situation of \cref{exFpolaron}, where $d=3\nu$, let $\sigma\in[2,\infty)$ and $x\in\RR^{3\nu}$. Then
we $\PP$-a.s. find
\begin{align}\label{uFeyn}
u_{\sigma,t}(x)&=2g^2\sum_{j,\ell=1}^\nu \int_0^t\int_0^s\frac{\eul^{-(s-r)}}{4\pi|x_j+b_{j,r}-x_\ell-b_{\ell,s}|}\Id r\,\Id s,
\quad t\ge0,
\end{align}
where $b_t=(b_{1,t},\ldots,b_{\nu,t})$ with independent three-dimensional $(\fr{F}_t)_{t\ge0}$-Brownian motions
$(b_{j,t})_{t\ge0}$. Similarly $x=(x_1,\ldots,x_{\nu})$ with $x_j\in\RR^3$.
The right hand side of \cref{uFeyn} is Feynman's famous and well studied complex action \cite{Feynman.1955}.
\end{example}
\begin{proof}
In fact, for every $n\in\NN$ and recalling the notation from \cref{def:cutoffv}, a short computation in polar coordinates shows that
\begin{align}\nonumber
	u_{\reg,t}(\tilde{v}_{n};x)
	&=\sum_{j,\ell=1}^\nu \int_0^t\int_0^s\eul^{-(s-r)}\int_{|k|^2<2n}\eul^{\ii k\cdot(x_\ell+b_{\ell,s}
	-x_j-b_{j,r})}\frac{2g^2}{(2\pi)^3|k|^2}\Id k\,\Id r\,\Id s
	\\\label{uvsw}
	&=\frac{g^2}{\pi^2}\sum_{j,\ell=1}^\nu \int_0^t\int_0^s\eul^{-(s-r)}E_{n}(x_j+b_{j,r}-x_\ell-b_{\ell,s})\Id r\,\Id s,
\end{align}
where $E_{n}(0)\coloneq\sqrt{2n}$ and we encounter a well-known family of Riemann integrals:
\begin{align}\label{uFeyn2}
	E_{n}(y)&\coloneq \frac{1}{|y|}\int_0^{\sqrt{2n}|y|}\frac{\sin(\rho)}{\rho}\Id \rho
	\xrightarrow{\;\;n\to\infty\;\;}\frac{\pi}{2|y|},\quad y\in\RR^3\setminus\{0\}.
\end{align}
With the help of Tonelli's theorem it is easily seen that the expectation of the right hand side of \cref{uFeyn}
is finite, whence, $\PP$-a.s., the double integral on the right hand side of \cref{uFeyn} is finite for all $t\ge0$.
On the complement of some $\PP$-zero set, we can thus use dominated convergence and \cref{uFeyn2} to argue that
the expression in the second line of \cref{uvsw} converges to the right hand side of \cref{uFeyn} for all $t\ge0$.
On the other hand, \cref{lemuregu,lem:convuLp} imply the existence of integers
$1\le n_1<n_2<\ldots$ such that, $\PP$-a.s., $u_{\reg,t}(\tilde{v}_{n_j};x)\to u_{\sigma,t}(x)$ as $j\to\infty$
for all $t\ge0$.
\end{proof}
\begin{example}\label{exuconfpolaron}
In the situation of \cref{exFS}, choose $m=3$ and $\theta(t)=t^{-1}$ for $t\ge\lambda(1)$.
Let $G_{\mc{G}}$ denote the Green's function of $\mc{G}$, i.e., the integral kernel of $(-\Delta_{\mc G})^{-1}$. 
Pick $\sigma\in[2,\infty)$ and $x\in\mc{G}^{\nu}$.
Then we $\PP$-a.s. find, for all $t\ge0$, that
\begin{align}\label{uFeynG}
u_{\sigma,t}(x)&=2g^2
\sum_{j,\ell=1}^\nu \int_0^t\int_0^s\eul^{-(s-r)}G_{\mc{G}}(x_j+b_{j,r},x_\ell+b_{\ell,s})\Id r\,\Id s,
\end{align}
on $\{t<\tau_{\mc{G}^\nu}(x)\}$, where we use the same notation for $x$ and $b$ as in \cref{exuFeyn}.
\end{example}
\begin{proof}
	Similar to the previous proof, we find
	\begin{align*}
		u_{\reg,t}(\tilde{v}_{n};x)
		&=2g^2\sum_{j,\ell=1}^\nu \int_0^t\int_0^s\eul^{-(s-r)} G_n (x_j+b_{j,r},x_{\ell}+b_{\ell,s})\Id r\,\Id s,
	\end{align*}
	for all $n\in\NN$, where
	\begin{align*}
		G_n(y,z) \coloneq\sum_{\substack{k\in\NN:\\\lambda(k)< n}}
		\frac{\phi_k(y)\ol{\phi_k(z)}}{2\lambda(k)},\quad y,z\in\mc{G}.
	\end{align*}
	The Weyl law for the eigenvalues $2\lambda(k)$ of $-\Delta_{\mc G}$ shows that $(1/\lambda(k))_{k\in\NN}$
	is a sequence in $\ell^2(\NN)$. Since $(y,z)\mapsto\phi_k(y)\ol{\phi_m(z)}$ form an orthonormal basis of $L^2(\mc{G}^2)$, 
	we find that $G_n$ converges to $G_{\mc G}$ in $L^2(\mc{G}^2)$. Employing this observation as well as the independence
	of $b_j$ and $b_\ell$ for $j\not=\ell$ and the fact that every Brownian motion has independent increments, 
	it is straightforward to show
	by direct estimations that $\chi_{\{t<\tau_{\mc{G}^\nu}(x)\}}u_{\reg,t}(\tilde{v}_{n};x)$ converges in $L^1(\PP)$ to the right hand side 
	of \cref{uFeynG} multiplied by $\chi_{\{t<\tau_{\mc{G}^\nu}(x)\}}$. We can now conclude as in the proof of \cref{exuFeyn}.
\end{proof}

\subsection*{Acknowledgements} BH acknowledges support by the Ministry of Culture and Science of the State of North Rhine-Westphalia within the project `PhoQC'.
OM is grateful for support during the early phase of this project by the Independent Research Fund Denmark via the 
project grant ``Mathematical Aspects of Ultraviolet Renormalization'' (8021-00242B).

\bibliographystyle{halpha-abbrv}
\bibliography{../../Literature/00lit}
\end{document}